\documentclass[11pt]{article}
\usepackage[utf8]{inputenc}
\usepackage{amsmath}
\usepackage{amsthm}
\usepackage{amssymb}
\usepackage{algorithm}
\usepackage{subfig}
\usepackage{color}
\usepackage[english]{babel}
\usepackage{graphicx}
\usepackage{grffile}
\usepackage{wrapfig,epsfig}
\usepackage{epstopdf}
\usepackage{url}
\usepackage{color}
\usepackage{epstopdf}
\usepackage{algpseudocode}
\usepackage[T1]{fontenc}
\usepackage{bbm}
\usepackage{dsfont}
\usepackage{thmtools}
\usepackage{thm-restate}
\usepackage{bm}
\usepackage{tcolorbox}
\usepackage{enumerate}

\usepackage{hyperref}  %%% arxiv don't allow this.
\hypersetup{colorlinks=true,citecolor=blue,linkcolor=blue}
\usepackage[margin=1in]{geometry}

\graphicspath{{./figs/}}
\usepackage{mathtools}

\newtheorem{theorem}{Theorem}[section]

\newtheorem{lemma}[theorem]{Lemma}
\newtheorem{definition}[theorem]{Definition}

\newtheorem{corollary}[theorem]{Corollary}

\newtheorem{fact}[theorem]{Fact}
\newtheorem{remark}[theorem]{Remark}
\newtheorem{claim}[theorem]{Claim}

\newcommand{\wh}{\widehat}
\newcommand{\wt}{\widetilde}

\newcommand{\eps}{\epsilon}

\newcommand{\R}{\mathbb{R}}

\renewcommand{\varepsilon}{\epsilon}
\renewcommand{\tilde}{\wt}
\renewcommand{\hat}{\wh}
\renewcommand{\bar}{\overline}
\renewcommand{\eps}{\epsilon}

\newcommand{\bx}{\mathbf{x}}
\newcommand{\ba}{\mathbf{a}}

\newcommand{\br}{\mathbf{r}}

\newcommand{\bd}{\mathbf{d}}

\def\calM{\mathcal{M}}

\def\cc{\mathbf{c}} \def\bpi{\boldsymbol{\pi}} \def\calF{\mathcal{F}} \def\calG{\mathcal{G}}\def\calA{\mathcal{A}}

\DeclareMathOperator*{\E}{{\mathbb{E}}}

\DeclareMathOperator{\OPT}{OPT}

\DeclareMathOperator{\poly}{poly}
\DeclareMathOperator{\polylog}{polylog}

\DeclareMathOperator{\rank}{rank}

\makeatletter
\newcommand*{\RN}[1]{\expandafter\@slowromancap\romannumeral #1@}
\makeatother

\title{On the Complexity of Dynamic Submodular Maximization}
\author{Xi Chen\\ Columbia University\\ \texttt{xichen@cs.columbia.edu} 
\and Binghui Peng\\ Columbia University\\ \texttt{bp2601}@columbia.edu}
\date{\today}

\begin{document}

\maketitle

\begin{abstract}
We study dynamic algorithms for the problem of maximizing a monotone submodular function over a stream of $n$ insertions and deletions. We show that any algorithm that maintains a $(0.5+\eps)$-approximate solution under a cardinality constraint, for any constant $\eps>0$, must have an amortized query complexity that is {\em polynomial} in $n$. Moreover, a linear amortized query complexity is needed in order to maintain a $0.584$-approximate solution. This is in sharp contrast with recent dynamic algorithms of \cite{lattanzi2020fully,monemizadeh2020dynamic} that achieve $(0.5-\eps)$-approximation with a $\polylog(n)$ amortized query complexity.

On the positive side, when the stream is insertion-only, we present efficient algorithms for the problem under a cardinality constraint and under a matroid constraint with approximation guarantee $1-1/e-\eps$ and amortized query complexities $\smash{O(\log (k/\eps)/\eps^2)}$ and $\smash{k^{\tilde{O}(1/\eps^2)}\log n}$, respectively, where $k$ denotes the cardinality parameter or the rank of the matroid.
\end{abstract}\thispagestyle{empty}\newpage

% !TEX root = main.tex

\section{Introduction}
\label{sec:intro}

Initiated by the classical work of \cite{nemhauser1978analysis,nemhauser1978best} in the 1970s, submodular maximization has~developed into a central topic of  discrete 
optimization during past decades~\cite{calinescu2011maximizing, feldman2011unified, vondrak2013symmetry,chekuri2014submodular,buchbinder2015tight} (see \cite{buchbinder2018submodular} for a survey).
Capturing the natural notion of diminishing returns, submodular functions and their optimization
problems have found numerous applications in areas such as
machine learning~\cite{golovin2011adaptive,wei2015submodularity}, data mining~\cite{lin2011class}, algorithmic game theory~\cite{roughgarden2010algorithmic}, social networks~\cite{kempe2003maximizing}, etc.
The canonical form of the problem is to maximize a monotone submodular function 
under a cardinality constraint $k$, for which the greedy algorithm of \cite{nemhauser1978analysis}
achieves the optimal approximation ratio of $1-1/e$ \cite{nemhauser1978best, Feige98}.
However, despite the simplicity and optimality of this celebrated algorithm, 
there has been a surge of recent research effort to reexamine the problem
under a variety of computational models motivated by unique challenges of working with massive datasets.
These %new application domains motivated a surge 
%of ongoing research effort for submodular maximization and its variants, including 
include streaming algorithms~\cite{badanidiyuru2014streaming, feldman2020one}, parallel algorithms~\cite{balkanski2018adaptive, balkanski2019exponential, ene2019submodular, fahrbach2019submodular, balkanski2019optimal, chekuri2019parallelizing,ene2019submodular1}, 
learning algorithms~\cite{balcan2011learning, balkanski2017limitations} and distributed algorithms~\cite{barbosa2015power, mirrokni2015randomized,barbosa2016new}.

\vspace{+1mm}
\noindent\textbf{Dynamic submodular maximization.}
We study the power and limitations of \emph{dynamic} algori\-thms {for submodular maximization}.
In this model, the set of elements that one can choose from is subject to changes
incurred by a stream of $n$ insertions and deletions.
Letting $V_t$ denote the current set of elements after the first $t$ operations,
an algorithm needs to maintain a subset $S_t\subseteq V_t$ of size at most $k$
that achieves a certain approximation guarantee for every round $t$, and its
performance is measured by its amortized query complexity (i.e.,
the average number of queries it makes per operation on
the underlying unknown monotone submodular function); see the formal definition of the model in Section \ref{sec:pre}.
%The study of dynamic algorithms for submodular maximization is motivated by 
%  and the current a dynamic stream of insertion and deletions.

The dynamic model is motivated by real-world applications of submodular maximization 
over massive
datasets that evolve frequently. % (such as the ever-changing structure of the Internet and social~networks). 
Many of these applications arise in machine learning and data mining tasks,
including the data 
subset selection problem~\cite{wei2015submodularity, golovin2011adaptive},
movie recommendation system~\cite{balkanski2018approximation},
influence maximization in social networks~\cite{kempe2003maximizing}, etc.
As an example,~in influence maximization, one is given a social network as well as a stochastic diffusion model,
and the goal is to select a seed set of size at most $k$ to maximize the influence spread over the network.
The spread function is submodular for many well-studied models and therefore, the problem becomes
submodular maximization with a cardinality constraint.
Given that social networks, like Twitter or Facebook, are involving continuously over time, 
an old seed set could become outdated quickly. 
A natural rescue is to have an efficient dynamic algorithm that maintains a seed set over time.

% where objective functions can be casted as a submodular function and the problem as submodular maximization.

%In data subset selection problem, one carefully chooses a subset of the data for downstream task so as to minimize any significant loss in performance.
%This arises in many machine learning tasks and data mining task, where the origin datasets have prohibitively large size and are redudant, and downstream computation task could be computationally demanding when running one the whole dataset.
%The prevalent approach is to treat the task as a submodular maximization problem, by carefully choosing appropriate submodular functions as the surrogate model to measure the utility of each subset for an underlying task.
%The data are gradually collected over time and one can not afford to solve the problem over and over again. 
%A natural question is whether we can achieve similar utility guarantee without large overhead on updating time.

%{\bf \noindent Application:  \ \ } 
%In movie recommendation, the objective is to recommend a small, diverse, and high quality set of movies, based on users' movie ratings. 
%The objective function can be casted as a coverage function.
%As it is clear that the movie set is growing over time, it is important to dynamically maintain a good recommendation sets without sacrificing too much on computation resources.

%{\bf \noindent Application:  \ \ }

%,
%making it impractical 
The problem of dynamic submodular maximization under a cardinality constraint was studied~in two recent papers \cite{monemizadeh2020dynamic, lattanzi2020fully}, giving
algorithms that achieve an $(1/2-\eps)$-approximation
guarantee with amortized query complexities that are $O(k^2\log^2 n\cdot \eps^{-3})$ and $O(\log^8 n \cdot \eps^{-6})$ in
a stream of length $n$. Compared with the $1-1/e$ approximation guarantee of the offline greedy algorithm,  however,  
a natural open question posted by \cite{lattanzi2020fully} is whether 
the $1/2$ can be further improved under the dynamic setting,
or even under the setting when the stream is insertion-only.

\subsection{Our results}

We resolve the open question of \cite{lattanzi2020fully} by showing that any algorithm
with an approximation ratio of $1/2+\eps$ 
must have amortized query complexity polynomial in the stream length $n$:

\begin{restatable}{theorem}{thmlowertwo} \label{thm:lower2}
	For any constant $\eps>0$, there is a constant $C_\eps> 0$ with the following property.~When $k\ge C_\eps$, any randomized algorithm that achieves an approximation ratio of $1/2+\eps$ for dynamic submodular~maxi\-mization
	under cardinality constraint $k$ requires amortized query complexity $\smash{n^{\tilde{\Omega}(\eps)}/ k^3}$.	
\end{restatable}

Moreover, we show that any algorithm with approximation ratio $0.584$ must have 
an amortized query complexity that is linear in $n$:

\begin{restatable}{theorem}{thmlowerone}
	\label{thm:lower1}
	%	For any $\eps > 0$, suppose $k \geq \Omega(\eps^{-3}\exp(\eps^{-14})\log n)$.
	%	In a dynamic stream with both insertion and deletion updates, an algorithm for maximizing monotone sumodular function under cardinality constrains with an approximation guarantee of $0.584 + \eps$ must have amortized number of query of at least $\tilde{\Omega}_{\eps}(n/k^3)$ 
	There is a constant $C>0$ with the following property. 
	%Let $n$ be the length of the stream.
	When $k\ge C \log n$,~any~randomized algorithm for dynamic submodular maximization 
	under cardinality constraint $k$~that obtains an approximation guarantee of $0.584$
	must have amortized query complexity at least $\Omega(n/k^3)$. 
	%suppose $k \geq \Omega(\eps^{-3}\exp(\eps^{-14})\log n)$.
	%	In a dynamic stream with both insertion and deletion updates, an algorithm for maximizing monotone sumodular function under cardinality constrains with an approximation guarantee of $0.584 + \eps$ must have amortized number of query of at least $\tilde{\Omega}_{\eps}(n/k^3)$ 
\end{restatable}

%  construction is based on
%  a tree structure with depth $1/\eps$. %, which can be viewed as a novel 

In the proof of both theorems we construct a family of hard functions to hide
a secret matching. % (bijections) in our hard functions
%  behind proofs of both theorems. 
The main challenge is to achieve the following two properties at the same time:
The first one, which we will refer to as the \emph{large-gap} property,
states that any dynamic algorithm with a certain approximation guarantee
can only succeed by recovering the secret matching hidden in the function.
On the other hand, the second so-called \emph{indistinguishability} property
shows that each query made by an algorithm can only reveal very little information
about the secret matching. 
The construction for Theorem \ref{thm:lower1} uses a simple bipartite structure
but requires a more detailed analysis to optimize for the constant $0.584$.
The construction for Theorem \ref{thm:lower2}, on the other hand, 
is based on a more sophisticated $(1/\eps)$-level tree structure
which, at a high level,
can be viewed as a~novel tree extension of the path-alike construction from 
\cite{feldman2020one}.
We discuss in more details about our lower bound proofs in Section \ref{sec:techoverview}.

On the positive side, when the stream is insertion only,
we obtain an algorithm that achieves an approximation guarantee
of $1-1/e-\eps$ with amortized query complexity $O(\log (k/\eps)/\eps^2)$:
%, independent of the stream length: 
%  the approximation guarantee can be improved to match
%  the ratio of $1-1/e$ 

\begin{restatable}{theorem}{thmupperone}
	\label{thm:insert-cardinality}
	Given any $\eps>0$,
	there is a deterministic algorithm that achieves an approximation guarantee of $ 1-1/e -\eps $
	for dynamic submodular maximization under cardinality constraint $k$ over insertion-only streams.
	The amortized query complexity of the algorithm is $O(\log (k/\eps)/\eps^{2})$.
	%that maintains a feasible set $S$ with $(1-1/e -\eps)$-approximation at each iteration. Moreover, the amortized number of queries per update is $O(\log k/\eps^2)$.
\end{restatable}

We also obtain a randomized algorithm that achieves an approximation guarantee of $1-1/e-\eps$
for the problem under a matroid constraint, with amortized query complexity $\smash{k^{\tilde{O}(1/\eps^2)}\log n}$.

\begin{restatable}{theorem}{thmuppertwo}
	\label{thm:insert-matroid}
	Given any $\eps > 0$, there is a randomized algorithm that achieves
	an approximation guarantee of $1-1/e-\eps$
	for dynamic submodular maximization under matroid constraints over 
	insertion-only streams.
	% Moreover, the 
	The amortized query complexity of the algorithm is $\smash{k^{\tilde{O}(1/\eps^2)}\log n}$,
	where~$k$~is the rank of the matroid.
\end{restatable}

Our algorithms are inspired by the classic greedy solution and we turn it into a dynamic algorithm via techniques involving lazy update, multilinear extension and accelerated continuous greedy, together with the introduction of a $O(1/\eps)$-pass prune-greedy algorithm that could be of independent interest. We elaborate about the intuition behind our algorithms in Section~\ref{sec:techoverview}.

%{\color{red}add some sentences about our upper bounds.}
\begin{remark}
	The best offline algorithm for maximizing monotone submodular function with cardinality constraints requires $O(n\log(1/\eps))$ queries \cite{buchbinder2017comparing}. Our dynamic algorithm incurs only poly-logarithmic overheads. While for matroid constraints, even the state-of-art offline algorithm requires $O(n\sqrt{k})$ value queries \cite{buchbinder2017comparing}.
	Hence, for matroid constraints, our focus is to design dynamic algorithms with $\poly(k)$ amortized query complexity.
\end{remark}

We review related work and then give a technical overview of
our results in Section \ref{sec:techoverview}.

%\Binghui{To replace with formal statement.}
%We consider dynamic submodular maixmization problem, and our main results are summarized below.
%\begin{itemize}%
%	\item For insertion only update, we present a $(1-1/e  -\eps)$-approximation algorithm for cardinality constraints, with amortized query complexity $O(\log k/\eps^2)$. See Section~\ref{sec:insert-cardinality} for details.
%	\item For insertion only update, we present a combinatorial algorithm that maintains a $(1/2 -\eps)$-approximation solution with amortized query complexity $k^{\tilde{O}(1/\eps)}$. Combining with the acceleration technique, we derive a $(1-1/e -\eps)$-approximation algorithm with amortized query complexity $k^{\tilde{O}(1/\eps^{2})}$. See Section~\ref{sec:insert-matroid} for details.
%	\item When both deletion and insertion are allowed, we give a lower bound showing that no algorithm with $\Omega(n)$ amortized query complexity can achieve $0.584$-approximation. See Section~\ref{sec:lower1} for details.
%	\item When both deletion and insertion are allowed, we can boost the above lower bound instance and prove that no algorithm with $\Omega(n^{\eps})$ amortized query complexity can achieve $(1/2 + \eps)$ approximation. See Section~\ref{sec:lower2} for details.
%\end{itemize}

\subsection{Related work}

%Initiated by the work of \cite{nemhauser1978analysis,nemhauser1978best} in 1970's, submodular maximization has been a central topic for discrete optimization over past decades~\cite{calinescu2011maximizing, feldman2011unified, vondrak2013symmetry,chekuri2014submodular,buchbinder2015tight}.
%Motivating by various applications domain including machine learning, data mining, algorithmic game theory and social networks, there has been a surge of ongoing reserach effort for submodular maximization and its variants, including streaming algorithms~\cite{badanidiyuru2014streaming, feldman2020one}, parallel/adaptive algorithms~\cite{balkanski2018adaptive, balkanski2019exponential, balkanski2019optimal, chekuri2019parallelizing}, learning algorithms~\cite{balcan2011learning, balkanski2017limitations}, dynamic algorithms~\cite{monemizadeh2020dynamic, lattanzi2020fully}. We refer interested reader to the survey of~\cite{buchbinder2018submodular} for a general coverage of the area.

Dynamic submodular maximization has only been studied recently, and the work of \cite{monemizadeh2020dynamic, lattanzi2020fully} are most relavant to us. The concurrent work of \cite{monemizadeh2020dynamic} and \cite{lattanzi2020fully} give $(1/2-\eps)$-approximation algorithm to the dynamic submodular maximization problem under cardinality constraints, with amortized query complexity $O(k^2 \log^2 n \cdot \eps^{-3})$ and $O(\log^8 n\cdot \eps^{-6})$ respectively.
%Our lower bounds validates the optimality on these algorithms and the algorithm for insertion only stream separates the model.

Our work is also closely related to the streaming setting. For cardinality constraints, \cite{badanidiyuru2014streaming} give an $(1/2 -\eps)$-approximation algorithm in the streaming model, The approximiation is tight, Feldman et al. \cite{feldman2020one} prove $\Omega(\eps n/k^3)$ space is necessary for achieving $(1/2+\eps)$ approximation. For matroid contraints, \cite{chekuri2015streaming} give an $1/4$-approximation algorithm and the approximation ratio is improved to 0.3178 by \cite{feldman2021streaming}.
We remark the algorithm of \cite{badanidiyuru2014streaming} and \cite{chekuri2015streaming} can be implemented in the insertion-only dynamic setting with the same approximation guarantee, and the amortized runing query complexity are $O(\eps^{-1}\log (k/\eps))$ and $O(k)$ respectively. 
%We give an improved approximation guarantee.

On lower bound side, our work make use of the symmetric gap techniques of \cite{mirrokni2008tight, vondrak2013symmetry} and the weight scheduling of \cite{feldman2020one}.

\subsection{Technical overview}\label{sec:techoverview}

\def\calF{\mathcal{F}}
%{\color{red}
%A submodular function ...
%Cardinality constraint $k$ ...
%Unweighted coverage functions ...}
We provide a streamlined technique overview of our approach.\vspace{0.2cm}

{\noindent\bf A linear lower bound for $0.584$-approximation. } 
Our lower bound is based on the construction of 
a family of monotone submodular functions %$\calF:V\rightarrow
%[0,1]$ that 
with the following properties.  
They share the same ground set $V$ which is 
partitioned into $2m$ 
sets $V = A_1 \cup\cdots \cup A_{m}\cup B_1 \cup\cdots \cup B_m$
and the algorithm knows both $V$ and the partition. % at the beginning. 
Each $A_i$ or $B_j$ contains $O(k)$ elements to be specified later. 
What is hidden inside the function is a secret bijection $\pi:[m]\rightarrow [m]$,
which is unknown to the algorithm, and we write the function as 
$\calF_\pi:2^V\rightarrow [0,1]$.
The three main properties we need about $\calF$ are:
\begin{flushleft}\begin{enumerate}
		\item[i)]  For any $j\in [m]$, there is an $S\subseteq A_{\pi(j)}\cup B_j$ of size $k$
		that achieves the optimal $\calF_\pi(S)=1$;
		\item[ii)] \textbf{Large gap:} For any $j\in [m]$, every $S\subseteq A_1\cup\cdots\cup A_m\cup B_j$ of size at most $k$ has $\calF_\pi(S)\le \kappa$
		for some (as small as possible) parameter $\kappa>0$, unless $S\cap A_{\pi(j)}\ne \emptyset;$ and 
		\item[iii)] \textbf{Indistinguishability:} Every query made by the algorithm reveals very little information about 
		$\pi$. (Looking ahead, we show that each query on $\calF_\pi$ is roughly 
		equivalent to no more than $O(k)$ queries on $\pi$, each of the limited form as ``whether 
		$\pi(a)$ is equal $b$ or not.'')
	\end{enumerate}\end{flushleft}
	With these properties in hand, 
	we use the following simple dynamic stream of length $\Theta(mk)$ in our lower bound proof:
	%We assume $m = \Omega(n)$.
	%In the dynamic stream, 
	$A_1,\ldots,A_m$ are inserted first and then
	we insert $B_1$, delete $B_1$, insert $B_2$, delete $B_2,\ldots$, until $B_m$ is 
	inserted and deleted. % $B_1, \ldots, B_m$ is inserted and deleted one after another,
	%There is a hidden matching $\pi: [m] \rightarrow [m]$ between the set $\{B_i\}_{i \in [m]}$ and $\{A_j\}_{j \in [m]}$.  
	The large-gap property (ii)  
	ensures that any algorithm with an approximation guarantee of $\kappa$
	must recover each entry $\pi(j)$ of the hidden $\pi$ after each 
	set $B_j$ is inserted.\footnote{Note that we are a bit sloppy here given that
		the dynamic algorithm only needs to output a set $S_j$ after $B_j$ is inserted
		to overlap with $A_{\pi(j)}$. But given that $|S_j|\le k$, this can be considered
		almost as good as knowing $\pi(j)$.}
	The indistinguishability property (iii), on the other hand, shows that~$\Omega(m^2/k)$
	queries are needed to recover the hidden bijection $\pi$ (by a standard calculation) 
	and thus, the amortized query complexity is at least 
	$\Omega(m/k^2)=\Omega(n/k^3)$ by plugging in the stream length $n=\Theta(mk)$.
	The main challenge is to find a construction that satisfies (i) and (iii) and 
	at the same time satisfies (ii) with a parameter $\kappa$ 
	that is as small as possible (optimized to be
	$0.584$ at the end).

	%there is a large gap between $f(A_{\pi(i)}, B_{i})$ and $f(A_j, B_i)$ ($j \neq \pi(i)$), and the major insight behind our lower bound construction is that,
	%if a dynamic algorithm can achieves $0.584$-approximation, it must recover the hidden matching $\pi$. Meanwhile, we prove the hidden matching is hard to solve and requires $\Omega(m^2)$ queries even in the offline setting. 
	%These two together imply a linear lower bound.
	
	We focus on minimizing $\kappa$ in the large-gap property (ii)  in this part, given that achieving the indistinguishability (iii)
	will become more challenging  in the polynomial lower bound for  
	$(0.5+\eps)$-approximation later and will be what we focus on there.
	%Consider $\pi$ to be the function with $\pi(j)=j$.
	Our first attempt is the following simple coverage function $\calF$ (though our final
	construction will not be a coverage function).
	Let $U$ be a large universe set (for intuition consider $U$ to be arbitrarily 
	large so that every expectation about random sets we draw holds exactly).
	For each $j\in [m]$ we randomly partition $U$ into two parts $U_{j,1}$ and $U_{j,2}$ of the same size:
	the union of elements of $A_{\pi(j)}$ (which correspond to subsets of $U$) will be $U_{j,1}$ and the union 
	of elements of $B_j$ will be $U_{j,2}$.
	Each $B_j$ is further partitioned into $w$ groups $B_{j,1},\ldots,B_{j,w}$, where $w$ is a large enough constant,
	such that each $B_{j,\ell}$ has size $k/2$ and corresponds to an even $(k/2)$-way random partition of $U_{i,2}$.
	The same happens with $A_{\pi(j)}$.
	Now imagine that the algorithm does not know $\pi(j)$ and returns 
	$S\subset A_i\cup B_j$ with $i\ne \pi(j)$ by picking $k/2$ elements randomly from $A_i$ and $k/2$ elements 
	randomly from $B_j$.
	Simple calculation shows that elements in $S$ would together cover
	$0.5(1-1/e)+0.25(1-1/e^2)\approx 0.53$, which is very close to our goal of $0.5$, and splitting 
	elements not evenly between $A_i$ and $B_j$ can only cover less.
	
	However, the algorithm can deviate by (1) not picking elements randomly from $A_i$ and $B_j$~and furthermore,
	(2) picking from not just $A_i\cup B_j$ but also the much bigger set of $A_1\cup\cdots\cup A_m$.~For (1), if the algorithm manages to find out all $k/2$ elements
	of group $A_{i,\ell}$ and all $k/2$ elements of group $B_{j,\ell'}$ for some $\ell,\ell'$,
	then together they would cover $3/4$.
	For (2), the algorithm can pick one random element each from $k$ different $A_i$'s 
	to cover $1-1/e$.
	In our final construction,~we~overcame~(1)~by applying a symmetric function constructed from
	\cite{mirrokni2008tight} on top of the coverage function above. It guarantees that 
	no algorithm (with a small number of queries) can ever pick a set of elements from any of $A_i$ or $B_j$
	that is non-negligibly ``\emph{unbalanced},'' by having noticeably more elements from any group $A_{i,\ell}$
	or $B_{j,\ell}$ than the average among $A_i$ or $B_j$.
	Accordingly, the large-gap property (ii) needed can be updated to be the following weaker version:
	\begin{enumerate}
		\item[ii)] \textbf{Large gap, updated:} For any $j\in [m]$, every ``balanced'' $S\subseteq A_1\cup\cdots\cup A_m\cup B_j$ of size at most $k$ has $\calF_\pi(S)\le \kappa$
		for some parameter $\kappa>0$, unless $S\cap A_{\pi(j)}\ne \emptyset;$ 
	\end{enumerate}
	% that applies only to
	%  sets $S$ in which elements are drawn randomly from each set $A_i$ or $B_j$ that 
	%  the algorithm decides to choose from; this will be captured by the notion of \emph{balance} in the proof. 
	To address the challenge of (2), we use two parameters to introduce a different type of \emph{unbalancedness}:
	First, $U$ is no longer partitioned $50$-$50$; instead $U_{j,1}$ is $\beta$-fraction and 
	$U_{j,2}$ is $(1-\beta)$-fraction of $U$.
	Second, it is no longer the case that all groups $A_{i,\ell}$ and $B_{j,\ell}$ have
	size $k/2$; instead every $A_{i,\ell}$ has size $\alpha k$ and every $B_{j,\ell}$ has size $(1-\alpha)k$.
	With a detailed analysis we show that $\kappa$ can be set to $0.584$ in the updated
	version of the large-gap property (ii), when $\alpha=0.56$ and $\beta=0.42$.\vspace{0.3cm}

	{\noindent\bf A polynomial lower bound for $(1/2+\eps)$-approximation.}
	Hard instances for our linear lower bound sketched earlier reduce at the end to a simple problem
	of recovering a hidden matching. 
	To obtain a polynomial lower bound for $(1/2+\eps)$-approximation, we need to 
	further extend the idea to amplify the (hardness of) approximation ratio by creating a depth-$L$ tree structure
	with $L = 1/\eps$ to hide multiple bijections at every level.
	Let $m_1, \ldots, m_{L}$ be a sequence of positive integers: $m_L=1$ but the other integers will be specified at the end.
	The ground set $V$ is then defined using a depth-$L$ tree $T$, in which each internal node at depth $\ell$ ($\ell \in [0: L-1]$) has $m_{\ell+1}$ children. 
	Let $U_{\ell}$ denote the set of nodes $u = (u_1, \ldots, u_{\ell})\in [m_1]\times \cdots \times [m_\ell]$ at depth $\ell$;
	its children are given by~$(u,1),\ldots,$ $(u,m_{\ell+1})$. The ground set is defined as $V= \cup_{u\in U_1 \cup \cdots \cup U_{L}}A_u$, where 
	each $A_u$ contains $\eps k$ elements.
	
	In the dynamic update stream, we perform a limited DFS walk on the tree,
	meaning that after reaching a node $u$, one only explores the first $d:=n^\eps$ children $(u,1),\ldots,(u,d)$ of 
	$u$ (except when $u$ is at depth $L-1$ in which case we just explore the only child of $u$).
	Let $U_L^*:= [d]^{L-1}\times \{1\}$ be the set of leaves visited by the DFS walk.
	We design the stream in a way such that whenever a leaf $u$ is reached,
	the current set of elements is  given by
	$$
	W_u:=\bigcup_{\ell\in [L ]} \bigcup_{i\in [m_{\ell }]} A_{u_1,\ldots,u_{\ell-1},i},
	$$ 
	i.e., the union of sets of nodes that are children of any node along the path from the root to $u$.
	Let $\calA_u=A_{u_1}\cup \cdots\cup A_{u_1,\ldots,u_{L-1}}$.
	The first two properties of our $\calF:2^V\rightarrow [0,1]$ can now be stated:
	\begin{enumerate}
		\item[i)] Let $S=\calA_u\cup A_u$. We have $|S|=k$ and $\calF(S)$ achieves the optimal  $\calF (S)=1$; 
		\item[ii)] \textbf{Large gap:} For any $S\subseteq W_u$ of size at most $k$, 
		$\calF (S)< 1/2+\eps$
		unless $S\cap \calA_{u}\ne \emptyset.$  
	\end{enumerate}
	%At every time step, we guarantee current elements in hand lies on a path and its siblings of the tree $T$.
	%We ensure that if the algorithm's solution set does not contain any element along the path (beside the leaf node), then it can achieve at most $(1/2+\eps)$-approximation. 
	%The major intuition behind our lower bound is that in order to achieve $(1/2+\eps)$-approximation, the algorithm needs to figure out at least some parts of the path, and this reduces to solve  some hidden matching problems.
	%We already have an $\Omega(m^2)$ lower bound for hidden matching, the difference is that there are many hidden matching problems along the dynamic stream, it is possible for dynamic algorithm to solve some of them, but we prove it is impossible to solve all of them.
	%The arguments follows from a counting based argument and the lower bound for single instance of hidden matching.
	At a high level, the construction of our monotone submodular function $\calF$   
	can be viewed as a~novel tree extension of a path-alike construction from 
	\cite{feldman2020one} (for one-way communication complexity of submodular maximization).
	In particular, we adopted a weight sequence $\{w_\ell\}$ from \cite{feldman2020one} 
	to play the role of parameter $\beta$ in our linear lower bound~discussed earlier,
	in order to achieve the large-gap property (ii) above no matter how the $S$ is spread 
	among different levels of the tree.

	%We wrap up the proof by showing a large gap between the optimal solution (which is along the path) and the one containing in the siblings of the path. For intuitions, if we enforce the algorithm to select $\eps k$ element in each level, then we can set $w_1 = \cdots = w_L = 1/\eps$. One can prove (via concavity) the best thing to do is to select an entire internal node, which gives only $(1/2+\eps)$-approximation.
	%A complication arises due to the fact that the algorithm can assign different budgets to different levels.
	%To this end, we adapt the weight sequence of ~\cite{feldman2020one} and guarantees that the best an algorithm can achieve is at most $1/2 + \eps\log^2(1/\eps)$. The intuition is to assign few weights to lower level and larger weights to higher level.
	
	Now $\calF$ is only the base function that we use to construct the family of hard functions.~To~this end, we introduce the notion of a \emph{shuffling} $\pi$ of tree $T$,
	which consists of one bijection $\pi_u$ at every internal node 
	$u$ of $T$ to shuffle its children.
	Let $\calF_\pi$ be the function obtained after applying a hidden shuffling $\pi$ on $\calF$.
	Under the same stream, at the time when  a leaf $u\in U_L^*$ arrives, the current set of elements remains to be $W_u$ but
	the $\calA_u$ we care about becomes
	$$
	\calA_u^\pi:=A_{\pi_\gamma(u_1)} \cup A_{u_1,\pi_{u_1}(u_2)}\cup \cdots \cup A_{u_1,\ldots,u_{L-2},\pi_{u_1,\ldots,u_{L-2}}}
	(u_{L-1}),
	$$ 
	where $\pi_\gamma$ denotes the hidden bijection at the root, $\pi_{u_1}$ denotes the hidden bijection
	at the depth-$1$ node $u_1$, so on and so forth.
	The two properties (i) and (ii) hold after replacing $\calA_u$ with $\calA_u^\pi$.
	So finding $S\subseteq W_u$ of size at most $k$ with $\calF(S)\ge 1/2+\eps$ 
	requires the algorithm to (essentially) identify one of the edges along the path
	after shuffling, i.e., one of the entries $\pi_\eps(u_1),\pi_{u_1}(u_2),\ldots,\pi_{u_1,\ldots,u_{L-2}}$ $(u_{L-1})$
	in the hidden bijections.
	The challenge left is to establish the indistinguishability property:
	\begin{flushleft}\begin{enumerate}
			\item[iii)] \textbf{Indistinguishability:} Every query made by the algorithm reveals very little information about 
			$\pi$. (More formally, we show that each query on $\calF_\pi$ is roughly 
			equivalent to no more than $O(k^2)$ queries on $\pi$, each of the limited form as ``whether 
			$\pi_u(a)$ for some $u$ is equal $b$.''
		\end{enumerate}\end{flushleft}
		With (iii) if an algorithm tries to win 
		by finding out $\pi_{u_1,\ldots,u_{L-2}}(u_{L-1})$ every time a leaf $u$ is~reached then each hidden entry
		requires $\Omega(m_{L-1}/k^2)$ queries and thus,
		$\Omega(m_{L-1}d^{L-1}/k^2)$ queries in total.
		In general if an algorithm keeps working on $\pi_{u_1,\ldots,u_{\ell}}(u_{\ell+1})$,
		the total number of queries needed is   $\Omega(m_{\ell+1}d^{\ell+1}/k^2)$. 
		Setting $m_\ell=n^{(L-\ell+1)\eps}/(2k)$ for each $\ell$ leads to a lower bound
		of $\Omega(n^{1+\eps}/k^3)$ for the total query complexity in all cases, while maintaining 
		the stream length to be $n$.
		(It is also not difficult to show that mixing effort among different levels does not help.)
		%A calculation shows that the length of the stream is $n$, leading to the
		%  polynomial amortized lower bound.
		
		Finally let's discuss how the indistinguishability property (iii) is implemented.
		%We first consider the hidden matching problem.
		%Imagine $f$ to be a coverage function defined over an universe $U$. Each $A_{j}$ (collaboratively) covers $\beta$-fraction of $U$ and $B_{i}$ (collaboratively) covers $(1-\beta)$-fraction of U$.$
		%In order to prove a query lower bound on the hidden matching problem, we intend to show each query reveals little information.
		%There is carried out via the following two ideas.
		The first trick is to cap  the function and have the final function take the form as
		$$\calF_\pi^*(S)=\min\big\{ \calF_\pi (S)+\eps |S|/k, 1\big\}.$$
		%First, we cap the value of the coverage function. In particular, our final function is defined as $F(S) = \{f(S) + \frac{\eps}{k}, 1\}$. 
		This simple modification changes the value of solutions $S$ we care about slightly, no more than $\eps$,
		given that $|S|\le k$, and at the same time guarantees no algorithm makes any
		query of size more than $k/\eps$; otherwise the value returned is trivially $1$ and reveals no information.
		%solution and it increases the value of any set (with size less than $k$) by at most $\eps$. At the same time, it guarantees that the query size is at most $k/\eps$, otherwise the value returned is always equal to $1$ and reveals no information.
		This implies that each query $S$ made by an algorithm can only involve no more than $k/\eps$ 
		nodes in the tree; let $U_S$ denote the set of such nodes.
		A final key observation is that the function $\calF_\pi$ on $S$ is uniquely determined 
		by the structure of the tree formed by paths of nodes in $U_S$ to the root after applying the shuffling $\pi$.
		The latter is uniquely determined by the depth of the lowest common ancestor of 
		every pair of nodes in $U_S$ after shuffling $\pi$, which in turn can be obtained by making at
		most one query to $\pi$ of the form described in (iii) for each pair in $U_S$.
		The bound $O(k^2)$ follows given that $|U_S|\le k/\eps$.

		\vspace{0.3cm}
		
		{\noindent\bf Insertion-only streams: Submodular maximization with a cardinality constraint.}
		Our starting point is the classic (offline) greedy approach, where in each round an element with the~maximum marginal contribution is added to the solution set.
		We adapt the greedy algorithm to the dynamic setting.
		For simplicity, let's assume in the overview 
		that a value $\OPT$ is given to the algorithm at the beginning and its goal is to
		maintain a set that achieves a $(1-1/e-\eps)$-approximation
		when the actual optimal value reaches $\OPT$. % in advance.
		While it is impossible to find the element with the maximum margin before having all of them in hands, we can instead choose one whose margin gives a moderate  improvement, i.e., add any element that satisfies $f_{S}(e) \geq ({\OPT - f(S)})/{k}$, where $S$ is the current solution set.
		This suffices to achieve the optimal $(1-1/e)$-approximation.
		A naive implementation, however, requires $O(k)$ amortized query complexity since the threshold $ ({\OPT - f(S)})/{k}$ is updated every time a new element is added to solution set, and the algorithm needs to scan over all elements in the worst case.
		We circumvent this with the idea of {\em lazy update}. 
		We divide the mariginal contribution into $O(1/\eps)$ buckets, with the $i$-th buckets containing marignal smaller than $i\cdot ({\eps \OPT}/{k})$.
		We don't update the margin value every time the solution set is augmented. 
		Instead, each time a new element is inserted, the algorithm only checks the bucket whose marginal is larger than the current threshold.
		One can show either a new element is added to $S$, or it is pushed down to the next level. The latter can happen at most $O(1/\eps)$ times, and after that, the element has negligible marginal contribution that can be ignored safely.\vspace{0.3cm}

		{\noindent\bf Insertion-only streams: Submodular maximization with a matroid constraint.}
		We first provide a deterministic combinatorial algorithm that achieves a $(1/2-\eps)$-approximation.
		We build upon the previous idea, but with significant adaptations.
		Again, our goal is to simulate an offline greedy algorithm but this time, one cannot relax the condition and hope to augment the solution set whenever the margin has a moderate improvement.
		The analysis of the offline greedy algorithm relies crucially on picking the maximum margin element in each step.
		Our first idea is to find the element with approximate maximum margin and branch over all possibilities.
		In particular, we divide the marginal gain into $L = O(\eps^{-1}\log (k/\eps))$ many levels, where the $\ell$-th level corresponds to $[ (1+\eps)^{-\ell+1}\OPT,  (1+\eps)^{-\ell}\OPT]$.
		When working on an insertion-only stream, we proceed to the $(\ell+1)$-th level only when there is no element in the $\ell$-th level anymore.
		Of course, we still don't know when the $\ell$-th level becomes empty so that we can move to $\ell+1$, but we can enumerate over all possibilities and guarantee that there exists one branch that fits the sequence of the offline approximate greedy algorithm.
		The caveat is that the total number of branch is $\binom{k}{L}$, which is quasi-polynomial in $k$.
		We further reduce this number by considering a pruned version of the offline approximate greedy algorithm, where the algorithm prunes extra element in each level and only keeps the majority.
		We prove the algorithm still guarantees a $(1/2-\eps)$-approximation and the total number of branch reduces to  $\smash{k^{\tilde{O}(1/\eps)}}$.
		It requires careful analysis to make the idea work, but the very high level intuition is that one needs to be more careful in the first few buckets (since the marginal is large) and less careful at the end. 
		To amplify the approximation ratio to $(1-1/e-\eps)$, we make use of the multi-linear extension and use the accelerated continuous greedy framework \cite{badanidiyuru2014fast}. In each iteration, we use the above combinatorial algorithm to find the direction of improvement of the multi-linear extension. 
		We remark a similar amplifying procedure has been used in the previous work on adaptive submodular maximization \cite{chekuri2019parallelizing, balkanski2019optimal}.\vspace{0.3cm}
		
		{\noindent\textbf{Organization.}} We begin by presenting the linear lower bound for 0.584-approximation under fully dynamic stream in Section~\ref{sec:lower1}. Section~\ref{sec:lower2} is devoted to prove the polynomial lower bound on $(1/2+\eps)$-approximation.
		We then turn to the insertion-only stream and provide our algorithms and analysis.
		The algorithm for a cardinality constraint is presented in Section~\ref{sec:insert-cardinality}, and we provide an efficient $(1-1/e)$-approximation algorithm for a matroid constraint in Section~\ref{sec:insert-matroid}.
		We discuss future research directions in Section~\ref{sec:discussion}.

\section{Preliminary}
\label{sec:pre}

\noindent\textbf{Submodular functions.}
Let $V$ be a finite ground set. 
A function $f:2^V\rightarrow \mathbb{R}$ is \emph{submodular} if 
$$
f(S\cup T)+f(S\cap T)\le f(S)+f(T)
$$
for all pairs of sets $S,T\subseteq V$.
We say $f$ is \emph{nonnegative} if $f(S)\ge 0$ for all $S\subseteq V$ and 
$f$ is \emph{monotone} if $f(S)\le f(T)$ whenever $S\subseteq T\subseteq V$.
We study submodular functions that are both nonnegative and monotone,
and assume without loss of generality that $f(\emptyset)=0$ (see Remark \ref{remark2}).
We say an algorithm has query access to $f$ if it can adaptively 
pick subsets $S\subseteq V$ to reveal the value of $f(S)$.
%In the rest of the paper, we will always assume $f(\emptyset)=0$ and refer
%  to them as monotone submodular functions for convenience.

%Submodularity can be equivalently defined using the notion of \emph{marginal gain}.

Given $f:2^V\rightarrow \mathbb{R}$ and $S,T\subseteq V$, the
marginal gain of adding $e$ to $S$ is defined as
$$
f_{S}(T):= f(S\cup T)-f(S).
$$
When $T=\{e\}$ is a singleton, we write $f_{S}(e)$ to denote
$f_{S}(\{e\})$ for convenience. 
Monotonicity and submodularity can be defined equivalently using marginal gains:
$f$ is monotone iff $f_{S}(e)\ge 0$ for all $e\in V$ and $S\subseteq V$;
$f$ is submodular iff $f_{S}(e)\ge f_{T}(e)$ for any $e\notin T$ and $S\subseteq T\subseteq V$.\medskip

%We are interested in understanding the query complexity of \emph{dynamic} algorithms
%  for maximizing a nonnegative, monotone submodular function under either a cardinality
%  or a matroid constraint.
%Such an algorithm is given the ground set $V$ and has query access to an unknown
%  $f:V\rightarrow \mathbb{R}$:
%For each round the algorithm can pick a set $S\subseteq V$ and make a query to get $f(S)$.\medskip
%We first define these problems in the offline setting and 
%  then describe the setting of \emph{dynamic streams}.\medskip

%We consider a collection $V$ of items, and a monotone submodular function $f: 2^{V}\rightarrow \R^{+}$ defined on the ground set $V$. We assume $f(\emptyset) = 0$\f
%ootnote{Xi: Is this usually assumed?} follow the convention of literature.
%Given two sets $X, Y \subseteq V$, the marginal gain of $X$ with respect to $Y$ is defined as
%\[
%f(X\mid Y) = f(X\cup Y) - f(Y)
%\]
%The function is monotone, if for any element $e\in V$ and any set $Y \subseteq V$, it holds that $f(e\mid Y) \geq 0$.
%We say the function is submodular, if for any two sets $X, Y$ satisfy $X\subseteq Y \subseteq V$ and any element $e \in V\backslash Y$, we have
%\[
%f(e\mid X) \geq f(e\mid Y)
%\]

\noindent\textbf{Dynamic submodular maximization under a cardinality constraint. \ \ }
Given $V$, a positive integer $k$ and query access 
to a nonnegative, monotone submodular function $f:V\rightarrow \mathbb{R}$,
the goal is to find a $\gamma$-approximate solution $S_i$ of size at most $k$
at the end of each~round $i$ when making a pass on a dynamic stream of insertions and deletions
(the stream is \emph{insertion-only} if no deletions are allowed).
%We consider algorithms that maintain a $\gamma$-approximate solution
%  under a \emph{dynamic} stream in which elements in $V$ are inserted and deleted round by round;
%  a stream is said to be \emph{insertion-only} when there are no deletions. 
More formally, starting with $V_0=\emptyset$, an element $e_i\in V$ is either inserted or deleted
at the beginning of round $i=1,\ldots,$ so that the current ground set
$V_{i}$ is set to be either $V_{i-1}\cup \{e_i\}$ 
if $e_i$ is inserted
or $V_i=V_{i-1}\setminus \{e_i\}$ if $e_i$ is deleted.
After this, the algorithm makes queries to $f$ 
to find a \emph{$\gamma$-approximate solution $S_i$ of $f$ with respect to $V_i$}.
This means that $|S_i|\le k$ and 
$$
f(S_i)\ge \gamma \OPT_i,\quad\text{where}\quad\OPT_i:= \max_{T\subseteq V_i,|T|\le k} f(T)
$$
%when  maximization under a cardinality constraint is concerned, or $S_i\in \calM$
%and
%$$
%f(S_i)\ge \gamma\cdot \max_{T\subseteq V_i, T\in \calM}f(T)
%$$ 
%when maximization under a matroid constraint is concerned.
We emphasize that an algorithm remembers every query it has made so far. Thus results
of queries
made in previous rounds may help finding $S_i$ in the current round. 

We will consider algorithms that are both deterministic and randomized.
We say a deterministic dynamic algorithm achieves an approximation guarantee of $\gamma$
if given  $n$ and any stream
of length $n$, it returns a $\gamma$-approximate solution $S_i$ of the $i$-th round for every $i\in [n]$.
We say a randomized dynamic algorithm achieves an approximation guarantee of $\gamma$ if
given $n$ and any stream of length $n$, with probability at least $2/3$ it returns
a $\gamma$-approximate solution $S_i$ for every round $i\in [n]$ at the same time.
We say an (deterministic or randomized) algorithm has \emph{amortized} query complexity $Q $
if the total number of queries it makes is no more than $n\cdot Q $.

\begin{remark}\label{detailremark}
	We discuss some details of the model behind our upper and lower bounds:
	\begin{flushleft}\begin{enumerate}
			\item 
			Our lower bounds hold even if the algorithm is given $V$ and  it is allowed to query during the $i$-th round
			any set $S$ of elements
			that have appeared in $V_j$ for some $j\le i$; our algorithms only query sets in $V_i$ and
			does not need to know $V$ initially.
			
			%This only makes our lower bounds stronger, while our algorithms 
			%  only query sets in $V_i$ during the $i$-th round. 
			
			\item Our lower bounds hold even if the algorithm is given $n$, the stream length;
			our algorithms do not need to know $n$ and meet the stated amortized bounds at the end of every round.
		\end{enumerate} \end{flushleft}
	\end{remark}
	
	%The goal is to~find a set $S\subseteq V$ with $|S|\le k$ to 
	%  maximize $f(S)$ among all sets of size at most $k$.
	%We say $S$ is a $\gamma$-approximate solution, for some $\gamma\in (0,1)$, if
	%  $|S|\le k$ and 
	%$$
	%f(S)\ge \gamma\cdot \max_{T\subseteq V, |T|\le k} f(T).
	%$$
	%\[
	%\max f(S), S\subseteq V, |S| \leq k
	%\]
	
	%\paragraph{Matroids.} 
	
	\noindent\textbf{Dynamic submodular maximization under a matroid constraint \ \ }
	A set system $\mathcal{M} \subseteq 2^{V}$ is a \emph{matroid} if it satisfies (i) $\emptyset\in \mathcal{M}$, (ii) the {\em downward closed} and (iii) {\em augmentation} {properties}. A set system~$\mathcal{M}$ is downward closed if $T \in \mathcal{M}$ implies $S\in \mathcal{M}$ for all $S\subseteq T$. The augmentation property is that~if $S, T \in\mathcal{M}$ and $|S| < |T|$, then there must be an element $e\in T\backslash S$ such that $S \cup \{a\} \in \mathcal{M}$. When $S\in \mathcal{M}$, we say $S$ is {\em feasible} or {\em independent}. The rank of the matroid $\mathcal{M}$, denoted as $\rank(\mathcal{M})$, is the maximum size of an independent set in $\mathcal{M}$.%\footnote{The algorithm does not need to be given the rank, right? It can compute the rank by itself.}
	
	The setting of dynamic submodular maximization under a matroid constraint is similar.
	In addition to $V$ and query access to $f$,
	the algorithm is given query access to a matroid $\calM$ over $V$:
	it can pick any  $S\subseteq V$ to query if $S\in \calM$ or not. 
	To goal is to find a  $\gamma$-approximate solution $S_i $ of $f$ with respect to $V_i$
	at the end of every round. This mdeans that %, for some $\gamma\in (0,1)$,
	$S_i\subseteq V_i$, $S_i\in \calM$ and
	%The problem of maximizing a monotone submodular fucntion $f$ under a matroid constraint $\mathcal{M}$ is to find an $S\in \mathcal{M}$ that maximizes $f(S)$.
	\[
	f(S_i)\ge \gamma \cdot\OPT_i,\quad\text{where}\quad \OPT_i:= \max_{T\in \calM,T\subseteq V_i} f(T).
	\]
	When measuring the amortized query complexity,
	we count queries to both $f$ and $\calM$.
	
	\begin{remark}\label{remark2}
		In both problems, one can assume without loss of generality
		that $f(\emptyset)=0$ (since any $\gamma$-approximate solution to $g$ with respect to $V_i$, where  $g(S):=f(S)-f(\emptyset)$, must be a $\gamma$-approximate solution to $f$
		with respect to $V_i$ as well).
		We will make this assumption in the rest of the paper.
	\end{remark}

	%\noindent\textbf{Dynamic algorithms for submodular maximization.}

	%We consider algorithms that are both deterministic and randomized.
	%We say a deterministic dynamic algorithm achieves an approximation guarantee of $\gamma$
	%  if given any $n$ and any stream
	%  of length $n$, it returns a $\gamma$-approximate solution $S_i$ of the $i$-th round for every $i\in %[n]$.
	%We say a randomized dynamic algorithm achieves an approximation guarantee of $\gamma$ if
	%  given any $n$ and any stream of length $n$, with probability at least $2/3$ it returns
	%  a $\gamma$-approximate solution $S_i$ of the $i$-th round for every $i\in [n]$.
	%We will measure the performance g  
	%
	%consider a dynamic stream and elements in $V$ are inserted and deleted over time. In an insertion-only stream, element can only be inserted, while in an fully dynamic stream, element is allowed to be inserted first and removed later.
	%We use $V_i$ to denote the set of all elements that has been inserted but not deleted until the $i$-th operation, and use $U_i$ to denote the set of elements that has been inserted. Let $O_i$ be the optimum solution up to $i$-th iterations, and denote $\OPT_i = f(O_i)$.
	%\Binghui{TODO: define amortized query?}
	
	\noindent\textbf{The multilinear extention} The multilinear extension $F: [0,1]^{|V|}$ $\rightarrow \R^{+}$ of a function $f$ maps a point $\bx \in [0,1]^{|V|}$ to the expected value of a random set $S\sim \bx$, i.e.
	\[
	F(\bx) = \sum_{S\subseteq V} \prod_{e\in S}x_{e} \prod_{e'\in V\backslash S}(1 - x_{e'}) f(S)
	\]
	We write $F(\bx) = \E_{S\sim \bx}[f(\bx)]$ for simplicity.  For any $\bx \in [0,1]^{|V|}$, $\lambda \in [0,1]$, $S\subseteq V$, we write $F(\bx + \lambda S)$ to denote $F(\bx')$, where $x'_i = x_i$ if $i \notin S$ and $x'_i = \min\{x_i + \lambda, 1\}$ if $i \in S$.
	
	For a continuous function $F:[0,1]^{|V|} \rightarrow \R$, we say it is monotone if $\frac{\partial F}{\partial x_i} \geq 0$ and it is submodular if $\frac{\partial^2 F}{\partial x_i \partial x_j} \leq 0$ for every $i, j$.
	When $f$ is monotone and submodular, the multilinear extension $F$ is also monotone and submodular.

% !TEX root =  main.tex

\section{A linear lower bound for $0.584$-approximation}
\label{sec:lower1}

%When the update stream contains both insertion and deletion operations, we provide a linear lower bound showing that it is impossible to achieve $(0.584 + \eps)$-approximation unless the amortized query complexity is $\Omega_{k,\eps}(n)$.
We restate the main theorem of this section:

\thmlowerone*

%\begin{theorem}%
%	\label{thm:lower1}
%	For any $\eps > 0$, suppose $k \geq \Omega(\eps^{-3}\exp(\eps^{-14})\log n)$.
%	In a dynamic stream with both insertion and deletion updates, an algorithm for maximizing monotone sumodular function under cardinality constrains with an approximation guarantee of $0.584 + \eps$ must have amortized number of query of at least $\tilde{\Omega}_{\eps}(n/k^3)$ 
%For any constant $\eps > 0$ there is a constant $C_\eps>0$ with the following property. 
%Let $n$ be the length of the stream.
%When $k\ge C_\eps \log n$, any randomized algorithm for dynamic submodular maximization 
%  under a cardinality constraint of $k$ with an approximation guarantee of $0.584+\eps$
%  must have amortized query complexity at least $\Omega(n/k^3)$. 
%suppose $k \geq \Omega(\eps^{-3}\exp(\eps^{-14})\log n)$.
%	In a dynamic stream with both insertion and deletion updates, an algorithm for maximizing monotone sumodular function under cardinality constrains with an approximation guarantee of $0.584 + \eps$ must have amortized number of query of at least $\tilde{\Omega}_{\eps}(n/k^3)$ 
%\end{theorem}

\subsection{Construction of the symmetric function} 
\label{sec:base}

%\footnote{Xi: Is the symmetric function meant to be $\hat{F}$?}\footnote{\Binghui{Yes}}
Let $w$ be a positive integer and let $\eps>0$ be a small constant to be fixed later. 
Given $x\in [0, 1]^{w}$, we let $\bar{x} =  \sum_{i=1}^{w}x_{i}/w$. Consider the following
  function $f$ over $[0,1]^w$ and its symmetric version $g$: 
\begin{align*}
f(x): = 1 - \prod_{i\in [w]}(1 - x_i)\quad\text{and}\quad
g(x): = 1 - (1 - \bar{x})^w.
\end{align*}
The following theorem is from \cite{mirrokni2008tight}. We need to make some minor changes on the choice of parameters and include a proof in Appendix~\ref{sec:lower1-app} for completeness.
\begin{theorem}[\cite{mirrokni2008tight}]
	\label{thm:base}
Given any positive integer $w$ and $\eps > 0$, let $\gamma=\gamma(w,\eps): = w^{-1}\exp(-4w^6/\eps)$. There is a
	monotone submodular function
	%\footnote{We should include a definition of submodular functions 
	%over continuous domains in preliminaries.} 
$\hat{f}: [0,1]^{w}\rightarrow [0,1]$ with the 
	following two properties:
	\begin{enumerate}
		\item Whenever $\max_{i, j\in [w]}|x_i - x_j| \leq \gamma$, we have $\hat{f}(x) = g(x)$;\vspace{-0.1cm}
		\item For any $x \in [0,1]^{w}$, we have $f(x)-\eps\le \hat{f}(x) \le f(x)$.
	\end{enumerate}
\end{theorem}
Let $m$ be a positive integer.
We use $\hat{f}$ from Theorem \ref{thm:base} to construct a function
  $\hat{F}:[0,1]^{mw}\rightarrow \mathbb{R}$ as follows:
Writing $x\in [0,1]^{mw}$ as $x=(x_1,\ldots,x_s)$ and $x_i=(x_{i,1},\ldots,x_{i,w})\in [0,1]^w$, 
  we have 
% and $\bx_1, \ldots, \bx_{s} \in \R^{w}$, let $\bx = (\bx_1,\cdots, \bx_{s})$, we use the aboved constructed $\hat{f}$ as a basis to construct the function $\hat{F}: [0,1]^{mw} \rightarrow \R^{+}$ as
\begin{align}
\label{eq:base1}
\hat{F}(x):= 1 - \prod_{i \in [m]}(1 - \hat{f}(x_i)).
\end{align}
We also define its symmetric version $G: [0,1]^{mw}\rightarrow \R^{+}$ as
\begin{align*}
G(x):=1 - \prod_{i \in [m]}(1 - g(x_i)) = 1 - \prod_{i \in [m]} (1 - \overline{x_i})^{w}.
\end{align*}

Here are some basic properties we need about $\hat{F}$, we defer detailed proof to Appendix~\ref{sec:lower1-app}.
\begin{lemma}
	\label{lem:base2}
	%
	%Given positive integers  $w \in \mathbb{Z}$, $\eps > 0$, let $\gamma = w^{-1}\exp(-4w^6/\eps)$.
	%For any $x\in [0,1]^{mw}$, we have
$\hat{F}$ satisfies the following properties:
	\begin{itemize}
		\item $\hat{F}$ is monotone, submodular and satisfies $\hat{F}(x) \in [0,1]$ for all 
		  $x\in [0,1]^{mw}$.
		\item 
		   $\hat{F}(x) = G(x)$ when $x\in [0,1]^{mw}$ satisfies 
		   $$\max_{i\in [m], j, j'\in [w]}|x_{i, j} - x_{i, j'}| \leq \gamma$$ and in this case, $\hat{F}(x)$ depends on 
		  $\overline{x_1}, \ldots, \overline{x_m}$ only.
		\item $\hat{F}(x) \geq 1- \eps$ when 
		  $x$ satisfies $x_{i, j} = 1$ for some $i\in [m]$ and $j\in [w]$, 
		  
	\end{itemize}
\end{lemma}

\def\calF{\mathcal{F}}

\subsection{Construction of the hard functions $\calF_{c,\pi}$}
\label{sec:lower1-contruction}

We now present the construction of the family of hard functions 
   that will be used in the proof.\medskip
   
\noindent\textbf{Choice of parameters.} 
Let $n$ be the length of the dynamic stream.
Let $\eps>0$ be a constant.
Let $\alpha$ and $\beta$ be two constants in $(0,1)$ that we fix at the end of the proof.\footnote{Looking ahead, we will choose $\alpha$ and $\beta$ to minimize the quantity 
  $Q(\alpha,\beta)$ discussed in Lemma \ref{hardlemma}; we will set them to be $\alpha = 0.56$ and $\beta =0.42$,
  respectively.}
%Let $\alpha, \beta \in (0,1), \eps > 0$, 
Let
$$
w = 10\left(\frac{1}{\alpha\eps} + \frac{1}{(1 -\alpha)\eps}\right)\quad\text{and}\quad
\gamma = w^{-1}\exp\left(-\frac{4w^6}{\eps}\right)
$$ 
be two constants.
Let $k$ be the cardinality constraint parameter that is at most $n^{1/3}$ (otherwise the lower bound 
  $\Omega(n/k^3)$ in the main theorem becomes trivial) and satisfies
$$
k\ge \frac{10}{\eps\gamma\alpha^2(1-\alpha)^2}\cdot \log n.
$$
Finally, let $m$ be such that $n=(2-\alpha)mkw$ so we have $m=\Omega(n^{2/3})$.\medskip

%\Omega(\eps^{-1}\gamma^{-2}\alpha^{-2}(1 - \alpha)^{-2}\log n)$, $n = (2 - \alpha)mkw \leq 2mkw$.
 
\noindent\textbf{The ground set $V$.}
We start with the definition of the ground set $V$, where
$$
V=A\cup B,\quad A= \bigcup_{i=1}^{m} A_i\quad\text{and}\quad B=\bigcup_{i=1}^{ m} B_i,
$$  
where each $A_i$ has $\alpha k w$ elements and each $B_j$ has 
  $(1-\alpha)kw$ elements and they are pairwise disjoint.\medskip
  
\noindent\textbf{The function $\calF_{c,\pi}$.}
Let $c:V\rightarrow [w]$.
We say $c$ is a \emph{proper} $w$-coloring of $V$ if
  each $A_{i,j}$, the set of elements in $A_i$ with color $j\in [w]$, has size $\alpha k$
  and each $B_{i,j}$, the set of elements in $B_i$ with color $j$, has size $(1-\alpha)k$.
(So $A_{i,1},\ldots,A_{i,w}$ form an even partition of $A_i$ and 
  $B_{i,1},\ldots,B_{i,w}$ form an even partition of $B_i$.)
Let $\pi:[m]\rightarrow [m]$ be a bijection which we will view as 
  matching $A_{\pi(i)}$ with $B_{i}$ for each $i\in [m]$.
Given any proper $w$-coloring $c$ of $V$ and any bijection $\pi:[m]\rightarrow [m]$, we define a function $\smash{\calF_{c,\pi}:2^V\rightarrow \mathbb{R}}$ as follows.
 
%  For each $i \in [m]$, let $A_{i, 1}, \cdots, A_{i,w}$ forms a partition of $A_{i}$, i.e., $A_i = A_{i,1} \cup \cdots \cup A_{i, w}$ and $A_{i, j} \cap A_{i, j'} = \emptyset$ when $j \neq j'$. We further assume $|A_{i,1}| = \cdots = |A_{i, w}| = \alpha k$, and $|A_i| = \alpha k w$. 
%Similarly, we assume $B_{i, 1}, \cdots, B_{i,w}$ forms a partition of $B_{i}$ and $|B_{i,1}| = \cdots = |B_{i,w}| = (1 - \alpha) k$, $|B_i| = (1- \alpha)kw$.
For any $S\subseteq V$, %$i \in [m], j \in [w]$ 
  let $\smash{y^S \in [0,1]^{mw}}$
  and $z^S\in [0,1]^{mw}$ be defined from $S$ as
$$
y^S_{i, j} = \frac{|S\cap A_{i, j}|}{\alpha k}\quad\text{and}\quad
z^{S}_{i,j} = \frac{|S\cap B_{i, j}|}{(1 -\alpha) k}.$$
For any any $I \subseteq [m]$, let $x^{S, I,\pi} \in [0,1]^{mw}$ be
$$
x^{S,I,\pi}_{i,j}=\begin{cases} y^S_{\pi(i),j} & \text{if $i\in I$} \\[0.8ex]
z^S_{i,j} & \text{if $i\notin I$}\end{cases}.
$$ %and we define $x_{S,A,i, j} = x_{S, A, i, j}$ when $i \in I$ and $x_{S,A,i, j} = x_{S, B, i, j}$ when $i \notin I$.
%we take $B_{i} = B_{i,1}\cup B_{i, w}$, $|B_{i, 1}| = \cdots = |B_{i,w}| = (1 - \alpha) k$ and $|B_i| = (1 - \alpha)kw$.
%Let $I_{A} \subseteq [m]$ to denote the index for $A$, and $I_B \subseteq [m]$ to denote the index of $B$. 
Finally we define $\calF_{c,\pi}:2^V\rightarrow \mathbb{R}$: For any $S\subseteq V$,
%Define the objective function $\mathcal{F}: 2^{V} \rightarrow \R^{+}$ as follow, for any $S\subseteq V$
\begin{align}
\label{eq:construction1}
\mathcal{F}_{c,\pi}(S) = \min\left\{ \sum_{I \subseteq [m]}\beta^{|I|}(1 - \beta )^{m - |I|} 
\cdot \hat{F}(x^{S, I,\pi}) + \frac{\eps}{k}|S|, \hspace{0.05cm}1 \right\}.
\end{align}
%and for each $I\subseteq [m]$, $\hat{F}_{I}: [0,1]^{mw} \rightarrow [0,1]$ only relates to $A_{i_a}$ where $i_a \in I$ and $B_{i_b}$ where $i_b \in [m]\backslash I$.
%Concretely, we have
%\begin{align}
%\hat{F}_{I}(S) = \hat{F}\left(\frac{|S\cup A_{i_a,1}|}{\alpha k}, \cdots, \frac{|S\cup A_{i_a, w}|}{\alpha  k}, \ldots, \frac{|S\cup A_{i_b,1}|}{1-\alpha  k}, \cdots, \frac{|S\cup A_{i_b, w}|}{1-\alpha  k} \right)
%\end{align}

We state some basic properties of the function $\mathcal{F}_{c,\pi}$.
\begin{lemma}
	\label{lem:basic}
	For any $c$ and $\pi$, $\calF_{c,\pi}$ satisfies the following two properties:
	\begin{enumerate}
		\item $\mathcal{F}_{c,\pi}$ is monotone, submodular and $\mathcal{F}_{c,\pi}(S) \in [0,1]$ for any
		  $S\subseteq V$.
		\item When $|S|\geq k/\eps$, $\mathcal{F}_{c,\pi}(S) = 1$.
		\item For any $i \in [m]$ and $j\in [w]$, we have $|A_{\pi   (i), j} \cup B_{i, j}| = k$ and $\mathcal{F}_{c,\pi}(A_{\pi(i), j} \cup B_{i, j}) \geq 1 - \eps$. 
	\end{enumerate}
\end{lemma}
\begin{proof}
We start with the first claim. By Lemma~\ref{lem:base2}, $\hat{F}$ is monotone and submodular. %Moreover, for any $I \subseteq [m]$, $i \in [m]$, $j \in [w]$, we defined $x_{S, I, i, j} = \frac{|S\cup A_{i, j}|}{\alpha k}$ when $i \in I$ and $x_{S, I, i, j} = \frac{|S\cup B_{i, j}|}{(1 - \alpha)k}$ when $i \notin I$. 
Hence, it~is easy to see that for each $I \subseteq [m]$, $\hat{F}(x^{S, I,\pi})$ is monotone and submodular in $S$. As addition and the min operation keep both submodularity and monotonity, $\mathcal{F}_{c,\pi}$ is both monotone and submodular.
	
	For the second claim, since $\smash{\hat{F}}$ is nonnegative, $\mathcal{F}_{c,\pi}(S) \geq \min\{\eps |S|/k, 1\} = 1$ when $|S| \geq k/\eps$.
	
	For the last claim, note that $|A_{\pi(i), j}| = \alpha k$, $|B_{i, j}| = (1 - \alpha)k$ and $A_{\pi(i), j} \cap B_{i,j} = \emptyset$. Therefore,  $|A_{\pi(i), j}\cup B_{i, j}| = k$. 
	Let $S = A_{\pi(i), j} \cup B_{i, j}$. 
	We have $\smash{x^{S,I,\pi}_{i,j}=1}$ for  all $I\subseteq [m]$ since 
	$\smash{y^{S}_{\pi(i),j}=z^S_{i,j}=1}$.
	By  Lemma~\ref{lem:base2}, $\smash{\hat{F}(x^{S, I,\pi}) \geq 1 -\eps}$ for all $I$ 
	and thus,
	%	 holds for all $I \subseteq [m]$. Consequently, one has
	\begin{align*}
	\mathcal{F}(S) =&~ \min\left\{ \sum_{I \subseteq [m]}\beta^{|I|}(1 - \beta )^{m - |I|}\cdot  \hat{F}(x^{S, I,\pi}) + \frac{\eps}{k}|S|, 1 \right\}  
	\ge \min\left\{  (1 - \eps)+ \frac{\eps}{k}|S|, 1 \right\}\ge 1 - \eps.
	\end{align*}
This finishes the proof of the lemma.\end{proof}

%For ease of notation, for any $S\subseteq V$, let $a_{S, i, j} = \frac{|S\cup A_{i, j}|}{\alpha k}$ and $b_{S, i, j} = \frac{|S\cup B_{i, j}}{ (1 -\alpha) k }$.  
For any subset $S\subseteq V$, we say $S$ is {\em balanced} with respect to $c$ if $$\max_{\substack{i\in [m]\\ j, j'\in [w]}}\left|y^{S}_{i, j} - y^{S}_{i, j'}\right| \leq \gamma\quad\text{and}\quad\max_{\substack{i\in [m]\\ j, j'\in [w]}}
\left|z^{S}_{i, j} - z^{S}_{i, j'}\right| \leq \gamma.$$ Define $\calG_\pi:2^V\rightarrow \mathbb{R}$ as
\begin{align}
\label{eq:construction}
\mathcal{G}_{\pi}(S) = \min\left\{ \sum_{I \subseteq [m]}\beta^{|I|}(1 - \beta )^{m - |I|}\cdot G(x^{S, I,\pi}) + \frac{\eps}{k}|S|, 1 \right\}.
\end{align}
Recall  the function $G(x)$ only depends on $\bar{x_i}$, $i\in [m]$.
As a result, $\calG_\pi(S)$ only depends on $|S\cap A_i|$ and $|S\cap B_i|$, $i\in [m]$.
%When one queries a balance set, the function value returned depends only on $\bar{x}_{S,A, i}:= \frac{1}{w}\sum_{j \in [w]}x_{S, A,i, j}$ and $\bar{x}_{S,B, i}:= \frac{1}{w}\sum_{j \in [w]}x_{S, B,i, j}$. That is, for each $i \in [m]$, it keeps the sysmmetry on $A_{i,1}, \cdots, A_{i, w}$ (resp. $B_{i, 1}, \cdots, B_{i, w}$) and one gets no information regarding the identity on the sub-group. 
The following Lemma is a direct consequence of Lemma~\ref{lem:base2}:
\begin{lemma}
When $S\subseteq V$ is balanced with respect to $c$, 
  we have $\mathcal{F}_{c,\pi}(S) = \mathcal{G}_\pi(S)$.
\end{lemma}

%Intuitively, one will always submit a balance query, since it is hard to distinguish the element. 
%For any subset $S\subseteq V$, let the support of $S$, denoted as $I(S)$, contain all index $i$($i\in [m])$ such that $S$ has non-empty intersection with $A_{i}\cup B_i$, i.e.,
%$
%I(S) = \left\{i \in [m]: S\cup (A_{i} \cup B_i)\neq \emptyset\right\}.
%$
%The next Lemma is essential to our proof. 
%Roughly speaking, let $S$ be a balance set and $d \in \mathbb{Z}$, let $S_i$ be any balance set that is contained in some $B_i$ and $|S_i| =d$, then the value of $\mathcal{F}(S \cup S_i)$ is the same as long as $i$ does not lie in the support of $S$. This indicates that the index $i$ is indistinguishable as long as $i \notin I(S)$.

Given any $S\subseteq V$, the next lemma captures the minimal information needed to evaluate 
  $\calG_{\pi}(S)$ without full knowledge of $c$ and $\pi$. 
  The proof can be found at Appendix~\ref{sec:lower1-app}.

%Given $S\subseteq V$, we write 
\begin{lemma}\label{hehehehe3}
Let $S\subseteq V$ and let $\pi,\pi':[m]\rightarrow [m]$ be two bijections. 
Then we have $\mathcal{G}_{ \pi}(S)=\mathcal{G}_{ \pi'}(S)$ when 
  the following condition holds: For each $i\in [m]$, we have
  $\pi(i)=\pi'(i)$ if 1) $S\cap B_i\ne \emptyset$ and 2) either $S\cap A_{\pi(i)}\ne \emptyset$
  or $S\cap A_{\pi'(i)}\ne \emptyset$.
%(1) $S$ is balanced with respect to both $c$ and $c'$; 
%(1) If $S\cap A_{\pi(i)}\ne \emptyset$ and $S\cap B_{ i }\ne \emptyset$ for some $i\in [m]$,
%  then we have $\pi'(i)=\pi(i)$; (2) If $S\cap A_{\pi'(i)}\ne \emptyset$ and 
%  $S\cap B_{i}\ne \emptyset$ for some $i\in [m]$, then we have $\pi(i)=\pi'(i)$.
%$\pi(i)=\pi'(i)$ for all $i\in [m]$ such that $S\cap B_{i}\ne \emptyset$.
%a balanced set with respect to $c$.
%Then 
\end{lemma}

\begin{remark}\label{heheremark1}
Let $S\subseteq V$. If (1) $S$ is known to be balanced with respect to $c$
  and (2) we are given 
\begin{equation}\label{hohoho}
\big\{(i,\pi(i)): \text{$i\in [m]$, $S\cap B_i\ne \emptyset$ and
  $S\cap A_{\pi(i)}\ne \emptyset$}\big\},
\end{equation}
  then one can evaluate $\calF_{c,\pi}(S)$ without more information about $c$ and $\pi$.
To see this, (1) implies that it suffices to evaluate $\calG_\pi(S)$.
Lemma \ref{hehehehe3} implies that $\calG_\pi(S)$ is uniquely determined 
  given (\ref{hohoho}).
\end{remark}

We delay the proof of the following lemma to Appendix \ref{sec:lower1-app}. This is 
  where we need to choose the two parameters $\alpha$ and $\beta$ carefully to 
  minimize the constant $0.584$ in the lemma.

\begin{lemma}\label{hardlemma}
Fix any $i^{*}\in [m]$.
Let $S\subseteq A\cup B_{i^{*}}$ be a set such that $S$ is balanced with respect to $c$
  and $S\cap A_{\pi (i^{*})}=\emptyset$.
Then we have $\calF_{c,\pi}(S)\le 0.5839+3\eps$.
\end{lemma}

\def\br{\mathbf{r}} \def\calD{\mathcal{D}} \def\bQ{\mathbf{Q}}
\def\ALGG{\mathsf{ALG}} 

\subsection{Lower bound for dynamic submodular maximization}

%Let $\ALGG$ be a randomized dynamic algorithm with approximation guarantee at least 
%  $0.584+4\eps$ and total query complexity $q$.
%We will use $\br$ to denote the randomness of $\ALGG$ so once $r$ is fixed,
%  $\ALGG_r$ is just a $q$-query deterministic algorithm.
%We prove   Theorem~\ref{thm:lower1} by giving a lower bound on $q$.
%  \medskip
  
%\paragraph{Choice of paramters} Throughout this section, we ues the following choice paramters. Let $\alpha, \beta \in (0,1), \eps > 0$, $w = \Omega\left(\frac{1}{\alpha\eps} + \frac{1}{(1 -\alpha)\eps}\right)$,$\gamma = w^{-1}\exp(-4w^6/\eps)$, $k = \Omega(\eps^{-1}\gamma^{-2}\alpha^{-2}(1 - \alpha)^{-2}\log n)$, $n = (2 - \alpha)mkw \leq 2mkw$.
\noindent\textbf{Hard streams.} 
%Let $\calD$ be the following distribution of instances of
%  dynamic submodular maximization, where each instance
%To draw $\bQ\sim \calD$, we first draw a $w$-coloring $c$ of $V$ and 
%  a bijection $\bpi$
% $\bQ\sim \calD$ consists of a monotone submodular function and a dynamic stream).
%To draw
%Let the ground set $V^*$ given to the algorithm at the beginning be $A\cup B^*$, where $B^*$ is the union of $B_1,\ldots,B_{0.1m}$.
%We choose a $w$-coloring $c$ of $V$ and a bijection $\pi$. %, which are both unknown to the algorithm,
%  and set the function $\calF^*:2^V\rightarrow \mathbb{R}$ to be $\calF_{c,\pi}$.
%They define the following instance which we denote by $Q_{c,\pi}$:
%$$
%\calF^*(S)=\calF_{c,\pi}(S),\quad\text{for any $S\subseteq V^*$.}
%$$
The ground set is $V$, which is known to the algorithm; the
  monotone submo\-dular function $\calF^*:2^V\rightarrow \mathbb{R}$ is 
  $\calF_{c,\pi}$, where $c$ is a proper $w$-coloring of $V$ and $\pi:[m]\rightarrow [m]$ is a bijection, both of which are unknown to the algorithm.
In the stream we first insert all elements of $A$, which takes
  $\alpha mkw$ insertions.  
We then divide the rest of the stream into $m$ batches: For the 
  $t$-th batch, we first insert all elements of $B_t$ and then delete them.  
So the total length is $n=(2-\alpha)mkw$.
This is the only stream we will use in the proof but the function is determined 
  by the unknown $c$ and $\pi$.

We will focus on the $0.1m$ rounds when elements of $B_t$, 
  for each $t=1,\ldots,0.1m$, has just been inserted. 
To ease the presentation we consider the following dynamic problem that consists of $0.1m$ stages.
During the $t$-th stage, an algorithm can query $\calF^*:2^V\rightarrow \mathbb{R}$ about any  $S\subseteq
  A\cup B_1\cup \cdots \cup B_t$ and can choose to start the next stage at any time by outputting a set   %
%  move to the next stage after it queries
  $S_t\subseteq A\cup B_t$ of size at most $k$. 
  % such that $|S_t|\le k$ and $\calF^*(S_t)\ge 0.584+3\eps$.
The algorithm succeeds after $0.1m$ stages if $\calF^*(S_t)> 0.5839+3\eps$ for every $t\in [0.1m]$. 
%it passes stage $0.1m$.
We~prove that any randomized algorithm that succeeds with probability at least $2/3$
  must have total query complexity of at least $\Omega(n^2/k^3)$.
It then follows from Lemma \ref{lem:basic}
  (part 3) that any randomized algorithm for dynamic submodular maximization 
  with an approximation
  guarantee of $0.5839+4\eps$ must have amortized query complexity $\Omega(n/k^3)$.
  
%  needs to return a set $S_t$ of size at most $k$
%  such that $\calF^*(S_t)\ge 0.584+3\eps$.

%Assuming that $\ALGG$ succeeds on every $Q_{c,\pi}$ with  
%We will focus on the time when $B_t$ has just been inserted, for each $t$ from $1$ to $0.1m$.
%A necessary condition for the algorithm to succeed is to output a set $S_t$
%  at the end of this round such that $\calF^*(S_t)\ge 0.584+3\eps$.
%$$
%%n=\alpha mkw + 0.1m\cdot 2(1-\alpha)kw=(0.2+0.8\alpha)mkw.
%$$
%We construct the hard instance as follow. 
%We assume elements in $A_1, \cdots, A_m$ arrive in a random order, and this takes up $\alpha mkw$ insertion operations.
%We then divide the update stream into $m$ batches, where for the $t$-th batch, the nature draws an index $i_t \in [m]$ uniformly at random and insert all elements of $B_{i_t}$ in a random order, it then deletes all element of $B_{i_t}$.
%\paragraph{Hard instance} 
%We construct the hard instance as follow. 
%We assume elements in $A_1, \cdots, A_m$ arrive in a random order, and this takes up $\alpha mkw$ %insertion operations.
%We then divide the update stream into $m$ batches, where for the $t$-th batch, the nature draws an index $i_t \in [m]$ uniformly at random and insert all elements of $B_{i_t}$ in a random order, it then deletes all element of $B_{i_t}$.   
%We are now ready to prove Theorem~\ref{thm:lower1}.

\begin{proof}[Proof of Theorem \ref{thm:lower1}]
%Let $\ALGG$ be a randomized algorithm that succeeds on any dynamic stream described earlier
%  (with unknown $c$ and $\pi$) with probability at least $2/3$.
Let $\ALGG$ be a randomized algorithm for the $0.1m$-stage problem described above 
  that succeeds with probability at least $2/3$. 
By Lemma \ref{lem:basic}, we may assume that 
  $\ALGG$ never queries a set of size more than $k/\eps$.
To lower bound the query complexity of $\ALGG$ we consider the following simple 
  $0.1m$-stage game.
%We focus on the time when $B_t$ has just been inserted, for $t\in [0.1m]$.

In the game there is a hidden proper $w$-coloring $c$ of $V$ and 
  a hidden bijection $\pi:[m]\rightarrow [m]$.
The game similarly consists of $0.1m$ stages.  
During each round of stage $t$, the algorithm can make a query by either 
  (1) picking a subset $S$ of $V$ of size at most $k/\eps$
  or (2) picking a number $i\in [m]$.
In case~(1) the $c$-oracle returns 
  ``balanced'' or ``unbalanced,'' as whether $S$ is balanced with respect to $c$ or not;
  when receiving ``unbalanced'' the algorithm wins the whole game.
In case (2) the $\pi$-oracle returns 
  ``matched'' or ``not matched,'' as whether $\pi(t)=i$ or not;
  after receiving ``matched,'' the algorithm can choose to proceed to the next stage,
  and it wins the game if it passes all $0.1m$ stages.
% 
%the algorithm wins the whole game if $S$ is not balanced with respect to the 
%  hidden $c$; 
%  and the oracle reveals whether $\pi^{-1}(t)=i$ or not; the algorithm
%  can move to the next stage after finding out $\pi^{-1}(t)$
%  and can claim victory if it passes all $0.1m$ stages.

We show that any randomized algorithm that wins the game with probability at least $2/3$ 
  must use at least $\Omega(m^2)$ queries.
To this end, we consider the distribution of $(\cc,\bpi)$ where $\cc$ and $\bpi$
  are drawn uniformly and independently, and show that any deterministic algorithm
  that wins the game with probability at least $2/3$ must use $\Omega(m^2)$ queries.

Showing this distributional $\Omega(m^2)$ lower bound for each subgame is easy.
For the subgame of finding an unbalanced set with respect to $\cc$, we have
  for any set $S$ of size at most $k/\eps$ that
\begin{align*}
\Pr\left[y^{S}_{i, j} - y^{S}_{i, j'}> \gamma\right] & = \Pr\Big[|S\cap A_{i, j}| - |S\cap A_{i, j'}| > \gamma \alpha k\Big] \\
&\leq \exp\left(- \frac{\eps\gamma^2\alpha^2 k^2}{2k }\right) \leq \frac{1}{n^5},\quad \text{and}\\[1.5ex]
\Pr\left[z^{S}_{i, j} - z^{S}_{i, j'}> \gamma\right]  &= \Pr\Big[|S\cap B_{i, j}| - |S\cap B_{i, j'}| > \gamma (1 -\alpha) k\Big] \\
&\leq \exp\left(- \frac{\eps\gamma^2(1 - \alpha)^2 k^2}{2k }\right) \leq \frac{1}{n^5}
\end{align*}
for any fixed $i$ and $j\ne j'$, given our choice of $k$.
Therefore, any deterministic algorithm that finds an unbalanced set 
  with probability at least $1/3$ requires $\Omega(n^5)$ queries.
For the subgame about $\bpi$,
  at the beginning of each stage $t$, any of the remaining $m-(t-1)$
  indices (other than $\pi (1),\ldots,\pi (t-1)$) is equally likely to be 
  $\pi (t)$.
Given that $m-(t-1)\ge 0.9m$, it takes $\Omega(m)$ queries to pass each of the $0.1m$ stages
  and thus, any deterministic algorithm that wins the subgame about $\bpi$ 
  with probability at 
  least $1/3$ requires $\Omega(m^2)$ queries. 

Back to the original game, if a deterministic algorithm $\calA$ can win with probability at least
  $2/3$, then it can either win the first or the second subgame with probability
  at least~$1/3$.
Assuming~for example it is the latter case, we get a randomized algorithm
  for winning the second subgame over $\bpi$ with probability at least $1/3$ by first drawing
  $\cc$ and then simulating $\calA$ on the original game. Given that the 
  number of queries is at most that of $\calA$, we have that query complexity of $\calA$ is $\Omega(m^2)$.
  
To finish the proof  we show how to use $\ALGG$ to play the game.
Let $c$ be the hidden coloring and $\pi$ be the hidden bijection.
We simulate the execution of $\ALGG$ on $\calF^*=\calF_{c,\pi}$ as follows:
\begin{flushleft}\begin{enumerate}
\item During the first stage (of both the dynamic problem of maximizing 
  $\calF^*$ over $A\cup B_1$ and the game), letting $S\subseteq A\cup B_1$ of size at most $k/\eps$
  be any query made by $\ALGG$,
  we make one query $S$ on the $c$-oracle (to see if $S$ is balanced or not) and make 
  one query $i$ on the $\pi$-oracle for each $i\in [m]$ with $S\cap A_i\ne\emptyset$ (to see if $\pi (1)=i$).
If $S$ is not balanced, we already won the game; otherwise, we know $\calF^*(S)=\calG_\pi(S)$
  and the latter can be computed (see Remark \ref{heheremark1}) using information returned by
  the $\pi$-oracle, so that we can continue the simulation of $\ALGG$.
When $\ALGG$ decides to output $S_1\subseteq A\cup B_1$ of size at most $k$,
  we query $S_1$ on the $c$-oracle and query each $i$ on the $\pi$-oracle with $S_1\cap A_i\ne \emptyset$.
If $S_1$ is unbalanced, we won the game; if $\pi (1)=i$ for some $i$ queried,
  we move to stage $2$ (in both the dynamic problem and the game); if $S_1$ is balanced and $\pi(1) \ne i$ for any $i:S_1\cap A_i\ne \emptyset$,
  it follows from Lemma \ref{hardlemma} that $\ALGG$ has failed so we terminate the simulation and fail the game.
\item During the $t$-th stage (of both the dynamic problem and the game),
  we similarly simulate each query $S\subseteq A\cup B_1\cup \cdots\cup B_t$ of $\ALGG$.
The only difference is that, given that we have passed the first $(t-1)$ stages of the game,
  we already know $\pi (1),\ldots,\pi (t-1)$ and thus, information
  returned by the $\pi$-oracle would be enough for us to evaluate $\calF^*(S)$.
When $\ALGG$ returns $S_t$, we query $S_t$ and each $i$ with $S_t\cap A_i\ne \emptyset$,
  and act according to results similarly.
\end{enumerate}\end{flushleft}
To summarize, the simulation has three possible outcomes: (1) we won the game
  because  an unbalanced set has been found; (2) we won the game because we have passed
  all $0.1m$ stages; or (3)~$\ALGG$ fails to find $S_t$ with $\calF_{c,\pi}(S_t)\ge 0.5839+3\eps$
  for some $t$.
Given that (3) only happens with probability at most $1/3$, we obtain an algorithm
  for the game that succeeds with probability at least $2/3$ and thus, must use
  $\Omega(m^2)$ queries.
To finish the proof, we note that if $\ALGG$ has total query complexity $q$ then
  the algorithm we obtain for the game has total query complexity at most
$$
q\cdot \left(1+\frac{k}{\eps}\right)+0.1m\cdot (1+k),
$$
which implies $q=\Omega(m^2/k)$ and thus, the amortized complexity of $\ALGG$ is $\Omega(m/k^2)=\Omega(n/k^3)$. 
Taking $\eps = 3\times 10^{-5}$, we get the lower bound on approximation ratio.
\end{proof}

\section{A Polynomial Lower Bound for $1/2+\eps$ Approximation}
\label{sec:lower2}

\def\Sample{\textsf{Sample}} \def\bOmega{\boldsymbol{\Omega}} \def\bR{\mathbf{R}}

We restate the main theorem of this section:

\thmlowertwo*

%\begin{theorem}%
%	\label{thm:lower2}
%For any constant $\eps>0$ there is a constant $C_\eps> 0$ with the following property. 
%Let $n$ be the length of the stream. 
%When $k\ge C_\eps$, any randomized algorithm for dynamic submodular~maxi\-mization
%  under a cardinality constraint of $k$ with an approximation guarantee of $1/2+\eps$
%  must have amortized query complexity $\smash{n^{\tilde{\Omega}(\eps)}/ k^3}$.	
	%	Let $\eps > 0$.	In a dynamic stream with both insertion and deletion updates, an algorithm for maximizing monotone sumodudlar function under cardinality constrains with an approximation guarantee of $\frac{1}{2} + \tilde{O}(\eps)$ must have amortized number of query of at least $\Omega(\eps^4 n^{\eps} / k^2)$.
%\end{theorem}

\subsection{The construction}
\label{sec:lower2-construction}

Let $\eps>0$ be a positive constant.
We assume that both $1/\eps$ and $\eps k$ are positive integers.
Our~goal is to show that any randomized algorithm for dynamic submodular maximization
  under a cardinality constraint of $k$ with approximation $\smash{1/2+\tilde{O}(\eps)}$
  must have amortized query complexity $\Omega(n^\eps/k^3)$.

%We assume without loss of generality that both $1/\eps$ and $\eps k$ are positive integers.
We start with the construction of a monotone submodular function $\calF:V\rightarrow \mathbb{R}$.
The family of functions used in our lower bound proof will be obtained from $\calF$
  by carefully shuffling elements of $V$.
Let $L=1/\eps\in \mathbb{N}$ and let $m=(m_1,\ldots,m_L)$ be a tuple of positive 
  integers to be fixed later, with $m_{L}$ set to be $1$.
Let $T$ be a tree of depth $L$, where the root is at depth $0$ and its~leaves are at depth $L $.
Each internal node of $T$ at depth $\ell \in [0:L-1]$ has $m_{\ell+1}$ children;
   the number of~nodes at depth $\ell\in [L]$ is $m_1\cdots m_\ell$.
We use $\gamma$ to refer to the root of $T$ and write $U_0=\{\gamma\}$;
  for each $\ell\in [L ]$,
  we use $V_\ell= [m_1]\times \cdots \times [m_\ell]$ to 
  refer to nodes of $T$ at depth $\ell$.
So the set of nodes is $U_0\cup U_1\cup\cdots \cup U_{L }$ and whenever we refer to a node 
  $u$ of $T$ at depth $\ell$, it should be
  considered as a tuple $(u_1,\ldots,u_\ell)\in U_\ell$. 
Children of $u\in U_\ell$, $\ell\le L-1$,  
  are given by $(u_1,\ldots,u_\ell,1),\ldots, (u_1,\ldots,u_\ell,m_{\ell+1})$.
  %,i_{\ell+1})$ for $i_{\ell+1}\in [m_{\ell+1}]$.

The ground set $V$ of $\calF$ is defined as follows. 
For each node $u\in U:=U_1\cup\cdots \cup U_L$, 
  we introduce a set $A_{u}=\{a_{u,1},\ldots,a_{u,w}\}$ of $w:=\eps k$ new elements.
We define $V$ as the union of 
   $A_{u}$ for all $u\in U$.  
% $= 1/\eps$ be the total number of levels. Let $m_1, \cdots, m_L \in \mathbb{Z}$, we will determine their value later. We partition the ground set $V$ into $L$ groups, i.e., $V= V_1 \cup \cdots \cup V_{L}$. For any $\ell \in [L]$, we further partition the $\ell$-th group into $m_1m_2\cdots m_\ell$ parts, denoted as
%\[
%V_{\ell} = \bigcup\limits_{(i_1, \cdots, i_{\ell}) \in [m_1]\times \cdots \times [m_{\ell}]} A_{i_1, \cdots, i_{\ell}},
%\]
%and we further assume $A_{i_1, \cdots, i_{\ell}}$ consists of $\eps k$ elements, i.e., $|A_{i_1, \cdots, i_{\ell}}| = \eps k$.
%Denote $M_{\ell} = [m_1]\times \cdots \times [m_\ell]$ and $M = M_1\cup \cdots M_{L}$.
%Let the weight sequence $\{w_\ell\}_{\ell \in [L]} \in \R_{+}^{n}$ satisfy $\sum_{\ell=1}^{L} w_{\ell} = 1$, we assign different weights to elements in different groups.
%We next construct the weight sequence, it follows from \cite{feldman2020one}. 
%Let $\delta_{L} = 1$ and for $\ell = L-1, L-2, \cdots,1$ let $\delta_{\ell} = 1 + \left(\frac{1+ \sqrt{1 + 4/\delta_{\ell+1}}}{2}\right)\cdot \delta_{\ell+1}$ and set the weight 
%\[
%w_{\ell} = \prod_{i =1}^{\ell-1}\left(\frac{\delta_i - 1}{\delta_{i+1}}\right) = \prod_{i=1}^{\ell-1}\frac{1}{1 - 1/\delta_{i}}.
%\]
To construct $\calF:2^V\rightarrow \mathbb{R}$, we 
  will utilize a weight sequence $\{w_\ell\}_{\ell \in [0:L]}$ from \cite{feldman2020one}
%  that satisfy $w_\ell\ge 0$ and $\sum_{\ell\in [L]} w_{\ell} = 1$
  to define a probability distribution $\calD$ over subsets of nodes of $T$.
We specify the sequence later in Lemma \ref{def:weight-sequence}; for now it suffices to know that
  $w_0=0$, $w_\ell$'s are nonnegative, and they sum to $1$.
  
  \begin{algorithm}[!t]
  	\caption{$\Sample$}
  	\label{algo:samplealgo}
  	\begin{algorithmic}[1]
  		\State \textbf{Input: }A node $v$ of $T$ at depth $\ell\in [0:L]$
  		%\If{$\ell=L$}
  		%\State Return $v$
  		%\Else 
  		\State Let $p_\ell$ be given as below:
  		$$
  		p_\ell= \frac{w_\ell}{1-\sum_{i=0}^{\ell-1} w_i}\in [0,1].
  		$$
  		\State With probability $p_\ell$, return $v$.  \Comment{Since $p_L=1$, this always happen when $\ell=L$ }
  		\State Otherwise, run $\Sample(v_i)$ independently
  		on each child $v_i$ of $v$ in the tree and return the union %, where 
  		%return the union of $\Sample(v_1),\cdots, \Sample(v_{m_{\ell+1}})$ using fresh
  		%\State randomness, where 
  		%$v_1,\ldots,v_{m_{\ell+1}}$ are children of $v$, and return their union.
  		%\EndIf
  		%\State 
  	\end{algorithmic}
  \end{algorithm}

We define the distribution of $\calD$ over $2^U$ as follows.
Drawing a sample from $\bR\sim \calD$ can be done by calling the recursive
  procedure $\Sample$ in Algorithm \ref{algo:samplealgo} on the root of $T$.
Informally, starting with $\bR=\emptyset$, $\Sample$ performs a DFS walk on $T$.
Whenever reaching a node $v$ at depth $\ell$, it adds $v$ to $\bR$ and does not 
  explore any of its children with probability $p_\ell$ (see Algorithm \ref{algo:samplealgo};
   note that $p_0=0$ and $p_L=1$); otherwise, it continues the DFS walk to visit each of its children.  
Because $p_0=0$,
  the root is never included in the set and thus, $\bR$ is always a subset of $U$.

	We need the following properties about $\bR$, the proof can be found at Appendix~\ref{sec:lower2-app}.
\begin{lemma}\label{usefulproperties}
	Let $\bR\sim \calD$, we need have the following properties
	\begin{itemize}
\item For any node $u$ of $T$ at depth $\ell$, the probability of $u\in \bR$ with
  $\bR\sim \calD$ is $w_\ell$.
  
\item For any root-to-leaf path of $T$, every $R$ in the support of $\calD$
  has exactly one node in the path.  
  \end{itemize}
\end{lemma}

Let $P_1\subseteq U$ and $P_2\subseteq U$ be two subsets of nodes of $T$
  and let $\tau$ be a bijection from $P_1$ to $P_2$ such that
%, both rooted at the root of $T$. 
%We say they are \emph{isomorphic}
%  if there is a bijection $\tau$ from the set $P_1$ of nodes of $T_1$ to the 
%  set $P_2$ of nodes of $T_2$ such that
  (1) $u$ and $\tau(u)$ are at the same depth for every  $u\in P_1$; and (2)
  for any $u,v\in P_1$, the~depth of the lowest common ancestor (LCA) of $u$ and $v$
  is the same as that of $\tau(u)$ and $\tau(v)$.
%We will abuse the notation to use $T_1$ and $T_2$ to denote their sets 
%  of nodes.
The following lemma follows from how the procedure $\Sample$ works:

\begin{lemma}\label{hhhhheeeee1}
%Let $T_1$ and $T_2$ are two subtrees of $T$ that are rooted at the root of $T$
%  and are isomorphic under $\tau:P_1\rightarrow P_2$.
%Then 
The distribution of $P_2\cap \bR$, $\bR\sim \calD$, is distributed the same as
  first drawing $\bR\sim \calD$, taking $P_1 \cap \bR$ and applying $\tau$ on $P_1\cap \bR$.\end{lemma}
\begin{proof}
Let $T_1$ (or $T_2$) be the subtree that consists of paths from nodes in $P_1$ (or $P_2$)
  to the root.
Using conditions (1) and (2), one can show (by induction) that $\tau$ can be extended to obtain an 
  isomorphism from $P_1$ to $P_2$. 
Then $P_1\cap \bR$ and $P_2\cap \bR$ can be obtained by running $\Sample$
  on the same tree and the two distributions are the same under $\tau$.
%The claim follows since $P_1\cap \bR$ (or $P_2\cap \bR$) with $\bR\sim \calD$
%  can be obtained by just running $\Sample$ on $T_1$ and then taking the intersection of its 
%  output with $P_1$ (or $P_2$).
\end{proof}

We use $\calD$ to define $\calF$. 
Given $S\subseteq V$, we define $x^S\in [0,1]^U$ as the vector indexed by $u\in U$ with
$$
x^S_u=\frac{|S\cap A_u|}{\eps k}\in [0,1].
$$
For any $R$ in the support of $\calD$ we define $g_R:[0,1]^U\rightarrow [0,1]$ 
  and $\calG:[0,1]^U\rightarrow [0,1]$ as 
%$I(S) = \{u \in M: S\cap A_{u} \neq \emptyset\} $.
%Then for any outcome $\omega \in \Omega$ we define
\begin{align}
\label{eq:obj1}
g_{R}(x) :=1 - \prod_{u\in R}\left(1 - x_u\right)\quad\text{and}\quad
\calG(x)= \E_{\bR\sim \calD} \big[g_{\bR}(x )\big].
\end{align}
We are now ready to define  $\calF$ over $2^{V}$ as 
\begin{align}
\label{eq:obj}
\mathcal{F}(S) = \min\left\{  \calG(x^S) + \frac{\eps}{k}|S|, \hspace{0.04cm}1\right\},
\quad\text{for any set $S\subseteq V$.}
\end{align}

We state some basic properties of $\calF$. The detailed proof can be found at Appendix~\ref{sec:lower2-app}.

\begin{lemma}\label{lem:basicbasic}
  $\mathcal{F}$ is monotone submodular and $\mathcal{F}(S) \in [0,1]$ for any $S\subseteq V$. Moreover, $\calF(S) = 1$ whenever $|S| \geq k/\eps$.
\end{lemma}

For each leaf $u\in U_L$, we define two important subsets of $V$.
The first set, $\calA_u$, is the union of $A$'s along the path from $u$'s parent to the root
  (not including the root):
$$
\calA_u:=A_{u_1}\cup A_{u_1,u_2}\cup \cdots\cup A_{u_1,\ldots,u_{L-1}}.
$$
The other set $W_u$ is the union of $A$'s of nodes that are children of nodes along
  the root-to-$u$ path: 
$$
W_u:=\bigcup_{\ell\in [L]} \bigcup_{j\in [m_\ell]} A_{u_1,\ldots,u_{\ell-1},j}.
$$
The following lemma shows that $\calF(\calA_u\cup A_u)$ is large:

\begin{lemma}
	\label{lem:prop2}
	 For any leaf $u\in U_L$, 
	 we have $|\calA_u\cup A_u| = k$ and $\mathcal{F}(\calA_u\cup A_u) = 1$.
\end{lemma}
\begin{proof}
Since these $L=1/\eps$ sets are pairwise disjoint and each has size $\eps k$,
$S$ has size $k$.
The second part follows from the second part of Lemma \ref{usefulproperties}
  and the definition of $g_R$.
	% Next, we prove
	%\[
	%\int_{\omega\in \Omega} f_{\omega}(S)\d\mu  =1.
	%\]
	%For any $\omega \in \Omega$, since the weight sequene satisfies $\sum_{\ell \in [L]}w_{\ell} = 1$, we know that $\Omega(A_{i_1})\cup \cdots \cup \Omega(A_{i_1, \cdots, i_{L}}) = \Omega$. Hence, there exists a level index $\ell \in [L]$, such that $\omega \in A_{i_1,\cdots, i_{\ell-1}}$. Since we have $|S\cap A_{i_1,\cdots, i_{\ell}}| = \eps k$, and 
	%\[
	%f_{\omega}(S) =1 - \prod_{u\in I(S): \omega\in A_{u}}\left(1 - \frac{|S\cap A_u|}{\eps k}\right) =1.
	%\]
	%This implies $\mathcal{F}(S) = \min\left\{  \int_{\omega} f_{\omega}(S)\d\,\mu   + \frac{\eps}{k}|S|, 1\right\} = \min\{1 + \frac{\eps}{k}|S|, 1\}=1$.
\end{proof}

We would like to show in the next lemma that any set $S\subseteq W_u$ 
  that has size at most $k$ and~does not overlap with $\calA_u$ must
  have a small $\calF(S)$.
For this we need to specify the weight sequence $\{w_{\ell}\}$
  that we will use from  \cite{feldman2020one} (recall that $w_0=0$).
 \begin{definition}[Weight sequence]
 	\label{def:weight-sequence}
 	Define $\{\delta_{\ell}\}_{\ell\in [L]}$, $\{a_{\ell}\}_{\ell \in [L]}$ and the weight sequence $\{w_\ell\}_{\ell \in [L]}$ inductively as follow.
 	We set $\delta_{L} = 1$ and for $\ell = L-1, L-2, \ldots,1$, let 
 	\[
 	\delta_{\ell} = 1 + \left(\frac{1+ \sqrt{1 + 4/\delta_{\ell+1}}}{2}\right)\cdot \delta_{\ell+1}
 	\]
 	and
 	\[
 	a_{\ell} = \prod_{i =1}^{\ell-1}\left(\frac{\delta_i - 1}{\delta_{i+1}}\right) = \prod_{i=1}^{\ell-1}\frac{1}{1 - 1/\delta_{i}},
 	\]
 	the weight sequence $\{w_\ell\}_{\ell \in [L]}$ is defined as $w_{\ell} =  {a_{\ell}}/{\sum_{\ell \in [L]}{a_{\ell}}}$.
 \end{definition}

%
%$(i_1,\cdots, i_{L}) \in M_{L}$, let $W_{i_1, \cdots, i_{\ell}} \subseteq V$ contains all element from the path $(i_1, \cdots, i_{\ell})$, i.e., $W = \bigcup_{\ell \in [L]}\bigcup_{j \in [m_{\ell}]} A_{i_1, \cdots, i_{\ell-1},j}$. 
%Lemma~\ref{lem:prop2} indicates that $A_{u_1}\cup A_{u_1, u_2} \cup \cdots  \cup A_{u_1, \cdots, u_{L}}$ is a subset of $W_u$ of size $k$ with value $1$.
%On the other hand, the weight sequence $\{w_\ell\}$ constructed above guarantees that 
%  any $S\subseteq W_u$ of size at most $k$ that does not overlap with 
%  $A_{u_1}\cup A_{u_1,u_2} \cup \cdots  \cup A_{u_1, \cdots, u_{L}}$, then 
%  $\calF(S)$ is at most $0.5 + \smash{\tilde{O}(\eps)}$. 
%We remark the proof is adapted from Lemma 5.3 in \cite{feldman2020one}, but with some 
%The next lemma is critical, it shows the optimal value is at most $\frac{1}{2} + H(\eps)$ when it does not include any element from the optimal path. 
%The proof follows from Lemma 5.3 in \cite{feldman2020one}. We have some minor difference and we provide a proof in Appendix~\ref{sec:lower2-app} for completeness.
The proof of the following lemma is adapted from Lemma 5.3 in \cite{feldman2020one} with some generalizations, and the proof can be found in Appendix~\ref{sec:lower2-app}.
\begin{lemma}
	\label{lem:lower2-2}
	For any leaf $u\in U_L$  and $S\subseteq W_u$ with $|S|\le k$ and 
	$S\cap \calA_{u}=\emptyset$, 
	we have
	\[
	%\max_{S\subseteq W\backslash (A_{i_1} \cup \cdots \cup A_{i_1, \cdots, i_{L}}), |S|\leq k}
	\mathcal{F}(S) \leq 0.5+O\left({\eps}\hspace{0.04cm}{\log^2(1/\eps)}\right).
	\]
\end{lemma}

As it will become clear soon at the beginning of the next subsection,
  our goal behind the family of hard streams is to have the dynamic algorithm solve repeatedly the question
  of finding an $S\subseteq W_u$ with $|S|\le k$
  and $S\cap \calA_u\ne \emptyset$, for a large number of leafs $u$ of $T$.
These questions are, however, only interesting after we shuffle 
  nodes of $T$ in the fashion to be described next.
%To this end, we need to shuffle elements of $V$ by shuffling nodes of $T$ as follows to obtain
%  the family of functions from $\calF$ for our hard instances. 
  
A \emph{shuffling} $\pi$ of $T$ consists of a bijection $\pi_u:[m_{\ell+1}]\rightarrow
  [m_{\ell+1}]$ for every $u\in U_\ell$, $\ell\in [0:L-1]$.
We use $\pi$ to shuffle each node $u\in U$ to $\pi(u)$ as follows: $\pi(\gamma)=\gamma$; 
  for each $u$ of depth $\ell\in [L]$, set
    $$\pi(u)=\big(u_1,\ldots,u_{\ell-1},\pi_{u_1,\ldots,u_{\ell-1}}(u_\ell)\big).$$
%\end{enumerate}
%We define $T_\pi$ to be the tree obtained after this shuffling:
%  $(u,v)$ is an edge in $T$ iff $(\pi(u),\pi(v))$ is an edge in $T_\pi$.
%with the same structure as $T$ but
%  $(u,v)$ is an edge in $T_\pi$ iff $(\pi(u),\pi(v))$ is an edge in $T$.
%So after shuffling, the parent node of $u=(u_1,\ldots,u_{\ell+1})$ becomes
%  the $(u_1,\ldots,u_{\ell-1},\pi_{u_1,\ldots,u_{\ell-1}}(u_\ell ))$.
  % iff $u_1=v_1,\ldots,u_{\ell-2}=v_{\ell-2}$ and
  %$\pi_{u_{-(\ell-1)}}(u_{\ell-1})=v_{\ell-1}$
A shuffling $\pi$ induces a bijection $\rho_\pi:V\rightarrow V$:
For each element $a_{u,i}$ for some $u\in U$ and $i\in [w]$, we set $\rho_\pi(a_{u,i})=
  a_{\pi^{-1}(u),i}$.
%  the $u$-th entry of $x$ becomes the $\pi^{-1}(u)$-th entry of $\rho_\pi(x)$.
Finally we define for each shuffling $\pi$ of $T$,
$$
\calF_\pi(S):=\calF\big(\rho_\pi(S)\big). 
$$
It is clear that Lemma \ref{lem:basicbasic} holds for $\calF_\pi$ for any shuffling $\pi$.

Given a leaf $u\in U_L$ of $T$ and a shuffling $\pi$, let 
$$
\calA_u^\pi:= A_{u_1,\ldots,u_{L-2},\pi_{u_1,\ldots,u_{L-2}}(u_{L-1})}
\cup \cdots A_{u_1,\pi_{u_1}(u_2)}\cup A_{\pi_\eps(u_1)}.
$$
%Note that $(u_1,\ldots,u_\ell,$
We get the following corollary of Lemma \ref{lem:prop2} and Lemma \ref{lem:lower2-2}:

\begin{corollary}\label{simplecoro1}
	For any shuffling $\pi$ of $T$, we have
\begin{itemize}
\item For any leaf $u\in U_L$, we have  $|\calA^\pi_u\cup A_u|=k$ and $\calF_\pi(\calA^\pi_u\cup A_u)=1$.

\item For any leaf $u$ and $S\subseteq W_u$ with $|S|\le k$ and $S\cap \calA_u^\pi=\emptyset$, we have
\begin{equation}\label{heheeqeq}
\calF_\pi(S)\le 0.5+O\left( \eps\hspace{0.04cm} {\log^2(1/\eps)}\right).
\end{equation}
%For any leaf $u$ and $S\subseteq W_u$ $(i_1, \cdots, i_L) \in M_{L}$, let $W = \bigcup_{\ell \in [L]}\bigcup_{j \in [m_{\ell}]} A_{i_1, \cdots, i_{\ell -1},j}$. Then we have
%	\[
%	\max_{S\subseteq W\backslash (A_{i_1} \cup \cdots \cup A_{i_1, \cdots, i_{L}}), |S|\leq k}\mathcal{F}(S) \leq \frac{1}{2} +O(\eps \log^2(1/\eps)).
%	\]
\end{itemize}
\end{corollary}
\begin{proof}
The first part follows from Lemma \ref{lem:prop2} and the observation that 
  $\rho_\pi(\calA_u^\pi\cup A_u)=\calA_u\cup A_u$.
The second part follows from $\rho_\pi(S)\subseteq W_u$,
  $\rho_\pi(S)\cap \calA_u= \emptyset$  and Lemma \ref{lem:lower2-2}.
\end{proof}
%The following Lemma indicates that the function value is indistinguishable, when one can not query the correct set.
%The following Lemma specify the intuition that if the algorithm does not query a matched index, then it 

We need a corollary of Lemma \ref{hhhhheeeee1}.
Similar to Lemma \ref{hehehehe3}, it captures
  the minimal information needed about $\pi$ to evaluate $\calF_\pi$ at a given set $S\subseteq V$:
  
\begin{corollary}
Let $S\subseteq V$ and let $\pi,\pi'$ be two shufflings of $T$. 
We have $\mathcal{F}_{ \pi}(S)=\mathcal{F}_{ \pi'}(S)$ when 
  the following condition holds: 
For every two nodes $u,v$ of $T$ such that $S\cap A_u$ and $S\cap A_v$
  are nonempty, the LCA of $\pi^{-1}(u)$ and $\pi^{-1}(v)$
  is at the same depth as the LCA of $\pi'^{-1}(u)$ and $\pi'^{-1}(v)$, both in $T$.
%  For each $i\in [m]$, we have
%  $\pi(i)=\pi'(i)$ if 1) $S\cap B_i\ne \emptyset$ and 2) either $S\cap A_{\pi(i)}\ne \emptyset$
%  or $S\cap A_{\pi'(i)}\ne \emptyset$.
\end{corollary} 
\begin{proof}
Let $P_1=\{\pi^{-1}(u):\text{$u\in U$ and $S\cap A_u\ne \emptyset$}\}$ and $P_2$
  be be the set of nodes defined similarly using $\pi'$.
It follows from Lemma \ref{hhhhheeeee1} that $\bR\cap P_1$ is the same
  as $\bR\cap P_2$, $\bR\sim \calD$, under the natural bijection between $P_1$ and $P_2$.
The statement follows from the definition of $\calF$.
\end{proof}

\begin{remark}\label{heheremark}
To evaluate $\calF_\pi(S)$, it suffices to know the LCA of $\pi^{-1}(u)$
  and $\pi^{-1}(v)$ for every $u,v$ with
  $S\cap A_u\ne \emptyset$ and $S\cap A_v\ne \emptyset$. 
The  LCA of $\pi^{-1}(u)$ and $\pi^{-1}(v)$ can be determined as follows.
\begin{flushleft}\begin{enumerate}
\item First consider the case when $u$ is a prefix of $v$ or $v$ is a prefix of $u$.
Let $u=(u_1,\ldots,u_\ell)$ and $v=(u_1,\ldots,u_\ell,v_{\ell+1},\ldots)$. (The case when $v$
  is a prefix is similar.)
Then the depth of LCA of $\pi (u)$ and $\pi^{-1}(v)$ is either $\ell$ if $\pi _{u_1,\ldots,u_{\ell-1}}(u_\ell)=u_\ell$,
  or $\ell-1$ otherwise.
  
\item Assume that $u=(u_1,\ldots,u_\ell,u_{\ell+1},\ldots)$ and $v=(v_1,\ldots,v_\ell,v_{\ell+1},\ldots)$ with $u_1=v_1,\ldots,u_\ell=v_\ell$ but $u_{\ell+1}\ne v_{\ell+1}$.
We have three subcases. If both $u$ and $v$ have length strictly longer than $\ell+1$, then the
  depth of LCA of of $\pi^{-1}(u)$
  and $\pi^{-1}(v)$ is $\ell$.
If both $u$ and $v$ have length $\ell+1$, then the depth of LCA of $\pi^{-1}(u)$
  and $\pi^{-1}(v)$ is also $\ell$.
(So in these two subcases we we do not need to know anything about $\pi$.)
Finally, if $u$ has length $\ell+1$ and $v$ has length longer than $\ell+1$, then
  the depth of LCA of $\pi^{-1}(u)$
  and $\pi^{-1}(v)$ is $\ell+1$ if $\pi_{v_1,\ldots,v_\ell}(v_{\ell+1})=u_{\ell+1}$,
  and is $\ell$ otherwise.
The case when $u$ has length longer than $\ell+1$ and $v$ has length $\ell+1$ is similar.
\end{enumerate}\end{flushleft}
To summarize, to determine the depth we only
  need to know whether a particular entry of $\pi$ is equal to a certain value or not.
Moreover, the entry is either $\pi_{u_1,\ldots,u_{\ell-1}}(u_{\ell})$
  or $\pi_{v_1,\ldots,v_{\ell-1}}(v_\ell)$ for some $\ell$.
\end{remark}

\subsection{Lower bound for dynamic submodular maximization}
\label{sec:lower2-lower}

\noindent\textbf{Choices of parameters.} 
Let $L=1/\eps$ be a positive integer. 
Let $k$ be a positive integer such that $k^2\le n^\eps$ (as otherwise the 
  lower bound we aim for becomes trivial) and $w=\eps k$ is a positive integer. 
Let $n$ be the length of dynamic streams and $d=n^\eps$. 
Set $m_L=1$ and %For each $\ell \in [L]$, %$d = n^{\eps}$ and 
  $$m_{\ell} = \frac{ n^{(L-\ell + 1)\eps}}{2k},\quad\text{for each $\ell\in [L-1]$}.$$ 
%For any $(i_1,\ldots, i_{L}) \in M_{L}$, we say it is a path 

\noindent\textbf{Hard streams.}
The ground set is $V$, which is known to the algorithm; the monotone submodular function 
  $\calF^*:2^V\rightarrow \mathbb{R}$ is $\calF_\pi$,  where $\pi$ is a shuffling
  of $T$, which is unknown to the algorithm. 
The stream is constructed by running \textsc{Tranverse} in Algorithm \ref{algo:dfs},
  and is independent of $\pi$.
It can be viewed as a DFS over the tree $T$, starting at the root,
  except that every time it reaches a node $v$,
  it inserts all children of $v$ but only explores the first $d$ children,
  and deletes all children of $v$ at the end of the exploration. 

We first bound the total number of operations by $n$. 
For each $\ell\in [0:L-2]$, \textsc{Tranverse} visits $d^\ell$ nodes and
  creates $2m_{\ell+1}\cdot w$ operations for each of them.
\textsc{Tranverse} visits $d^{L-1}$ nodes at depth $L-1$ and creates 
  $2w$ operations for each of them.
%The total number of insertions/deletions for the $\ell$-level equals. 
%The \textsc{tranverse} procedure inserts elements of the $\ell$-th level node for $n_1\cdots n_{\ell-1}$ times. Since the number of elements in a $\ell$-th level node equals $m_{\ell}\cdot (k/\eps)$, the total number of insertion operations equals $n_1\cdots n_{\ell-1}m_{\ell} \cdot \frac{k}{\eps}$. The same holds for deletion operations. 
Hence,  the total number of operations  is
\[
\sum_{\ell=0}^{L-2} d^\ell\cdot m_{\ell+1}\cdot 2w +d^{L-1}\cdot 2w\le   \sum_{\ell=0}^{L-2}
\eps n+ \eps n\le n.
%n^{(\ell - 1)\eps}\cdot \frac{\eps^2}{2k}  n^{(L-\ell+1)\eps} \cdot 2(k/\eps) = n
\]

%  In the stream we first insert all elements of A, which takes αmkw insertions. We then divide the rest of the stream into m batches: For the t-th batch, we first insert all elements of Bt and then delete them. So the total length is n = (2 − α)mkw. This is the only stream we will use in the proof but the function is determined by the unknown c and π.

% We view the ground set as a depth $L$ tree. The root node contains all element in $A_{1}, \ldots, A_{m_1}$, and it has $m_1$ child nodes. Generally, in the $\ell$-th level, there are $m_1\cdots m_{\ell-1}$ nodes and each node has $m_{\ell}$ child nodes. For any  $(i_1,\ldots, i_{\ell-1})\in M_{\ell-1}$, we refer the node $(i_1, \ldots, i_{\ell-1})$ as the one reached by choosing the $i_{\tau}$-th child nodes ($\tau =1,2,\ldots, \ell-1$) in $\tau$-th level. The node $(i_1, \ldots, i_{\ell-1})$ contains all elements in $A_{i_1, \ldots, i_{\ell-1}, i}$ ($i \in [m_{\ell}]$). 
%The hard instance is constructed in Algorithm~\ref{algo:dfs}, it can be seen as a random DFS over the depth $L$ tree. The DFS recurses at a level $\ell$ node for $n_{\ell}$ times before going back the level $(\ell - 1)$-th levels. 

To gain some intuition behind the stream,
  we note that leafs of $T$ that appear in the stream~are exactly those in $U_L^*:=[d]^{L-1}\times \{1\}$,
  and they appear in the stream under the lexicographical~order (which we will denote by $\prec$).
For each such leaf $u\in U_L^*$, at the time when the set $A_u$ was~inserted, the 
  current set of elements is $W_u$.
Inspired by Corollary \ref{simplecoro1} 
  we will consider the following simplified $d^{L-1}$-stage dynamic problem, 
  where  stages are indexed using leaves in $U^*_L$ under the lexicographical order.
%Let $\smash{u^1,u^2,\ldots,u^{d^{L-1}}}$ be the lexicographical ordering of $[d]^{L+1}\times \{1\}$.
During the $u$-th stage, an algorithm can query $\calF^*$
  about any subset $S\subseteq \cup_{u'\in U_L^*:u'\preceq u} W_{u'}$~and 
  can choose to start the next stage at any time by returning an $S_u\subseteq W_{u }$
  of size at most $k$.~We say an algorithm succeeds if $\smash{\calF^*(S_u)\ge 0.5+\tilde{O}(\eps)}$ as
   on the RHS of (\ref{heheeqeq}) for every stage.
We show below that any randomized algorithm that succeeds with probability at least $2/3$ must have
  total query complexity $\Omega(n^{1+\eps}/k^3)$.
It follows from Corollary \ref{simplecoro1} that any randomized algorithm for dynamic submodular maximization
  with an approximation guarantee of $0.5+\tilde{\Omega}(\eps)$ must have amortized query
    complexity $\Omega(n^\eps/k^3)$.

%Let $v_1,\ldots$ be the node at depth $L$.
%We care about the time when each $v_i$ was just inserted.
%For any $(i_1, \ldots, i_{L}) \in M$, it corresponds to a path on the tree and we denote $W_{i_1, \cdots, i_{L}} = \bigcup_{\ell \in [L]}\bigcup_{j \in [m_{\ell}]} A_{i_1, \cdots, i_{\ell-1},j}$ to contain all element on the path.
%Our hard instance guarantees that at any time step, the element in hand is a subset of the elements along some paths.

\begin{algorithm}[!t]
	\caption{\textsc{Tranverse}}
	\label{algo:dfs}
	\begin{algorithmic}[1]
	\State \textbf{Input: }A node $u$ of $T$ at depth $\ell\in [0:L-1]$
 	\If{$\ell=L-1$}
	 	\State Insert all elements of $A_{u_1,\ldots,u_{L-1},1}$ and then delete them
	\Else
		
			%\Procedure{\textsc{}}{u} \Comment $\ell \in [L], (i_1, \ldots, i_{\ell- 1})\in M_{\ell-1}$
			%\If{$\ell = L$}
			%\State Insert all elements in $A_{i_1,\cdots, i_{L-1}, i}$ ($i \in [m_{L}]$) and then delete them
			%\Else
			\State Insert all elements in $A_{u_1,\cdots, u_{\ell}, i}$ for each $i \in [m_{\ell+1}]$
			\For{$i$ from $1$ to $d$}
			\State  Call \textsc{Tranverse}  on $(u_1,\ldots, u_\ell,i)$
			 \EndFor
			\State Delete all elements in $A_{u_1,\cdots, u_{\ell}, i}$ for each $i \in [m_{\ell+1}]$
		\EndIf
			%\EndProcedure
	\end{algorithmic}
\end{algorithm}

\begin{proof}[Proof of Theorem~\ref{thm:lower2}]
Let $\bpi$ be a shuffling of $T$ drawn uniformly at random (i.e., every bijection~in $\bpi$
  is drawn independently and uniformly).
Consider any deterministic algorithm $\ALGG$ that succeeds with probability $2/3$  
  on the $d^{L-1}$-stage dynamic problem described above with $\calF^*=\calF_{\bpi}$.
  %must have a total query complexity of $???$.
%Let $\ALGG$ be such a deterministic algorithm.
By Lemma \ref{lem:basicbasic} we assume without loss of generality that 
  $\ALGG$ only queries sets of size at most $k/\eps$.
When $\ALGG$ succeeds on $\calF^*=\calF_\pi$ for some shuffling $\pi$,
  we have from Corollary \ref{simplecoro1} that the $S_u$ it outputs during the $u$-th stage
 must satisfy
$ 
\smash{S_u\cap \calA_{u }^\pi\ne \emptyset}
$, which implies (using the definition of $\calA_u^\pi$) that 
\begin{equation}\label{hehej}
S_u\cap A_{u_1,\ldots,u _\ell ,\pi_{u_1,\ldots,u_\ell}(u_{\ell+1})}\ne \emptyset
\end{equation}
for some $\ell \in [0: L-2]$.
By an averaging argument, we have that there exists an $\ell\in [0:L-2]$ such that
  with probability at least $2/(3L)$ over $\bpi$, at least $(1/L)$-fraction of $S_u$ returned by 
  $\ALGG$ satisfy (\ref{hehej}) for this $\ell$.
Fix such an $\ell$ and this inspires us to introduce the following simple game.

In the game there are $d^\ell$ hidden bijections $\pi_{v_1,\ldots,v_\ell}:[m_{\ell+1}]\rightarrow
  [m_{\ell+1}]$, for each $v_1,\ldots,v_\ell\in [d]$.
The game consists of $d^{\ell+1}$ stages;  each stage is indexed by a $v=(v_1,\ldots,v_{\ell+1})\in 
  [d]^{\ell+1}$ and ordered by the lexicographical order.
During the $v$-th stage of the game, 
  an algorithm can send a number $i\in [m_{\ell+1}]$ to the oracle and the latter
  reveals if $\pi_{v_1,\ldots,v_\ell}(v_{\ell+1}) =i$.
We say the algorithm wins~the $v$-th stage if it queries an $i$ that matches 
  $\pi_{v_1,\ldots,v_\ell}(v_{\ell+1})$ during the $v$-th stage.
At any time it can choose to give up and move
  forward to the next stage, in which case $\pi_{v_1,\ldots,v_\ell}(v_{\ell+1})$ is revealed to the algorithm. 
We say an algorithm succeeds if it wins at least $(1/L)$-fraction of the $d^{\ell+1}$ stages.

We prove the following lower bound for this game in Appendix~\ref{sec:lower2-app}:

\begin{claim}\label{blablaclaim}
When the hidden bijections are drawn independently and uniformly, any
  deterministic algorithm that succeeds with probability at least $2/(3L)$   
  has total query complexity $\Omega(m_{\ell+1}d^{\ell+1})$.
\end{claim}

To finish the proof, we show that $\ALGG$ can be used to play the game as follows:
\begin{flushleft}\begin{enumerate}
\item We start by drawing a random bijection for every node in the tree $T$
  except for those at $(v_1,\ldots,v_\ell)\in [d]^\ell$.
Let $\pi$ be the shuffling when they are combined with hidden 
  bijections $\pi_{v_1,\ldots,v_\ell}$ in the game.
We simulate $\ALGG$ on $\calF_\pi$ over the stream  
  \textsc{Tranverse} and maintain\\ the following invariant.
During the $v$-th stage of the game, with $v=(v_1,\ldots,v_{\ell+1})\in [d]^{\ell+1}$,
  we simulate $\ALGG$ through its stages for leaves $u\in U_L^*$ that 
  have $v$ as a prefix 
  and assume\\ that we already know 
  $\smash{\pi_{w_1 ,\ldots,w_\ell }(w_{\ell+1} )}$   
  for all $\smash{w=(w_1 ,\ldots,w_{\ell+1} )\prec v}$ and $w\in [d]^{\ell+1}$.

\item During the $u$-th stage of the simulation of $\ALGG$ for some leaf $u\in U_L^*$, %, where $u$ is a leaf and $u=(v_1,\ldots,v_{\ell+1})\in [d]^{\ell+1}$,
  we are in the $v$-th stage of the game with $v=(u_1,\ldots,u_{\ell+1})\in [d]^{\ell+1}$.
For each query $S\subseteq \cup_{u'\in U_L^*:u'\preceq u}{W_{u'}}$ (of size at most $k/\eps$) made by $\ALGG$,
  it follows from Remark \ref{heheremark} that, to evaluate $S$ at $\calF_\pi$,
  we only need to know the depth of LCA of
  $\pi^{-1}(u'')$ and $\pi^{-1}(v'')$ in $T$ for no more than $(k/\eps)^2$ many pairs of $u'',v''$ 
  with $S\cap A_{u''}\ne \emptyset$
  and $S\cap A_{v''}\ne \emptyset$.
For each such pair, it follows from Remark \ref{heheremark} again that either we know the answer already or we need to compare
  $\smash{\pi_{u_1'',\ldots,u_\ell''}(u_{\ell+1}'')}$ or $\smash{\pi_{v_1'',\ldots,v_\ell''}(v_{\ell+1}'')}$ with a 
  certain value.
Given that $u''\preceq v$ and $v''\preceq v$,
  either the answer is already known or we just need to make a query to the game oracle.
%Also note that the only nontrivial case is when $u$ is an ancestor of $v$ (or vice versa).
%Let $u=u_1\ldots u_p$ and $v=u_1\ldots u_p v_{p+1}\ldots$.
%This is only about if $\pi$ ... and can be answered by making $(k/\eps)^2$ queries
%  on the oracle of the game.
So we only need to make at most $(k/\eps)^2$ queries to continue the simulation of $\ALGG$.
At the end of the $u$-th stage of $\ALGG$, let $S_u$ be the set that $\ALGG$ returns.
For every $i$ such that $S_u\cap A_{u_1,\ldots,u_\ell,i}\ne \emptyset$
   (note that there are at most $k$ such $i$ given that $|S|\le k$),
   we query $i$ on the game oracle.
Note that we would have won the $v$-th stage of the game by now if  (\ref{hehej}) holds  
  for $S_u$.
%$$
%S_u\cap A_{u_1,\ldots,u_\ell,\pi_{u_1,\ldots,u_\ell}(u_{\ell+1})}\ne \emptyset.
%$$
We then continue to simulate $\ALGG$ on the next stage of the dynamic problem.
If the next stage of $\ALGG$ is about a new leaf $u$ with 
  $v\prec (u_1,\ldots,u_{\ell+1})$, then we also move to the next stage in the game.
%\item When $\ALGG$ finishes the $u^*$-th stage with $u^*=(u_1,\ldots,u_\ell,u_{\ell+1},d\ldots,d)$,
%  we collect all $S_u$ with $u=(u_1,\ldots,u_\ell,u_{\ell+1},*,\ldots,*)$ and 
%  for each $S_u$ query the oracle of the game $k$ times.
%We then move to the next stage of the game.
%We won the $u^*$-stage if we queried $\pi_{u_1,\ldots,u_\ell}(u_{\ell+1})$;
%  otherwise we failed this stage and received it from the game for free.
\end{enumerate}\end{flushleft}
It is clear from the simulation that for any $\pi$, if $\ALGG$ running on 
  $\calF_\pi$ satisfies (\ref{hehej}) for at least $(1/L)$-fraction of leaves in $U_L^*$,
  then we win the game when $\pi$ is the shuffling we get by combining our own
  random samples at the beginning with hidden bijections in the game.
Using the promise about $\ALGG$ at the beginning, we get a randomized algorithm that succeeds in the game with probability
  at least $2/(3L)$ when the hidden bijections are drawn independently and uniformly at random.
On the other hand, if the query complexity of $\ALGG$ is $q$, then our simulation uses
$$
q\cdot \left(\frac{k}{\eps}\right)^2+d^{L-1}\cdot k
$$
Combining with Claim \ref{blablaclaim} we have $q=\Omega(n^{1+\eps}/k^3)$.
\end{proof}

\section{Insertion-only streams under a cardinality constraint}%: near optimal algorithm for cardinality constaints}
\label{sec:insert-cardinality}

We consider insertion-only streams and give a deterministic
  $(1-1/e -\eps)$-approximation algorithm with $O(\log (k/\eps) /\eps^2)$ amortized query complexity. As discussed in Remark \ref{detailremark}, our algorithm does not need to know the 
  ground set $V$ or the number of rounds $n$ at the beginning. 

\thmupperone*

%\begin{theorem}
%\label{thm:insert-cardinality}
%Given any $\eps>0$,
%there is a deterministic algorithm that achieves an approximation guarantee of $ 1-1/e -\eps $
%  for dynamic submodular maximization under a cardinality constraint over insertion-only streams.
%The amortized query complexity of the algorithm is $O(\log (k/\eps)/\eps^{2})$.
%that maintains a feasible set $S$ with $(1-1/e -\eps)$-approximation at each iteration. Moreover, the amortized number of queries per update is $O(\log k/\eps^2)$.
%\end{theorem}

To prove Theorem \ref{thm:insert-cardinality}, we give a deterministic algorithm with the following performance guarantees.
%In addition to the cardinality $k$ and the approximation parameter $\eps$, 
%  the algorithm is given a positive number $\OPT$ at the beginning and outputs a set 
%  $S_t\subseteq V_t$ of size at most $k$ at the end of each round $t$
%  such that (1) $S_1\subseteq S_2\subseteq \cdots$ and (2) when $\OPT_t\ge \OPT$ for the first time, 
%  the set $S_t$ must satisfy $f(S_t)\ge (1-1/e-\eps)\OPT$. % at the end of the $t$-th round.
%  whenever $\OPT_t\ge \OPT$ (and no condition is required on $S_t$~if $\OPT_t<\OPT$).
We %prove the following lemma in this section, 
   follow standard arguments to finish the proof of Theorem \ref{thm:insert-cardinality} in Appendix \ref{sec:insert-cardinality-app}.
\begin{lemma}
There is a deterministic algorithm that
  satisfies the following performance guarantees.
Given a positive integer $k$, $\eps>0$ and $\OPT>0$,
  the algorithm runs on an insertion-only stream and outputs
  a set $S_t\subseteq V_t$ of size
  at most $k$ at the end of each round $t$ such that 
  (1) $S_1\subseteq S_2\subseteq \cdots$ and (2) when $\OPT_t\ge \OPT$ for the first time, 
  $S_t$ must satisfy $f(S_t)\ge (1-1/e-\eps)\OPT$. 
The amortized query complexity of the algorithm is $O(1/\eps)$.
%	Let $t$ be the smallest integer satisfying $\OPT_{t} \geq \OPT$. After inserting the $t$-th element, the Algorithm~\ref{algo:insert-cardinality} outputs a set $S$ of cardinality at most $k$ and satisfies $f(S)\geq (1-1/e-\eps)\OPT$. Moreover, the amortized number of oracles per operation is $O(1/\eps)$.
\end{lemma}

%For simplicity, we assume $\OPT$\footnote{Xi: Describe the meaning of $\OPT$: If $\OPT_t$ is close to $\OPT$ then
%  the algorithm succeeds in round $t$.} is known to the algorithm and we remove this assumption in Appendix~\ref{sec:opt-app} with $O(\eps^{-1}\log k)$ overhead.

\begin{algorithm}[!h]
	\caption{Dynamic submodular maximization with a cardinality constraint.}
	\label{algo:insert-cardinality}
	\begin{algorithmic}[1]
	  \Procedure{\textsc{Initialize}}{$k,\eps,\OPT$}
	  %\State Let $\Delta=\eps \OPT/(4k)$
	  \State Initialize $S = \emptyset$, $\Delta = \eps\OPT/k$ and $B_{\ell} = \emptyset$ for each $\ell =0,1,\ldots, \lfloor 1/\eps \rfloor$	  \EndProcedure\\
	  
	  \Procedure{\textsc{Insert}}{$e$}

			\If{$f_{S}(e) \geq ( \OPT - f(S))/k - \Delta$ \textbf{and} $|S| < k$} \label{line:eva1}
				\State Update $S \leftarrow S\cup \{e\}$ and call \textsc{Revoke}   \label{line:update1}
			\Else
				\State Update $B_{\ell} \leftarrow B_{\ell} \cup \{e\}$ with $\ell = \lfloor  f_{S}(e)/\Delta \rfloor$\hfill \Comment{$\ell\le \lfloor 1/\eps\rfloor$ given $f_{S}(e)<\OPT/k$}
			\EndIf
	\EndProcedure\\
	
	\Procedure{\textsc{Revoke}}{}
	\State Let $r = \lfloor  (\OPT - f(S))/(k\Delta) \rfloor$. \label{line:revoke}
	\If{$|S|<k$ and there exists an index $\ell \geq r$ with $B_{\ell }\neq \emptyset$} 
	\State Let $\ell $ be any such index and let $e'$ be any element in $B_{\ell }$.
			\If{$f_{S}(e') \geq  (\OPT-f(S))/k-\Delta $}  \label{line:eva2}
			\State Update $S \leftarrow S\cup \{e'\}$ and $B_{\ell } \leftarrow B_{\ell } \backslash \{e'\}$\label{line:update2}
			\Else
			\State Update $B_{\ell } \leftarrow B_{\ell } \backslash \{e'\}$, $B_{\ell' } \leftarrow B_{\ell' }\cup \{e'\}$ with $\ell'  = \lfloor f_{S}(e')/\Delta \rfloor $ 
			\Comment{$0\le \ell'<r\le \ell$}
			\EndIf
			\State Go to Line~\ref{line:revoke}.
	\EndIf
	\EndProcedure
	\end{algorithmic}
\end{algorithm}

\begin{proof}
The algorithm is described in Algorithm~\ref{algo:insert-cardinality} with $\Delta:=\eps\OPT/(k)$.
We first run \textsc{Initialize}  and then 
  run \textsc{Insert}$(e_t)$ when the $t$-th element $e_t$ is 
  inserted in the stream.
It outputs $S$ at the end of \textsc{Insert}$(e_t)$ as $S_t$.
It is clear that $S_1\subseteq S_2\subseteq \cdots$ given that 
  the algorithm only grows $S$ by adding elements to it.
 
 We prove the approximation guarantee of the algorithm. 
 Let $t$ be the first time $\OPT_t\ge$ $ \OPT$.
We first consider the case when $|S_t| = k$. 
Let $S_t=\{s_1,\ldots,s_k\}$ where $s_i$ is the $i$-th element inserted to $S$.
Let $T_i=\{s_1,\ldots,s_i\}$ for each $i\in [0:k]$ (so $T_k=S_t$ and $T_0=\emptyset$
  with $f(T_0)=0$).
Given that 
$$
f_{S}(e)\ge \frac{\OPT-f(S)}{k}-\Delta,
$$
whenever the algorithm adds an element $e$ to $S$ (line~\ref{line:update1} and line~\ref{line:update2}),
we have
%Let $s_i$ be the $i$-th element inserted to $S_t$, and $S_i = \{s_1, \cdots, s_i\}$. The element $s_i$ is added to the solution set $S$ only if it satisfies
	\[
	f_{T_i}(s_{i+1}) \geq \frac{\OPT - f(T_i)}{k} - \Delta
	\]
for each $i\in [k]$. By standard calculation, one has
	\begin{align*}
	\OPT - f(S_t) = \OPT-f(T_k) &=  \OPT - f(T_{k-1}) - f_{T_{k-1}}(s_k ) \\[0.8ex]
	&\leq    \OPT - f(T_{k-1}) - \frac{\OPT - f(T_{k-1})}{k} + \Delta\\
	&=   \left(1 - \frac{1}{k}\right)\big(\OPT - f(T_{k-1})\big) + \Delta\\
	&\hspace{1cm}\vdots   \\
	&\leq   \left(1 - \frac{1}{k}\right)^k \big(\OPT-f(T_0)\big) + k\Delta \\[0.5ex]
	&\le   \frac{1}{e} \cdot \OPT+\eps\cdot \OPT.
	\end{align*}
which implies that $f(S_t)\ge (1-1/e-\eps)\OPT$. 
%	\[
%	f(S) \geq \left(1 - \left(1- \frac{1}{k}\right)^k\right)\OPT- \frac{\eps}{4}\OPT \geq (1 - 1/e - \eps/4)\OPT.
%	\]
	
	Next  consider the case when $|S_t| < k$. In this case we have that 
	  every element inserted is either in $S_t$ or in one of the buckets $B_\ell$.
	Let the optimal solution be $O_t = \{o_1, \ldots, o_k\}$ and let
$$
r=\left\lfloor \frac{\OPT-f(S_t)}{k\Delta}\right\rfloor.
$$ Then at the end of round $t$ (after \textsc{Insert}),
  we must have that every $o_i$ with $o_i\notin S_t$ lies in a bucket $B_\ell$
  with $\ell<r$ and thus, 
	\begin{align}
	\label{eq:cardinality}
	f_{S_t}(o_i)<(\ell+1)\Delta\le r\Delta \le%\max\left\{ r\Delta, \Delta \right\}\le 
	 \frac{\OPT-f(S_t)}{k} .
	\end{align}
This is because the algorithm makes sure that at any time, any element $e'\in B_\ell$
  satisfies $f_{S}(e')< (\ell+1)\Delta$.
To see this is the case, we note that the inequality holds when $e'$ is added to $B_\ell$
  and that $S$ only grows over time.
%	This is because $o_i$ is located in bucket\footnote{Need to prove something like this by induction first.} $B_{\ell}$ with $\ell < \max\{\lceil \frac{4k}{\eps \OPT}\cdot \frac{\OPT - f(S)}{k} \rceil,2\}$.
	We conclude that $f(S_t) \geq (1 - \eps)\OPT$; otherwise, one has
	\[
	\OPT \leq f(O_t) \leq f(S_t) + \sum_{i \in[k]} f_{S_t}(o_i ) < f(S_t) + \sum_{i=1}^{k}\frac{\OPT - f(S_t)}{k}  = \OPT,
	\]
	a contradiction.
	The second step above follows from the submodularity and monotonicity, the third step follows from Eq.~\eqref{eq:cardinality} and $f(S_t)< (1 - \eps)\OPT$.
	
Finally we bound the amortized query complexity of the algorithm.
We charge the two queries made in the evaluation of $f_{S}(e)$ (line \ref{line:eva1} or \ref{line:eva2}) to $e$ and 
  show that the number of queries charged to $e$ is at most $O(1/\eps)$.
To see this, we note that $e$ is charged twice when it is just inserted.
Every time $e$ is charged during \textsc{Revoke}, either it is 
  added to $S$ so that it is never charged again, or it gets moved to a bucket
  with a strictly smaller $\ell$, which can only happen $O(1/\eps)$ many times.
% the total number of queries. The \textsc{insert} procedure requires $2$ queries per element. For the \textsc{revoke} procedure, each time it queries the marginal value of an element, either the element is added to $S$ afterwards, this requires 2 queries and can happen at most $k$ times; or it moves the element to lower level bucket (that is, decrease the associated bucket index by at least $1$), this can happend at most $O(1/\eps)$ times.
	%. If no elem To bound the number of query oracle there, notice that each call either succeed, this add a new element, and can happen at most $k$ times, or decrease the marginal value to the next level, which can happen at most $1/\eps$ times. We conclude the proof here.
\end{proof}

% !TEX root =  main.tex

\section{Insertion-only stream: efficient algorithm for matroid constaints}
\label{sec:insert-matroid}

We present an efficient $(1 - 1/e -\eps)$-approximation algorithm under the general matroid constraint. 
We first give a $(1/2 -\eps)$-approximate  deterministic combinatorial algorithm (Section~\ref{sec:com}), we then embed it into the accelerated continuous greedy framework of \cite{badanidiyuru2014fast} to achieve  $(1 - 1/e - \eps)$ approximation (Section~\ref{sec:accelerate}).
%In particular, we use the combinatorial algorithm to select a set $S$ as the direction of improvement at each step for accelerated continous greedy.%\footnote{Xi: Here $k$ is the rank of the matroid?}

\subsection{The combinatorial algorithm}
\label{sec:com}

Given $k$ as the rank of the matroid $\calM$, we define
  $L,R$ and $\calA$ as
\[
L = \left\lceil \frac{\log (k/\eps)}{\eps} \right\rceil, \quad
R=\left\lceil \frac{2\log (k/\eps)}{\eps^2}\right\rceil \quad\text{and}\quad
\mathcal{A} = \left\{(a_1, \cdots, a_{L})\in \mathbb{Z}_{\ge 0}^{L}: \sum_{\ell\in [L]}a_{\ell} \leq  R    \right\}.
\]
We note the size of $\mathcal{A}$ can be upper bounded by
\begin{align}
\label{eq:branch_size}
|\mathcal{A}|= \sum_{d=0}^{R} \binom{d+L-1}{L-1} \leq (R+1)\cdot \binom{R  + L-1}{L-1} 
\le (R+1)\left(\frac{2eR}{L}\right)^L
%\lesssim  \left(\frac{e}{\eps}\right) ^{\log(k/\eps)/\eps} = 
=k^{\tilde{O}(1/\eps)}.
\end{align}

\def\ALG2{\textsc{Prune-greedy}}
\def\tildef{\tilde{f}}

Both of our deterministic algorithm and randomized algorithm use a deterministic
  subroutine called \ALG2 described in Algorithm \ref{alg:alg2}.
The inputs of \ALG2 include $k$ as the rank of the underlying matroid $\calM$,
  a positive number $\OPT$, a tuple $a\in \calA$, as well as query access to both $f$ and $\calM$.
When running on an insertion-only stream, \ALG2 may decide to terminate
  at the end of a round and output a set $S$.
A complication due to the application of this subroutine in the randomized algorithm
  is that we will give it query access to a perturbed version of $f$:
We say $h:V\rightarrow \mathbb{R}$ is a \emph{$\kappa$-close} of $f$ if $h(S)=f(S)\pm \kappa$
  for every $S\subseteq V$.

We state its performance guarantees in the following lemma:

\begin{algorithm}[h]
	\caption{\textsc{Prune-greedy}}\label{alg:alg2}
	\label{algo:greedy}
	\begin{algorithmic}[1]
		 \Procedure{\textsc{Initialize}}{$k,\OPT,a,h,\calM$} \Comment{$a \in \mathcal{A}$}
		\State Initialize $S \leftarrow \emptyset$ and $c_{\ell} \leftarrow a_{\ell}\Delta $ for each $ \ell \in [L]$, where $$\Delta:=\frac{\eps^2\OPT}{\log (k/\eps)}.$$  
		\State Let $\ell^*$ be the smallest $\ell\in [L]$ with $c_{\ell}>0$
		 \EndProcedure\\
			
			  \Procedure{\textsc{Insert}}{$e$}
			  %\State Find the smallest $\ell$ such that $c_{\ell} > 0$
			  \If{$S \cup e \in \mathcal{M}$ and $h_{S }(e)  \geq (1+\eps)^{-\ell^*}\OPT$}\label{line:eva3}
			  \State Update $c_{\ell^*} \leftarrow c_{\ell^*} - h_{S }(e)$ and $S \leftarrow S \cup \{e\}$
			  %\State Call \textsc{Revoke} \textbf{if} $c_\ell\le 0$
			  \If{$c_{\ell^*}\le 0$}
			  	\State Call \textsc{Revoke}
			    %\If{$c_\ell=0$ for all $\ell>\ell^*$}
			   % 	\State Terminate and output $S$
				%\Else
			%		\State Update $\ell^*$ to be the smallest $\ell>\ell^*$ with $c_\ell>0$, and
 %call \textsc{Revoke}
%				\EndIf   
			\EndIf
			  \EndIf
			  \EndProcedure\\
			  
			  	\Procedure{\textsc{Revoke}}{} 
			  \State Update $\ell^*$ to be the smallest $\ell$ with $c_{\ell} > 0$;
			  \textbf{Terminate} and return $S$ \textbf{if} no such $\ell$ exists
			  %\If{there exists $f_{S }(e) \geq (1+\eps)^{-\ell}\OPT$ and $S \cup e \in \mathcal{M}$}
			  \For{each element $e_i$ inserted so far (in the order of insertions)}
			  \If{$S \cup e_i \in \mathcal{M}$ and $h_{S }(e_i) \geq (1+\eps)^{-\ell^*}\OPT$} \label{line:eva4}
			  \State Update $c_{\ell^*} \leftarrow c_{\ell^*} - h_{S }(e_i)$ and $S  \leftarrow S \cup \{e_i\}$, 
			  %\State \textsc{Revoke}(S) 
			  \EndIf
			  \EndFor
			  \State Go to line 16 \textbf{if} $c_\ell^*\le 0$ 
			  	\EndProcedure
	\end{algorithmic}
\end{algorithm}

\begin{lemma}
	\label{lem:matroid-approx}
There is a deterministic algorithm that, given a positive integer $k$, $\OPT>0$,~$a\in \calA$ and 
  query access to a matroid $\calM$ over $V$ of rank $k$ and a 
  function $h:V\rightarrow \mathbb{R}$ that is $\kappa$-close to 
  a nonnegative and monotone submodular function $f:V\rightarrow \mathbb{R}$, where $\kappa=\eps^3\OPT/k$. 
The~algorithm~runs on insertion-only streams
  with amortized query complexity $O(L)$ and has the following performance guarantee.
Given any insertion-only stream $e_1,\ldots,e_t$ such that 
%	Let $t > 0$, 
  $\OPT\leq \OPT_{t} \leq$ $ (1+\eps)\OPT$,
  there exists an $a^*\in \calA$ such that 
  when given $\OPT$ and $a^*$ as input, the algorithm terminates before the end of round $t$ and outputs a feasible set 
  $S\subseteq V_t$ that satisfies
%   After inserting the $t$-th element, there exists a branch $a \in \mathcal{A}$ whose output set $S_a$ guarantees
	\[
	f(S )\geq (1-O(\eps))f_{S }(O),\quad\text{for any set $O$ such that $O\in \calM$ and 
	$O\subseteq V_t$.}
	\]
\end{lemma}
\begin{proof}
The algorithm is described in Algorithm \ref{alg:alg2}, with $$\Delta:=\frac{\eps^2\OPT}{\log (k/\eps)}.$$
%We will place elements into $L$ buckets $B_1,\ldots,B_L$ according to their marginal gains, where
%$$$$
%Let $B_{\ell} = [ (1+\eps)^{-\ell}\OPT,  (1+\eps)^{-\ell + 1}\OPT]$ ($\ell \in [L]$) be the $\ell$-th level. 

%Given an estimate of $\OPT$, for each sequence $a \in \mathcal{A}$, our algorithm instantiates a branch of Algorithm~\ref{algo:greedy} and runs them separately. At each time step, our algorithm segregates all solution sets $S_a$ and returns the one with the maximum value.

	Let $e_1, \ldots, e_t$ be the stream with $\OPT\le \OPT_t\le (1+\eps)\OPT$. 
	To specify the $a^*\in \calA$ in the statement of the lemma, we consider 
	 the following $L$-pass greedy algorithm. 
The algorithm maintains a set $T\in \calM$. 
It starts with $T=\emptyset$ and updates $T\rightarrow T\cup T_\ell$
  at the end of the $\ell$-th pass (so we have $T=T_1\cup \cdots \cup T_\ell$ at the end of 
    the $\ell$-th pass).
During the $\ell$-th pass, we set $S_\ell=\emptyset$ and 
  go through $e_1,\ldots,e_t$.
For each $e_i$, the algorithm checks if $T\cup S_\ell\cup \{e_i\}\in \mathcal{M}$ and 
  $h_{T\cup S_\ell}(e_i)\geq (1+\eps)^{-(\ell-1)}\OPT$.
If so, $e_i$ is added to $S_\ell$.
At the end of the $\ell$-th pass, we do \emph{not} add all elements $S_\ell$ to $T$.  
% $S=\emptyset$ and adds a set 	 
%	 Initially, the solution set $S = \emptyset$. In the $\ell$-th pass ($\ell \in [L]$), the algorithm only accepts feasible element with marginal value greater than $(1+\eps)^{-\ell}\OPT$. That is to say, it goes over $e_1, \cdots, e_t$ in order, and update $S$ only if $S \cup e \in \mathcal{M}$ and $f_{S}(e)\geq (1+\eps)^{-\ell}\OPT$. Let $S_{\ell}$ be all elements added in the $\ell$-th round, 
Instead, we further prune $S_\ell$ to get $T_{\ell}$:
  Let $e_{i_1},e_{i_2},\ldots$ be elements added to $S_\ell$ during the $\ell$-th pass.
  $T_\ell$ is set to be $\{e_{i_1},\ldots,e_{i_j}\}$ such that $j$ is the smallest integer
  such that 
% ($T_{\ell} \subseteq S_{\ell}$) as follows:
%  Let $S_\ell=\{e_{i_1},e_{i_2},\ldots,e_{i_such that $T_{\ell}$ contains the first few element of $S_{\ell}$ and it is the smallest set that satisfies
	\[
	h_{T_{1}\cup \cdots \cup T_{\ell-1}}(T_{\ell}) \geq \Delta \left\lfloor \frac{h_{T_{1}\cup \cdots \cup T_{\ell-1}}(S_\ell)}{\Delta} \right\rfloor,
	\]
where $T_\ell=\emptyset$ when the RHS above is $0$.
This finishes the $\ell$-th pass and the algorithm updates $T$ with $T\cup T_\ell$.
Let $(S_1,\ldots,S_L)$ and $(T_1,\ldots,T_L)$ be the two sequences of sets 
  obtained from this $L$-pass algorithm.	
	Let $a^*\in \mathbb{Z}_{\ge 0}^L$ be defined as 
	$a^{*}_{\ell} = \lfloor h_{T_{1}\cup \cdots \cup T_{\ell-1}}(S_\ell)/\Delta\rfloor$
	for each $\ell\in [L]$. %and let $A_{\ell}^{*} = \sum_{r=1}^{\ell}a_{\ell}^{*}$.
	It is easy to see that
	\begin{align*}
	\Delta\sum_{\ell\in [L]} a^{*}_{\ell} &= \Delta\sum_{\ell\in [L]}  \left\lfloor \frac{h_{T_{1}\cup \cdots \cup T_{\ell-1}}(S_\ell)}{\Delta}\right\rfloor\\
	\leq&~ \sum_{\ell=1}^{L}h_{T_{1}\cup \cdots \cup T_{\ell-1}}(T_\ell)\le f(T)+2\kappa
	<2\OPT.
	\end{align*}
	Hence $\sum_{\ell\in [L]} a_{\ell}^*\le R$ and thus, $a^{*}\in \mathcal{A}$. 
The following lemma connects \textsc{Prune-Greedy} with this $L$-pass greedy algorithm: 

\begin{lemma}\label{hehe2}
Suppose that \textsc{Prune-Greedy} is given $a^*$ at the beginning,
  then it terminates before the end of the $t$-th round and outputs 
  exactly $T=T_1\cup\cdots \cup T_L$. 
\end{lemma}
%the branch $a^{\star}$ returns the solution $S_{a^{\star}}$, the crucial observation here is that  %
%	$S_{a^{\star}} = T_1\cup \cdots \cup T_{L} = T$,
%	i.e., the branch $a^{\star}$ exactly returns the above greedy solution.\footnote{Xi: I think this needs an explanation.} 
	Given Lemma \ref{hehe2}, it suffices to prove that
	\[
	f(T) \geq (1-O(\eps))f_{T}(O),\quad\text{for every feasible set $O \subseteq V_t$.}
	\]
	%\begin{align*}
	%f(T) \geq (1-O(\eps))f_{T}(O_t) -O(\eps)\cdot \OPT \geq &~ (1-O(\eps))f(T\cup O_t) - f(T)-O(\eps)\cdot \OPT \\
	%\geq &~(1-O(\eps))f(O_t) - f(T)-O(\eps)\cdot \OPT 
	%\end{align*}
	%Here $O_t$ is the optimal solution in $V_t$, and hence,
	%\[
	%f(S_{a^{\star}}) = f(T) \geq \left(\frac{1}{2} - O(\eps)\right)\OPT
	%\]	
Fix an $O\in\calM$ and $O\subseteq V_t$.
The following Lemma is a folklore.
	\begin{lemma}
		\label{lem:folklore}
		Let $\mathcal{M}$ be a matroid and $T\in \calM$ with $T=T_1\cup \cdots \cup T_L$ such that 
		$T_1,\ldots,T_L$ are pairwise disjoint.
		Then any $O\in \calM$ can be partitioned into pairwise disjoint $O_1,\ldots,O_L$
		such that 
\begin{flushleft}\begin{enumerate}
\item If $|O|\ge |T|$ then $|O_i|=|T_i|$ for all $i<L$; % and $|O_L|=|O|-\sum_{i<L}|O_i|$;
If $|O|<|T|$, letting $\ell$ be the smallest integer such that $|O|\le \sum_{i\le \ell} |T_\ell|$,
  then $|O_i|=|T_i|$ for all $i<\ell$ and $|O_\ell|=|O|-\sum_{i<\ell} |O_i|$. (Note that we always
  have $|O_i|\le |T_i|$ except for $i=L$.)
\item For all $i<j$, $T_i\cap O_j=\emptyset$ and for every $i<L$,
  $$
  T_1\cup\cdots \cup T_i\cup O_{i+1}\cup\cdots \cup O_L\in \calM.
  $$
\end{enumerate}\end{flushleft}
		
%		, let $U = \{u_1,\cdots, u_k\}$ be any feasible set and let $O$ be of size $k$, there exists an ordering of $O = \{o_1, \cdots, o_k\}$ where for all $i\in [k]$, $T_{:i} \cup O_{i+1:k} \in \mathcal{M}$ and $T_{:i}\cap O_{i+1:k} = \emptyset$.
\end{lemma}

Recall $S_1,\ldots,S_L$ from the $L$-pass greedy algorithm. We have for every $\ell\in [L]$,
$$
T_1\cup \cdots \cup T_{\ell-1}\cup S_\ell\in \calM.
$$
%	For simplicity, we assume $|T| = k$, since we can add a few dummy element with zero value.
%	Denote $O_\ell = \{o_{A_{\ell -1}^{*} + 1}, \cdots, o_{A_{\ell}^{*}}\}$, then by Lemma~\ref{lem:folklore}, we know that $T_1\cup \cdots \cup T_{\ell-1} \cup O_{\ell} \cup \cdots \cup O_{L} \in \mathcal{M}$. Moreover, for any $\ell \in L$, we know $T_1\cup \cdots T_{\ell-1} \cup S_{\ell} \in \mathcal{M}$.
For the analysis we partition $O$ into pairwise disjoint
  sets $E_0,E_1\ldots,E_{L-2},P_1,\ldots,P_L$ as follows.
\begin{flushleft}\begin{enumerate}
\item $P_1=O_1$;
\item For each $\ell\ge 2$ and each $o\in O_\ell$, we consider two cases.
If $T_1\cup \cdots \cup T_{\ell-2}\cup S_{\ell-1}\cup \{o\}\in \calM$
  then we have $o\in P_\ell$;
otherwise, we have $o\in E_r$ where $r\ge 0$ is the largest integer such that 
  $T_1\cup\cdots \cup T_{r-1}\cup S_r\cup \{o\}\in \calM$ and 
  $T_1\cup\cdots \cup T_r\cup S_{r+1}\cup \{o\}\notin \calM$.
Note that $r\le \ell-2$ and such an $r\ge 0$ always exists given that 
  the condition when $r=0$ is just that $\{o\}\in \calM$.
\end{enumerate}\end{flushleft}
%For each $O_\ell$ 
%	For each $\ell \in \{0,1,\cdots, L-2\}$, we define $E_\ell\subseteq O_{\ell+2}\cup\cdots 
%	\cup O_L$ as follows. 
%Each $o \in O_{\ell+2}\cup \cdots\cup O_L$ is in $E_\ell$ if and only if the following
%  two conditions hold:
%$, suppose $o \in O_r$ for some $r\in [L]$, then $o \in E_{\ell}$ if and only if
%	\begin{enumerate}
%		\item $r \geq \ell + 2$, i.e. $o \in O_{\ell + 2} \cup \cdots \cup O_{L}$,
%		\item $T_1\cup \cdots T_{\ell-1}\cup S_{\ell} \cup \{o\} \in \mathcal{M}$
%		  (when $\ell=0$ this condition is $\{o\}\in \calM$ which holds trivially);
%		\item For any $r>\ell$, $T_1\cup \cdots  \cup T_{r -1}\cup S_{r } \cup \{o\} \notin \mathcal{M}$.
%	\end{enumerate}
%Equivalently, taking any element $o\in O_\ell$ for some $\ell\ge 2$,
%  either $T_1\cup \cdots \cup T_{\ell-2}\cup S_{\ell-1}\cup \{o\}\in \calM$ in which case
%  $o$ does not belong to any of $E_0,\ldots,E_{\ell-2}$,
%  or $o\in E_r$ with $r$ being the largest integer (at most $\ell-2$) 
%  $T_1\cup \cdots \cup T_{r-1}\cup S_r\cup \{o\}\in \calM$.
%	 
The following claim about the size of $E_\ell$ follows from the definition, the detailed proof can be found at Appendix~\ref{sec:insert-matroid-app}.
	
\begin{claim}	\label{claim111}
	$|E_{\ell}| \leq |S_{\ell+1} \backslash T_{\ell+1}|$.
\end{claim}

	%to a subset of $O_{\ell}\cup \cdots \cup O_{L}$,  be a set contains element $e$ such that $T_1\cup \cdots T_{\ell-1} \cup S_{\ell} \cup e \in M$ but $T_1\cup \cdots T_{s} \cup S_{s+1} \cup e \notin M$ for any $s \geq \ell$.
	Now we have
	\begin{align}
	f_T(O) %= &~ f(O \cup T) - f(T)\notag\\
	& \leq \sum_{\ell\in [L]}f_{T_{1}\cup \cdots \cup T_{\ell-1} }(O_{\ell})\notag\\
	%&\leq f(O_1)+ \sum_{\ell\in [2:L]}\left(\sum_{r\in [0:\ell-2]}f_{T_{1}\cup \cdots \cup T_{\ell-1} }(O_{\ell} \cap E_{r}) + f_{T_{1}\cup \cdots \cup T_{\ell-1} }(P_\ell) \right)\notag\\
	&\leq f(O_1)+\sum_{\ell\in [2:L]}\sum_{r\in [0:\ell-2]}f_{T_{1}\cup \cdots \cup T_{\ell-1}}(O_{\ell} \cap E_{r}) + \sum_{\ell\in [2:L]}f_{T_{1}\cup \cdots \cup T_{\ell-1}}(P_{\ell}),\label{eq:term12}
	\end{align}
%	\begin{align}
%	f_T(O) %= &~ f(O \cup T) - f(T)\notag\\
%	& \leq \sum_{\ell=1}^{L}f_{T_{1}\cup \cdots \cup T_{\ell-1}\cup O_1\cup \cdots \cup O_{\ell-1}}(O_{\ell})\notag\\
%	\leq &~ \sum_{\ell=1}^{L}\left(\sum_{r=0}^{\ell-2}f_{T_{1}\cup \cdots \cup T_{\ell-1}\cup O_1\cup \cdots \cup O_{\ell-1}}(O_{\ell} \cap E_{r}) + f_{T_{1}\cup \cdots \cup T_{\ell-1}\cup O_1\cup \cdots \cup O_{\ell-1}}(O_{\ell} \backslash (E_{1} \cup \cdots \cup E_{\ell-2}))\right)\notag\\
%	\leq&~ \sum_{\ell=1}^{L}\sum_{r=0}^{\ell-2}f_{T_{1}\cup \cdots \cup T_{\ell-1}}(O_{\ell} \cap E_{r}) + \sum_{\ell=1}^{L}f_{T_{1}\cup \cdots \cup T_{\ell-1}}(O_{\ell} \backslash (E_{1} \cup \cdots \cup E_{\ell-2}))\label{eq:term12}
%	\end{align}
where both steps follow from the submodularity of $f$.

%he first step follows from the monotonicity, the second, third and the fourth step follow from the submodularity.
	
We first bound the first and last term.
\begin{lemma}
	\label{lem:upper2-first}
	We have
	\[
	f(O_1)+\sum_{\ell\in [2:L]} f_{T_1\cup\cdots\cup T_{\ell-1}}(P_\ell)
	\le (1+\eps)f(T)+6\eps \OPT.
	\]
\end{lemma} 

We will need the following fact, the proof can be found at Appendix~\ref{sec:insert-matroid-app}: 
%$$
%f(O_1)\le \sum_{o\in O_1} f(\{o\})\le |O_1|\cdot \OPT_t\le |O_1
%$$
\begin{fact}\label{fact111}
	Suppose $f$ is a monotone submodular funciton. For any three sets $A, B$ and $C$ we have 
	\[
	f_{A}(C) \leq f_{A\cup B}(C) + f_{A}(B).
	\]
\end{fact}	
	
%For the second term, for any element $o\in O_{\ell}$, we claim %
%	
\begin{proof}[Proof of Lemma~\ref{lem:upper2-first}]
Recall that for every $o\in P_\ell$ with $\ell\ge 2$, we have  
	$T_1 \cup \cdots \cup T_{\ell-2}\cup S_{\ell-1} \cup \{o\} \in \mathcal{M}.
	$
We have
%	Otherwise, there must exists some $0 \leq \ell' < \ell -1$ such that $T_1 \cup \cdots \cup T_{\ell'-1}\cup S_{\ell'} \cup o \in \mathcal{M}$ and $T_1 \cup \cdots \cup T_{r-1}\cup S_{r} \cup o \notin \mathcal{M}$ for any $r > \ell'$, this indicates that $o \in E_{\ell'}$, which can not happen. Hence, we conclude
	%for any $e \in O_{\ell} \backslash (E_{1} \cup \cdots \cup E_{s})$, we know $e\in T_{1}\cup T_{\ell-1}\cup S_{\ell} \cup e \in \mathcal{M}$, hence
	\begin{align*}
 \sum_{\ell\in [2:L]}f_{T_{1}\cup \cdots \cup T_{\ell-1}}(P_\ell) 
	&\leq \sum_{\ell\in [2:L]}f_{T_{1}\cup \cdots \cup T_{\ell-2}\cup S_{\ell-1}}(P_\ell) + f_{T_1\cup \cdots \cup T_{\ell-1}}(S_{\ell-1}\backslash T_{\ell-1})\notag\\
	&\leq \sum_{\ell\in [2:L]}f_{T_{1}\cup \cdots \cup T_{\ell-2}\cup S_{\ell-1}}(P_\ell) + 2\eps\OPT\notag\\
	&\leq \sum_{\ell\in [2:L]}\sum_{o\in P_\ell}f_{T_{1}\cup \cdots \cup T_{\ell-2}\cup S_{\ell-1}}(o)  + 2\eps\OPT\notag\\
	&\leq \sum_{\ell\in [2:L]}|O_\ell|(1+\eps)^{-\ell+2}\OPT  + 3\eps\OPT
%	&=  \sum_{\ell\in [2:L]}|T_\ell|(1+\eps)^{-\ell+1}\OPT  + 2\eps\OPT\notag\\%
%	&\leq (1+\eps)\sum_{\ell=1}^{L}\sum_{i=1}^{|T_{\ell}|}f_{T_1\cup\cdots\cup T_{\ell-1} \cup T_{\ell, :i-1}}(T_{\ell, i})  + 3\eps\OPT\notag\\
%	&= (1+\eps)f(T) + 3\eps\OPT. \label{eq:term2}
	\end{align*}
	The first step comes from Fact~\ref{fact111}.
	The second step follows from 
	\begin{align}
	f_{T_1\cup \cdots \cup T_{\ell-1}}(S_{\ell-1}\backslash T_{\ell-1}) &= f_{T_1\cup \cdots \cup T_{\ell-2}}(S_{\ell -1}) -  f_{T_1\cup \cdots \cup T_{\ell-2}}( T_{\ell -1}) \notag\\[0.6ex]
	&\le h_{T_1\cup \cdots \cup T_{\ell-2}}(S_{\ell -1}) -  h_{T_1\cup \cdots \cup T_{\ell-2}}( T_{\ell -1})+4\kappa \notag\\[0.6ex]
	&\leq h_{T_1\cup \cdots \cup T_{\ell-2}}(S_{\ell -1})- \Delta\left\lfloor \frac{h_{T_{1}\cup \cdots \cup T_{\ell-2}}(S_{\ell-1})}{\Delta}\right\rfloor +4\kappa\notag\\
	&\leq \Delta+4\kappa \label{eq:matroid1}
	\end{align}
	and that $(\Delta+4\kappa)L\le 2\eps \OPT$.
%	\begin{align}
%	\sum_{\ell=1}^{L}f_{T_1\cup \cdots T_{\ell-1}}(S_{\ell-1}\backslash T_{\ell-1})  \leq \frac{\eps^2 \OPT}{\log (k/\eps)}\cdot L = \frac{\eps^2 \OPT}{\log (k/\eps)}\cdot \left\lceil \frac{\log(k/\eps)}{\eps} \right\rceil \leq 2\eps\OPT. \label{eq:matroid1} 
%	\end{align}
	The third step comes from the submodularity. The fourth step holds since $T_1\cup\cdots \cup T_{\ell-2}\cup S_{\ell-1} \cup \{o\}$ is feasible for each $o\in P_\ell$, by the 
	$L$-pass greedy algorithm,
	\[
	f_{T_{1}\cup \cdots \cup T_{\ell-2}\cup S_{\ell-1}}(o) 
	\le h_{T_{1}\cup \cdots \cup T_{\ell-2}\cup S_{\ell-1}}(o)+\kappa
	\leq (1+\eps)^{-\ell+2}\OPT +\kappa 
	\]
and that $k\kappa\le \eps \OPT$.
On the other hand, we have $f(O_1)$ is $0$ if $|O_1|=0$ or when $|O_1|>0$ (and thus, $|T_1|>0$),
$$
f(O_1)\le \OPT_t\le (1+\eps)\OPT\le |O_1|(1+\eps)\OPT.
$$
In summary, we have
\begin{align*}
f(O_1)+\sum_{\ell\in [2:L]} f_{T_1\cup\cdots\cup T_{\ell-1}}(P_\ell)
&\le \sum_{\ell\in [L]} |O_\ell|(1+\eps)^{-\ell+2} \OPT+3\eps \OPT \\
&\le \sum_{\ell\in [L-1]} |T_\ell|(1+\eps)^{-\ell+2} \OPT +4\eps \OPT
\end{align*}
where we used $k(1+\eps)^{-L+2}\OPT\le \eps \OPT$.
To connect with $f(T)$, we have from the greedy algorithm that
$$
f(T)=\sum_{\ell\in [L]} f_{T_1\cup\cdots \cup T_{\ell-1}}(T_\ell)
\ge \sum_{\ell\in [L]} h_{T_1\cup\cdots \cup T_{\ell-1}}(T_\ell) -2L\kappa
\ge 
\sum_{\ell\in [L]}|T_\ell| (1+\eps)^{-\ell+1} \OPT -\eps\OPT.
$$
As a result, we have 
\begin{equation}\label{hehe3}
f(O_1)+\sum_{\ell\in [2:L]} f_{T_1\cup\cdots\cup T_{\ell-1}}(P_\ell)
\le (1+\eps)(f(T)+\eps \OPT)+4\eps \OPT
\le (1+\eps)f(T)+6\eps \OPT.
\end{equation}
%	The fifth step holds since $|O_\ell| = |T_\ell|$, and the sixth step holds since when $\ell \leq L-1$, for any $i =\{1,2,\cdots, |T_\ell|\}$, one has
%	\[
%	f_{T_1\cup\cdots\cup T_{\ell-1} \cup T_{\ell, :i-1}}(T_{\ell, i})  \geq (1+\eps)^{-\ell}\OPT
%	\]
%	and when $\ell = L$, one has
%	\[
%	|T_{\ell}| (1+\eps)^{-L+1}\OPT \leq k\cdot \frac{\eps}{k}\OPT = \eps\OPT.
%	\]
	
\end{proof}
	We bound the second term of Eq. (\ref{eq:term12}) as follow.
\begin{lemma}\label{lem:upper2-second}
	We have
	\[
	\sum_{\ell\in [2:L]}\sum_{r\in [0:\ell-2]}f_{T_{1}\cup \cdots \cup T_{\ell-1} }(O_{\ell} \cap E_{r})  \leq 8\eps\OPT.
	\]
\end{lemma}

\begin{proof}
	The proof can be derived via the following chain of inequalities.
	\begin{align}
	\sum_{\ell\in [2:L]}\sum_{r\in [0:\ell-2]}f_{T_{1}\cup \cdots \cup T_{\ell-1} }(O_{\ell} \cap E_{r})  &=  \sum_{r\in [0:L-2]}\sum_{\ell\in [r+2:L]}f_{T_{1}\cup \cdots \cup T_{\ell-1} }(O_{\ell} \cap E_{r})\notag \\
	&\leq   \sum_{r\in [0:L-2]}f_{T_{1}\cup \cdots \cup T_{r}}(E_{r}) \notag\\
	&\leq   \sum_{r\in [0:L-2]}f_{T_{1}\cup \cdots \cup T_{r-1}\cup S_{r}}(E_{r}) + \sum_{r\in [0:L-2]} f_{T_1\cup \cdots \cup T_{r}}(S_{r}\backslash T_{r})\notag\\
	&\leq   \sum_{r\in [0:L-2]}f_{T_{1}\cup \cdots \cup T_{r-1}\cup S_{r}}(E_{r})  + 2\eps\OPT\notag\\
	&\leq   \sum_{r\in [0:L-2]}\sum_{o \in E_r} f_{T_{1}\cup \cdots \cup T_{r-1}\cup S_{r}}(o) + 2\eps\OPT\notag\\
	&\leq   \sum_{r\in [0:L-2]}\sum_{o \in E_r} h_{T_{1}\cup \cdots \cup T_{r-1}\cup S_{r}}(o) + 3\eps\OPT\notag\\
	&\leq   \sum_{r\in [0:L-2]}|E_r|(1+\eps)^{-r+1}\OPT + 3\eps\OPT\notag\\
	&\leq   \sum_{r\in [0:L-2]}|S_{r+1} \backslash T_{r+1}|(1+\eps)^{-r+1}\OPT + 3\eps\OPT\notag\\
	&\leq    (1+\eps)\sum_{r=0}^{L-2}f_{T_1\cup \cdots \cup T_{r+1}}(S_{r+1}\backslash T_{r+1}) + 4\eps\OPT\notag\\
	&\leq   8\eps\OPT \label{eq:term1}.
	\end{align}
	%	\begin{align}
	%	\sum_{\ell\in [2:L]}\sum_{r=0}^{\ell-2}f_{T_{1}\cup \cdots \cup T_{\ell-1}\cup O_1 \cup \cdots O_{\ell-1}}(O_{\ell} \cap E_{r})  =&~ \sum_{r=0}^{L-2}\sum_{\ell=r+2}^{L}f_{T_{1}\cup \cdots \cup T_{\ell-1}\cup O_1 \cup \cdots O_{\ell-1}}(O_{\ell} \cap E_{r})\notag \\
	%	\leq &~ \sum_{r=0}^{L-2}f_{T_{1}\cup \cdots \cup T_{r}}(E_{r}) \notag\\
	%	\leq &~ \sum_{r=0}^{L-2}f_{T_{1}\cup \cdots \cup T_{r-1}\cup S_{r}}(E_{r}) + \sum_{r=0}^{L-2} f_{T_1\cup \cdots \cup T_{r}}(S_{r}\backslash T_{r})\notag\\
	%	\leq &~ \sum_{r=0}^{L-2}f_{T_{1}\cup \cdots \cup T_{r-1}\cup S_{r}}(E_{r})  + 2\eps\OPT\notag\\
	%	\leq &~ \sum_{r=0}^{L-2}\sum_{o \in E_r} f_{T_{1}\cup \cdots \cup T_{r-1}\cup S_{r}}(o) + 2\eps\OPT\notag\\
	%	\leq &~ \sum_{r=0}^{L-2}|E_r|(1+\eps)^{-r}\OPT + 2\eps\OPT\notag\\
	%	\leq &~ \sum_{r=0}^{L-2}|S_{r+1} \backslash T_{r+1}|(1+\eps)^{-r}\OPT + 2\eps\OPT\notag\\
	%	\leq &~  (1+\eps)\sum_{r=0}^{L-2}f_{T_1\cup \cdots \cup T_{r+1}}(S_{r+1}\backslash T_{r+1}) + 2%\eps\OPT\notag\\
	%	\leq &~ 6\eps\OPT \label{eq:term1}
	%	\end{align}
	The second step holds due to the submodularity and monotonicity. The third step holds due to Fact~\ref{fact111}, and the fourth step holds due to Eq.~\eqref{eq:matroid1}. The fifth step follows from submodularity. The sixth step holds since $h$ is $\kappa$-close to $f$ and 
	using $2k\kappa\le \eps\OPT$.
	The seventh step holds since every $o \in E_r$ has $T_1\cup \cdots\cup T_{r-1}\cup S_r \cup \{o\} \in \mathcal{M}$. Hence, it follows from the $L$-pass greedy algorithm that
	\[
	h_{T_{1}\cup \cdots \cup T_{r-1}\cup S_{r}}(o) \leq (1+\eps)^{-r+1}\OPT.
	\]
	The seventh step holds due to Claim \ref{claim111} that $|E_r| \leq |S_{r+1} \backslash T_{r+1}|$. The eighth step holds since 
	\[
	f_{T_1\cup \cdots \cup T_{r+1}}(S_{r+1}\backslash T_{r+1})
	\ge h_{T_1\cup \cdots \cup T_{r+1}}(S_{r+1}\backslash T_{r+1})-\kappa
	\geq |S_{r+1}\backslash T_{r+1}|(1+\eps)^{-r }\OPT-\kappa,
	\]
	and the last step holds due to \eqref{eq:matroid1}.
\end{proof}

	Combining Lemma~\ref{lem:upper2-first}, Lemma~\ref{lem:upper2-second}  and Eq.~\eqref{eq:term12}, we have
	%\[
	%f(T) \geq (1-O(\eps))f_{T}(O) -O(\eps)\cdot \OPT.\qedhere
	%\]
	$f_{T}(O) 
	%\leq&~ \sum_{\ell=1}^{L}\sum_{r=0}^{\ell-2}f_{T_{1}\cup \cdots \cup T_{\ell-1}}(O_{\ell} \cap E_{r}) + \sum_{\ell=1}^{L}f_{T_{1}\cup \cdots \cup T_{\ell-1}}(O_{\ell} \backslash (E_{1} \cup \cdots \cup E_{\ell-2}))\\
	\leq (1+\eps)f(T) + 11\eps\OPT .$
	Since this holds for all feasible sets, and in particular, it holds for the set $O_t$. Taking a linear combination, we have
	\begin{align*}
	f(T)  \geq \frac{1}{1+\eps} \big( (1 -11\eps)f_{T}(O) + 11\eps  f_{T}(O_t) - 11\eps\OPT\big) 
	%\geq&~ \frac{1}{1+\eps} \left((1 -11\eps)f_{T}(O) + 11\eps f(O_t) - 11\eps f(T) - 11\eps\OPT\right)\\
	\geq  \frac{1}{1+\eps} \big((1 -11\eps)f_{T}(O) - 11\eps f(T) \big).
	\end{align*}
	Rearranging the term, we get the desired.

	Finally, we bound the amortized query complexity of the algorithm.
	We charge the two queries made in the evaluation of $h_{S}(e)$ (line \ref{line:eva3} or \ref{line:eva4}) to $e$ and 
	show that the number of queries charged to $e$ is at most $O(L)$.
	To see this, we note that $e$ is charged twice when it is just inserted.
	Every time $e$ is charged during \textsc{Revoke}, either it is 
	added to $S$ so that it is never charged again, or its marginal contribution is small and won't be queried in the $\ell$-th level later.
	Hence, $e$ has been queried for at most $O(L)$ times.
	We conclude the proof here.
\end{proof}

By standard argument, we conclude

\begin{lemma}
	\label{lem:matroid}
	Given a matroid $\mathcal{M}$, let $t>0$,  $\OPT\leq \OPT_{t} \leq (1+\eps)\OPT$. After inserting the $t$-th element, our algorithm outputs a feasible set $S$ that satisfies 
	%\[
	%f(S)\geq (1-O(\eps))f_{S}(O) -O(\eps)\cdot \OPT.
	%\]
	%for any feasible set $O\subseteq V_t$, and therefore,
	%\[
	$f(S)\geq (1/2-O(\eps))\OPT.$
	%\] 
	The expected number of oracles per operation is $k^{\tilde{O}(1/\eps)}$.
\end{lemma}

The proof can be found at Appendix~\ref{sec:insert-matroid-app}.
Combine with Lemma~\ref{lem:opt}, we have
\begin{theorem}
	\label{thm:insert-matroid-com}
	 Given any matroid $\mathcal{M}$, for any $\eps > 0$, there is a combinatorial algorithm that maintains a feasible set $S$ with $(1/2 -\eps)$-approximation at each iteration. Moreover, the amortized number of queries per update is $k^{\tilde{O}(1/\eps)}$.
\end{theorem}

%In the Appendix~\ref{sec:intersection}, we show how to generalize our algorithm to the general constraints of intersection of $P$ matroid, and achieve $\frac{1}{P+1} - \eps$ approximation.
%\begin{theorem}
%		Given any intersection of $P$ matroids and any $\eps> 0$, there is a combinatorial algorithm that maintains a feasible set $S$ with $(1/(P+1) -\eps)$-approximation at each iteration. Moreover, the amortized number of queries per update is $k^{\tilde{O}(1/\eps)}$.\footnote{\Binghui{TODO soon}}
%\end{theorem}

\subsection{Amplification via accelerated continuous greedy}
\label{sec:accelerate}

We amplify the approximation ratio of the combinatorial algorithm via the accelerated continuous greedy framework~\cite{badanidiyuru2014fast}. 
Let $m = O(1/\eps)$ and
\[
 D = \{\OPT, (1+\eps)^{-1}\OPT, \cdots, (1+\eps)^{-\lceil 4\log(1/\eps)/\eps \rceil}\OPT\}\cup \{0\}.
\] 
We run a separate branch for each $\bd \in D^{m} $ and $\ba \in \mathcal{A}^{m}$. Intuitively, $d_\tau \in D$ should be seen as an estimate on the progress of optimal solution in the $\tau$-th iteration, and $a_\tau\in \mathcal{A}$ is a guess on the greedy sequence.
Our algorithm segragates answer from all branches and outputs the one with the maximum value. 
In order to return an integral solution, the algorithm rounds the fractional solution via the swap rounding approach~\cite{chekuri2010dependent}.
The algorithm description is presented in Algorithm~\ref{algo:insert-matroid}.

\begin{algorithm}[!h]
	\caption{Amplification via accelerated continuous greedy}
	\label{algo:insert-matroid}
	\begin{algorithmic}[1]
		\State \textbf{Input: }$\bd \in D^{m}, \ba\in \mathcal{A}^{m}$
		\State Initialize $\bx \leftarrow \mathbf{0}$
		\For{$\tau =1,2,\cdots, m$}
		\State Define $g(S) = F(\bx+\frac{1}{m} S)-F(\bx)$ for all $S\subseteq V$
		\State Invoke \textsc{Prune-greedy}($d_\tau, a_\tau$, g), and wait until it returns a solution $S_{\tau}$
		\State  $\bx \leftarrow \bx + \frac{1}{m}S_\tau$\label{line:update}
		\EndFor
	\end{algorithmic}
\end{algorithm}

We need the following lemma in our analysis. The proof idea follows from~\cite{badanidiyuru2014fast} and it appears in \cite{balkanski2019optimal}. It has some minor difference with previous work and we provide a proof for completeness.

\begin{lemma}
	\label{lem:accelerated}
	Let $\OPT\leq \OPT_t \leq (1+\eps)\OPT$.
	Suppose in each iteration of Algorithm~\ref{algo:insert-matroid}, \textsc{Prune-greedy} returns a set $S$ that satisfies
	\begin{align}
	g(S) \geq (1-O(\eps)) \sum_{i\in [L]}g_{S\backslash (O_i \cup \cdots \cup O_L)}(O_i) -\eps^2\OPT,\label{eq:acc1}
	\end{align}
	for some partition $O = O_1 \cup \cdots \cup O_L$ of $O$
	and some partition $S = S_1 \cup \cdots \cup S_L$ 
	of $S$ such that $\forall i \in [L]$,  $S_1 \cup \cdots S_i \cup O_{i+1} \cup \cdots \cup O_{L} \in \mathcal{M}$ and $S_i \cap (O_{i+1}\cup\cdots \cup O_L) = \emptyset$.
	Then the final solution $\bx$ satisfies
	\[
	F(\bx)\geq (1-1/e - O(\eps))\OPT.
	\]
\end{lemma}

\begin{proof}
	For any $\tau \in [m]$, let $\bx_{\tau}$ denote the fractional solution at the end of $\tau$-th iteration, i.e. after the update at Line~\ref{line:update} of Algorithm~\ref{algo:insert-matroid}. The we claim
	\begin{align}
	F(\bx_\tau) - F(\bx_{\tau-1}) \geq \frac{1}{m}(1 - O(\eps))(F(O) - F(\bx_{\tau-1})) - \eps^2 \OPT.\label{eq:acc2}
	\end{align}
	This is true since whenever $g(S) = F(\bx_\tau) - F(\bx_{\tau-1})  \leq  \frac{1}{m}(F(O) - F(\bx_{\tau-1}))$, one has 
	\begin{align*}
	F(\bx_\tau) - F(\bx_{\tau-1}) = &~ g(S) \geq (1-O(\eps)) \sum_{i\in [L]}g_{S\backslash (O_i \cup \cdots \cup O_L)}(O_i) -\eps^2\OPT \\
	\geq &~ \frac{1}{m}\left(1 - \frac{1}{m}\right)(1-O(\eps)) (F(O) - F(\bx_{\tau-1})) -\eps^2 \OPT \\
	=&~ \frac{1}{m}(1 - O(\eps))(F(O) - F(\bx_{\tau-1})) - \eps^2 \OPT
	\end{align*}
	where the second step follows from the assumption (see Eq.~\eqref{eq:acc1}), the third step follows from Claim~\ref{lem:aux} below and the last step follows from the choice of parameter $m = O(1/\eps)$.
	%\begin{align*}
	%F(\bx_{\tau}) - F(\bx_{\tau-1}) = &~ g(S_{\tau}) \\
	%\geq &~ (1-O(\eps))g_{S_{\tau}}(O_t) - \eps^{2}\OPT\\
	%=&~ (1-O(\eps)) \left(F\left(\bx_{\tau} + \frac{1}{m}O_t\right) - F(\bx_{\tau})\right) - \eps^{2}\OPT\\
	%\geq &~  (1-O(\eps)) \left(\frac{1}{m}F(\bx_{\tau} + O_t) + \left(1-\frac{1}{m}\right)F(\bx_{\tau}) - F(\bx_{\tau})\right)  - \eps^{2}\OPT\\
	%= &~\frac{1}{m}(1-O(\eps)) (F(\bx_{\tau}+O_t) - F(\bx_{\tau}))- \eps^{2}\OPT\\
	%\geq &~ \frac{1}{m}(1-O(\eps)) (F(O_t) - F(\bx_{\tau}) )- \eps^{2}\OPT.
	%\end{align*}
	%The first and third step holds due to the definition of $g$, the second step holds due to Eq.~\eqref{eq:acc1}. The fourth step holds due to the submodularity of the multilinear extension $F$. The last step is due to monotonicity.
	
	Recurring the Equation \eqref{eq:acc2} for all $\tau \in [m]$, via standard calculation, one can show that
	\[
	F(T) \geq (1-1/e -O(\eps))F(O) = (1-1/e -O(\eps))  \OPT. \qedhere
	\]
\end{proof}

 \begin{claim}[Adapted from Lemma 7 in \cite{balkanski2019optimal}]
 	\label{lem:aux}
 	For any $\tau\in [m]$, suppose that $$g(S) \leq \frac{1}{m}\left(F(O)- F(\bx_{\tau-1})\right),$$ then we have $$\sum_{i\in [L]}g_{S\backslash (O_i \cup \cdots \cup O_L)}(O_i)  \geq \frac{1}{m}\left(1-\frac{1}{m}\right) (F(O) - F(\bx_{\tau-1})).$$
 \end{claim}
 \begin{proof}
We have
\begin{align*}
\sum_{i\in [L]}g_{S\backslash (O_i \cup \cdots \cup O_L)}(O_i)  = &~ \sum_{i\in [L]} F_{\bx_{\tau-1} + \frac{1}{m}  S\backslash (O_i \cup \cdots \cup O_L)}\left(\frac{1}{m}O_i\right) \\ \geq&~ \sum_{i\in [L]} F_{\bx_{\tau-1} + (O_{1} \cup \cdots \cup O_{i-1})+\frac{1}{m} S}\left(\frac{1}{m}O_i\right)  \\
\geq &~ \frac{1}{m}  \sum_{i\in [L]} F_{\bx_{\tau-1} + O_{1: i-1} +\frac{1}{m} S}(O_i)  = \frac{1}{m}F_{\bx_{\tau-1} + \frac{1}{m}S}(O),
\end{align*}
where the first step follows from $S\backslash (O_i \cup \cdots \cup O_L) \cap O_i = \emptyset$, the second step follows from submodularity, the third step follows from multi-linearity. 
Moreover, we  have
\begin{align*}
F_{\bx_{\tau-1} + \frac{1}{m}S}(O) =&~ F_{\bx_{\tau-1}}\left(O +\frac{1}{m}S\right) - F_{\bx_{\tau-1}}\left(\frac{1}{m}S\right) \geq F_{\bx_{\tau-1}}(O) - \frac{1}{m} (F(O) - F(\bx_{\tau-1})) \\
\geq &~ \left(1- \frac{1}{m}\right) (F(O) - F(\bx_{\tau-1})) ,
\end{align*}
where the second step comes from the condition that $g(S) \leq \frac{1}{m}(F(O)- F(\bx_{\tau-1}))$. Combining the above two inequalies, we conclude the proof.
 \end{proof}

\thmuppertwo*

%\begin{theorem}%
%		\label{thm:insert-matroid}
%		Given a matroid $\mathcal{M}$, let $\eps > 0$, there is a randomized algorithm that maintains a feasible set with $(1-1/e-\eps)$-approximation at each iteration. Moreover, the expected number of oracles per operation is $k^{\tilde{O}(1/\eps^2)}\log n$.
%\end{theorem}

\begin{proof}
	With Lemma~\ref{lem:opt}, it suffices to prove for any $t$, when $\OPT\leq \OPT_t \leq (1+\eps)\OPT$, our algorithm returns a solution set $S$ with $f(S)\geq (1-1/e-O(\eps))\OPT_t$. 
	Let $O$ be the optimal solution.
	For each $\tau = 1,\ldots, m$ (with $\bx_0=0$), we define
	  $g^{(\tau)}$, $O^{(\tau)}$, $d_\tau$, $a_\tau$ and $S_{\tau}$ round by round as follows:
\begin{enumerate}
\item
First given $\bx_{\tau-1}$ we define $g^{(\tau)}$ and $O^{(\tau)}$ as
$$g^{(\tau)}(S) := F\left(\bx_{\tau -1} + \frac{1}{m} S\right)\quad\text{and}\quad
O^{(\tau)} = \arg\max_{S \in \mathcal{M}}g^{(\tau)}(S).$$ 
\item Let $\ell$ be the integer that satisfies 
 $(1+\eps)^{-\ell}\OPT \leq g^{(\tau)}(O^{(\tau)}) \leq (1+\eps)^{-\ell+1}\OPT$. When $\ell \leq \lceil 4 \log(1/\eps)/\eps\rceil$, we take $d_{\tau} = (1+\eps)^{-\ell}\OPT$ and we know $d_{\tau} \leq g^{(\tau)}(O^{(\tau)}) \leq (1+\eps)d_{\tau}$; otherwise, we set $d_{\tau} = 0$.
When $d_\tau=0$, we set $S_{\tau}=\emptyset$ and set $a_\tau$ to be an arbitrary
  element in $\mathcal{A}$.
When $d_\tau>0$, we set $a_\tau\in \calA$ in the same way as in the proof 
  regarding $g^{(\tau)}$ and $d_\tau$.
We run   Prune-Greedy on $g^{(\tau)}$ using $d_\tau$ and 
  $a_\tau$ and let $S_\tau$ be the set it returns.
\end{enumerate}
	%For each $\tau \in [m]$, let $g^{(\tau)}(S) := F(\bx_{\tau -1} + \frac{1}{m} S)$, $O^{(\tau)} = \arg\max_{S \in \mathcal{M}}g^{(\tau)}(S)$ and  $(1+\eps)^{-\ell}\OPT \leq g^{(\tau)}(O_t) \leq (1+\eps)^{-\ell+1}\OPT$. When $\ell \leq \lceil 4 \log(1/\eps)/\eps\rceil$, we take $d_{\tau} = (1+\eps)^{-\ell}\OPT$ and we know $d_{\tau} \leq g^{(\tau)}(O^{(\tau)}) \leq (1+\eps)d_{\tau}$; otherwise, we set $d_{\tau} = 0$. 
It is clear that the branch of the algorithm that runs with $d$ and $a$
  defined above would return $\textbf{x}=\frac{1}{m}\sum_{\tau\in [m]}S_\tau$.
We prove below that each set $S_\tau$ satisfies
	%\begin{align}
	%g^{(\tau)}(S_{a, \tau}) \geq (1 - O(\eps))g^{(\tau)}_{S_{a, \tau}}(O^{(\tau)}) \label{eq:condition1}
	%\end{align}
	%and 
	\begin{align}
	g^{(\tau)}(S_{\tau}) \geq (1 - O(\eps)) \sum_{i\in [L]}g^{(\tau)}_{S_{\tau} \backslash (O_{i} \cup \cdots O_L) }(O_{i}) -\eps^2 \OPT \label{eq:condition1}
	\end{align}
	for some partition $O = O_1 \cup \cdots \cup O_{L}$ of $O$
	and some partition $S_\tau=T_1\cup\cdots \cup T_L$ that satisfy for every
  $i \in [L]$,  $T_1 \cup \cdots T_i \cup O_{i+1} \cup \cdots \cup O_{L} \in \mathcal{M}$ and $T_i \cap (O_{i+1}\cup\cdots \cup O_L) = \emptyset$;
%	(1) for all $i < j$, $T_i \cap O_j = \emptyset$ and (2) for every $i < L$, $T_1\cup, \cdots T_{i} \cup O_{i+1} \cup \cdots \cup O_L \in \mathcal{M}$. that we will specify below
	note that both partitions depend on $\tau$. 
It follows from Lemma~\ref{lem:accelerated}
  that $F(\bx) \geq (1-1/e-O(\eps))\OPT$.
Given $\bx$, we round the fractional solution to an integral one. The solution is the convex combination of at most $O(1/\eps)$ independent sets, it takes only $O(k^2/\eps)$ query to perform the swap rounding~\cite{chekuri2010dependent}.	
	
To prove Eq.~\eqref{eq:condition1}, note that the 
	right hand side of the first line of \eqref{eq:term12} is 
$$
\sum_{i\in [L]}g^{(\tau)}_{T_1 \cup \cdots \cup T_{i-1}}(O_{i}).
$$where $T_1, \ldots, T_{L}$ is a partition of $S_{\tau}$ and $O_{1}, \ldots, O_{L}$ is a partition of $O$ that satisfy
for every
  $i \in [L]$,  $T_1 \cup \cdots T_i \cup O_{i+1} \cup \cdots \cup O_{L} \in \mathcal{M}$ and $T_i \cap (O_{i+1}\cup\cdots \cup O_k) = \emptyset$.
On the other hand, the second line of \eqref{eq:term12} 
  can be bounded from above using
  Lemma \ref{lem:upper2-first} and Lemma \ref{lem:upper2-second} by
  %and the second and third steps of Eq.~\eqref{eq:term12}, one has
	\begin{align*}
	%\label{eq:condition1-2}
	(1+\eps) g^{(\tau)}(S_{\tau}) % \geq (1 - O(\eps))\sum_{i\in [L]}g^{(\tau)}_{T_1 \cup \cdots \cup T_{i-1}}(O_{i})  -
	+O(\eps)g^{(\tau)}(O^{(\tau)}).
	\end{align*}
It follows from \eqref{eq:term12} that 
\begin{align}\label{eq:condition1-2}
 g^{(\tau)}(S_{\tau}) \geq (1 - O(\eps))\sum_{i\in [L]}g^{(\tau)}_{T_1 \cup \cdots \cup T_{i-1}}(O_{i})  -
	 O(\eps)g^{(\tau)}(O^{(\tau)})
	\end{align}
	In addition we have
	\begin{align}
	\sum_{i\in [L]}g^{(\tau)}_{T_1 \cup \cdots \cup T_{i-1}}(O_{i}) \geq \sum_{\ell \in [L]}g^{(\tau)}_{S_{ \tau}\backslash (O_i \cup \cdots O_{L})}(O_{i}) 	\label{eq:condition1-3}
	\end{align}
	due to submodularity and the fact that $T_1 \cup \cdots \cup T_{i-1} \in S_{\tau}$ and $T_1 \cup \cdots \cup T_{i-1} \cap (O_{i} \cup \cdots \cup O_{L})= \emptyset$. Combining Eq.~\eqref{eq:condition1-2} and Eq.~\eqref{eq:condition1-3}, we have
	\begin{align}
	g^{(\tau)}(S_{\tau}) \geq (1 - O(\eps)) \sum_{i \in [L]}g^{(\tau)}_{S_{\tau}\backslash (O_i \cup \cdots O_{L})}(O_{i})   - O(\eps)g^{(\tau)}(O^{(\tau)}).	\label{eq:condition1-4}
	\end{align}
	To get rid of low order term, we note by Lemma~\ref{lem:matroid} that
	\begin{align}
	\label{eq:condition1-5}
	g^{(\tau)}(S_{\tau}) \geq \left(\frac{1}{2} - O(\eps)\right)g^{(\tau)}(O^{(\tau)})- \eps^{2}\OPT,
	\end{align}
	the additive term comes when $d_{\tau} = 0$ and make RHS no more than $0$ in that case. Adding Eq.~\eqref{eq:condition1-4} and Eq.~\eqref{eq:condition1-5}, we proved Eq.~\eqref{eq:condition1}.
	
	%We will use this sequence and set $a_\tau = a$. 
%	By recursively applying it for all $\tau \in [m]$, the resulting $\bd, \ba$ satisfy the condition of Lemma~\ref{lem:accelerated}, and we conclude the algorithm returns a solution $\bx$ with $F(\bx) \geq (1-1/e-O(\eps))\OPT$.
	
	For the amortized query complexity, the total number of branch equals 
	\[
	|\mathcal{A}^{m}|\cdot |D^{m}| = k^{\tilde{O}(1/\eps^2)} \cdot (1/\eps)^{\tilde{O}(1/\eps)} = k^{\tilde{O}(1/\eps^{2})}.
	\]
	For each branch, there are $O(1/\eps)$ iterations, and in each iteration, we run our combinatorial algorithm with a scheduel sequence $a_{\tau}\in \mathcal{A}$. By Lemma~\ref{lem:matroid-approx}, the amortized number of query (to function $g$) equals $O(L) = O(\log(k/\eps)/\eps)$. We note that in the accelerated continuous greedy algorithm, a query to function $g$ is equivalent to two queries of the multilinear extension.
	To satisfy the requirement of Lemma~\ref{lem:matroid-approx}, one evaluates the value of multilinear extension with error $\pm\eps^{3}/k$ and confidence $\delta = \frac{\eps}{\poly(n)|\mathcal{A}^m||\mathcal{D}^m|L}$  using $\tilde{O}(\eps^{-8}k^2\log n)$ queries to the original function $f$. Thus in total, the amortized number of query is at most
	\[
	k^{\tilde{O}(1/\eps^2)} \cdot \eps^{-1} \cdot \eps^{-1}\log(k/\eps)\cdot \tilde{O}(\eps^{-8}k^2\log n)= k^{\tilde{O}(1/\eps^2)} \log n.
	\]
	 Hence, the overall amortized query complexity is $k^{\tilde{O}(1/\eps^{2})}\log n$.
\end{proof}

% !TEX root =  main.tex

\section{Conclusions}
\label{sec:discussion}

We study the power and limitations of dynamic algorithms 
for submodular maximization. % and characterize the amortized query complexity. 
On~the lower bound side, we prove a polynomial lower bound on the amortized query complexity for achieving a $(1/2+\eps)$-approximation, together with a linear lower bound for $0.584$-approximation, under fully dynamic streams with insertions and deletions.
On the algorithmic side, we develop efficient $(1-1/e)$-approximation algorithms for  insertion-only streams under both cardinality and~matroid constraints.
There are
many interesting directions for further investigations:
\begin{flushleft}\begin{itemize}
		%\item There is still space for improvement on our lower bounds. Is it possible to prove a linear lower bound for $({1}/{2}+\eps)$-approximation? Or is it actually possible to develop a sublinear algorithm (in amortized complexity) that achieves a $(1/2+ \Omega(1))$-approximation?
		\item Many submodular functions important in practice can be accessed in white box models instead of query models, e.g., the MAX-$k$ coverage problem, influence maxmization (see \cite{peng2021dynamic} for an example).  Can  ideas in this paper be extended to obtain upper/lower bounds on amortized time complexity for these problems? 
		\item Can we extend results (algorithm or hardness) to non-monotone submodular maximization? As far as we know, there is no known constant-factor approximation algorithm with $\poly(k)$ amortized query complexity for the non-monotone setting under fully dynamic streams. How does the dynamic model compare to the streaming model \cite{Naor20}
		under this setting?
		\item For matroid constraints, can one improve the query complexity to $O(\sqrt{k})$ over insertion-only streams? Also, for fully dynamic streams, there is no known constant-factor approximation algorithm with $\poly(k)$ amortized queries for matroid constraints.
	\end{itemize}\end{flushleft}

\section*{Acknowledgement} Research of X.C. and B.P. were supported in part 
by NSF  grants CCF-1703925, IIS-1838154, CCF-2106429, CCF-2107187, CCF-1763970, CCF-1910700 and DMS-2134059. We would like to thank 
Paul Liu for pointing out a mistake in an earlier version of the paper.

%\newpage
\bibliographystyle{alpha}
\bibliography{ref}
\newpage
\appendix
%\paragraph{Roadmap of the Appendix}
%We provide missing proof in the Appendix.
%Appendix~\ref{sec:lower1-app} provides proof for Section~\ref{sec:lower1} and Section~\ref{sec:lower2-app} provides proof details for Section~\ref{sec:lower2}.

\section{Missing proof from Section~\ref{sec:lower1}}
\label{sec:lower1-app}

\subsection{Missing proof from Section~\ref{sec:base}}
\label{sec:base-app}

Let $x\in [0, 1]^{w}$ and $\bar{x} = \frac{1}{w}\sum_{i=1}^{w}x_{i}$. Consider the function
\begin{align*}
f(x) = 1 - \prod_{i=1}^{w}(1 - x_i)
\end{align*}
and its symmetric version
\begin{align*}
g(x) = 1 - (1 - \bar{x})^w.
\end{align*}

The follow Lemma comes from Lemma 2.4 of \cite{mirrokni2008tight}, we make some minor modifications and provide the proof here for completeness.
\begin{theorem}[Restatement of Theorem~\ref{thm:base}]
	\label{thm:base-app}
	For any $w \in \mathbb{Z}$, $\eps > 0$ let $\gamma = w^{-1}\exp(-4w^6/\eps)$. There is a function $\hat{f}: [0,1]^{w}\rightarrow [0,1]$ satisfies
	\begin{itemize}
		\item When $\max_{i, j}|x_i - x_j| \leq \gamma$, $\hat{f}(x) = g(x)$.
		\item For any $x \in [0,1]^{w}$, $f(x) \geq \hat{f}(x) \geq f(x) - \eps$.
	\end{itemize}
\end{theorem}

%It remains to perturbed the function so that $f(\bx) - g(\bx)$ when $\delta = \max_{i, j}|x_i - x_j| \leq \eps$. Let $h(\bx)$ denote

Let $h(x):[0,1]^{w}\rightarrow [0,1]$ be the difference of $f(x)$ and $g(x)$, i.e., $h(x) = f(x) - g(x)$. Here are some basic properties of $h(x)$.

\begin{lemma}[Claim of \cite{mirrokni2008tight}]
	\label{lem:base-prop}
	Let $h(x)= f(x) - g(x)$, and $\max_{i, j \in [w]}|x_i - x_j| = \delta$, then one has
	\begin{itemize}
		\item $h(x) \leq w\delta(1 - \bar{x})^{w-1}$.
		\item $h(x) \geq w^{-4}\delta^2(1 - \bar{x})^{w-2}$
		\item $|\frac{\partial h}{\partial x_j}| \leq w\delta (1 - \bar{x})^{w-2}$, i.e. $|\frac{\partial h}{\partial x_j}| \leq w^3 (1 - \bar{x})^{w/2-1}\sqrt{h(x)}$.
	\end{itemize}
\end{lemma}

Define $\tilde{f}(x) = f(x) - \phi(h(x))$ where $\phi:[0,1]\rightarrow [0,1]$ controls the interpolation between $f(x)$ and $g(x)$. 
\begin{lemma}[\cite{mirrokni2008tight}]
	\label{lem:phi}
	Let $\eps_1 = w\gamma = \exp(-4w^6/\eps)$, $\eps_2 = \exp(-2w^6/\eps)$, $\alpha = \frac{1}{2}w^{-6}\eps$, There is a function $\phi:[0,1] \rightarrow [0,1]$ that satisfies
	\begin{itemize}
		%\item For $t\in [0, \eps_1]$, $\phi(t) = t$. We take $\eps_1 = k\eps$, by the claim, we have $h(\bx) \leq \eps_1$. This time, we have $\tilde{f}(\bx) = g(\bx)$.
		\item For $t \in [0, \eps_1]$, $\phi(t) = t$.
		\item For $t \in [\eps_1, \eps_2]$, the second derivative is $\phi''(t) = -\alpha / t$ and the first derivative of $\phi$ is continuous at $t=\eps_1$, hence,
		\[
		\phi'(t) = 1- \int_{\eps_1}^{t} \frac{\alpha}{\tau} d\tau = 1 - \alpha \log \frac{t}{\eps_1}.
		\]
		%We choose $\alpha = 2/\log(1/\eps_1)$ and $\eps_2 = \sqrt{\eps_1}$, so we have $\phi'(\eps_2) = 0$. Since $0 \leq \phi'(t) \leq 1$ everywhere, we have $0 \leq \phi(\eps_2) \leq \eps_2$.
		\item For $t > \eps_2$, we set $\phi(t) = \phi(\eps_2) < 1$.
	\end{itemize}
\end{lemma}

Next, we show that $\tilde{f}$ is close to a monotone submodular function, by bounding its first derivative and second derivative.
\begin{lemma}
	\label{lem:derivative}
	For any $i, j \in [w]$, one has
	\[
	\frac{\partial \tilde{f}}{\partial x_j} \geq 0 \quad \text{and}\quad	\frac{\partial^2 \tilde{f}}{\partial x_i \partial x_j}  \leq\frac{1}{2}\eps(1 - \bar{x})^{w-2}.
	\]
\end{lemma}

\begin{proof}
	For the first derivative, we have
	\[
	\frac{\partial \tilde{f}}{\partial x_j} = \frac{\partial f}{\partial x_j} - \phi'(h) \frac{\partial h}{\partial x_j} = (1 - \phi'(h))\frac{\partial f}{\partial x_j}  + \phi'(h)\frac{\partial g}{\partial x_j}. 
	\]
	Since $0 \leq \phi'(h) \leq 1$,  and  $\frac{\partial g}{\partial x_j}, \frac{\partial f}{\partial x_j}\geq 0$, the first partial derivative is always non-negative.
\end{proof}

%Since we have $0 \leq \phi(t) \leq \eps_2$ everywhere and $\tilde{f}(\bx) = f(\bx) - \phi(h(\bx)) \geq \eps_2$. We didn't corrupt the function too much. We next show we don't corrupt the monotonicity and submodularity too much.

For the second partial derivatives, we have
\begin{align*}
\frac{\partial^2 \tilde{f}}{\partial x_i \partial x_j} =&~ \frac{\partial^2 f}{\partial x_i \partial x_j} - \phi'(h)\frac{\partial^2 h}{\partial x_i \partial x_j} - \phi^{''}(h)\frac{\partial h}{\partial x_i} \frac{\partial h}{\partial x_j} \\
= &~(1 - \phi'(h)) \frac{\partial^2 f }{\partial x_i \partial x_j} + \phi'(h)\frac{\partial^2 g}{\partial x_i \partial x_j} - \phi''(h)\frac{\partial h}{\partial x_i} \frac{\partial h}{\partial x_j} 
\end{align*} 

The first two terms are non-negative. To control the third term, we note that $\phi''(h) \leq \alpha/h$ (Lemma~\ref{lem:phi}) and by Lemma~\ref{lem:base-prop}
\[
\left|\frac{\partial h}{\partial x_j}\right| \leq w^3 (1 - \bar{x})^{w/2-1}\sqrt{h(x)}
\]

Hence, we conclude that
\[
\frac{\partial^2 \tilde{f}}{\partial x_i \partial x_j} \leq \left|\phi''(h)\frac{\partial h}{\partial x_i} \frac{\partial h}{\partial x_j} \right| \leq \alpha w^6(1 - \bar{x})^{w-2} = \frac{1}{2}\eps (1 - \bar{x})^{w-2}
\]

\begin{proof}[Proof of Theorem~\ref{thm:base-app}]
	In order to make $\tilde{f}$ a submodular function, we need to make the second partial derivatives non-positive. Note that
	$\frac{\partial^2 g}{\partial x_i \partial x_j} = -\frac{w-1}{w}(1- \bar{x})^{w-2}$. Therefore, it suffices to take 
	\[
	\hat{f}(x) = \frac{1}{1+\eps}\left(\tilde{f}(x) + \eps g(x)\right).
	\]
	
	It is clear that $\hat{f}$ is nonnegative and $f(x) \leq \frac{1}{1+\eps}(\tilde{f}(x) + \eps g(x)) \leq \frac{1}{1+\eps}(1 + \eps ) = 1$. The function $\hat{f}$ is monotone since both $\tilde{f}$ and $g$ are monotone. For second partial derivatives, by Lemma~\ref{lem:derivative} one has
	\[
	\frac{\partial^2 f}{\partial x_i \partial x_j} \leq \frac{2}{2+\eps}\left(\frac{1}{2}\eps (1 - \bar{x})^{w-2} - \eps\frac{w-1}{w}(1- \bar{x})^{w-2} \right) \leq 0.
	\] 
	
	Hence, we proved $\hat{f}$ is sbumodular. When $\max_{i, j}|x_i - x_j| \leq \gamma$, one has $h(x) \leq w\gamma(1-\bar{x})^{w-1} \leq w\gamma = \eps_1$. Therefore, $\phi(h(x)) = h(x) =f(x) - g(x)$ and
	\begin{align*}
	\hat{f}(x) =&~ \frac{1}{1+\eps}\left(\tilde{f}(x) + \eps g(x)\right) =  \frac{1}{1+\eps}\left(f(x) - \phi(h(x)) + \eps g(x)\right) \\
	=&~  \frac{1}{1+\eps}\left(g(x) +\eps g(x)\right) = g(x).
	\end{align*}
	While for any $x \in [0,1]$, one has
	\[
	\hat{f}(x) = \frac{1}{1+\eps}\left(\tilde{f}(x) + \eps g(x)\right) \leq \frac{1}{1+\eps}\left(f(x) + \eps f(x)\right) = f(x)
	\]
	and 
	\[
	\hat{f}(x)  \geq \frac{1}{1+\eps}\tilde{f}(x) \geq \frac{1}{1+\eps}(f(x) - \eps_2) \geq f(x) \geq f(x) - \eps.
	\]

	%The choice of parameters, recall $\alpha = 2/\log(1/w\eps)$, for any given $\beta > 0$, we choose $\eps = \frac{1}{k}e^{-4w^6/\beta}$, and $\beta =2\alpha w^6$. 
	
	%And we have $\hat{f}(\bx) \geq \tilde{f}(\bx) \geq f(\bx) - \eps_2 \geq f(\bx) - \beta$, moreover, we have $\hat{f}(\bx) \leq f(\bx) + 2\alpha w^6 g \leq 1 + \beta$.
\end{proof}

We next provide the proof for Lemma~\ref{lem:base2}, which provides basic properties of $\hat{F}$.

\begin{proof}[Proof of Lemma~\ref{lem:base2}]
	Since $\hat{f}(x_i) \in [0,1]$, we have $\hat{F}(x) \in [0,1]$. $\hat{F}$ is monotone, because for any $i \in [m], j \in [w]$,
	\[
	\frac{\partial \hat{F}}{\partial x_{i, j}} = \frac{\partial \hat{f}}{\partial x_{i, j}}\prod_{\ell \neq i}\big(1 - \hat{f}(x_{\ell})\big) \geq 0.
	\]
	Next we show that $\hat{F}$ is submodular. For any $i, i' \in [m], j, j' \in [w]$, when $i \neq i'$, %since $\frac{\partial \hat{f}}{\partial x_{i, j}} \geq 0$ and $\frac{\partial \hat{f}}{\partial x_{i', j'}} \geq 0$, 
	one has
	\[
	\frac{\partial^2 \hat{F}}{\partial x_{i, j}\partial x_{i', j'}}  = - \frac{\partial \hat{f}}{\partial x_{i, j}}\frac{\partial \hat{f}}{\partial x_{i', j'}} \prod_{\ell \neq i, i'}\big(1 - \hat{f}(x_{\ell})\big) \leq 0 
	\]
	given that $\hat{f}$ is monotone.
	When $i = i'$, since $\hat{f}$ is submodular, %$\frac{\partial^2 \hat{F}}{\partial x_{i, j}\partial x_{i', j'}}\leq 0$, 
	one has
	\[
	\frac{\partial^2 \hat{F}}{\partial x_{i, j}\partial x_{i', j'}}  = \frac{\partial^2 \hat{f}}{\partial x_{i, j}\partial x_{i, j'}}\prod_{\ell \neq i}\big(1 - \hat{f}(x_{\ell})\big) \leq 0.
	\]
	Hence, we conclude $\hat{F}$ is submodular.
	
	The second property directly follows from Theorem~\ref{thm:base}.
	%, i.e., when $\max_{i, j, j'}|x_{i, j} - x_{i, j'}| \leq \gamma$, one has $\hat{f}(x_i) = g(x_i)$ and thus, $\hat{F}(x) = G(x)$. 
	For the last property, suppose $x_{i, j} = 1$. Then we have $\smash{\hat{f}(x_i) \geq f(x_i) - \eps = 1 - \eps}$. Hence, 
	\[
	\hat{F}(x) = 1 - \prod_{i \in [s]}(1 - \hat{f}(x_i)) \geq 1 - \eps.
	\]
	This finishes the proof of the lemma.
\end{proof}

\subsection{Missing proof from Section~\ref{sec:lower1-contruction}}

We provide the proof for Lemma~\ref{hehehehe3}.

\begin{proof}[Proof of Lemma~\ref{hehehehe3}]
	Let $S_B$ be the index defined as
	\[
	S_B:=\{i: S\cap B_i \neq \emptyset, i\in [m]\}.
	\]
	We also define index set $S_{A, \pi}$ for the bijection $\pi, \pi'$
	\[
	S_{A, \pi}:= \{i: S\cap A_{\pi(i)} \neq \emptyset, i\in [m]\} \quad \text{and} \quad S_{A, \pi'}:= \{i: S\cap A_{\pi'(i)} \neq \emptyset, i\in [m]\}.
	\]
	Suppose $\pi, \pi'$ satisfies the condition of the Lemma, we conclude that
	\begin{align*}
	S_B \cap S_{A, \pi} = 	S_B \cap S_{A, \pi'} =  \{i:    S\cap B_i \neq \emptyset \text{ and } S\cap (A_{\pi(i)} \cup A_{\pi'(i)})  \neq \emptyset ,  i \in [m]  \}.
	\end{align*}
	
	Consider the following mapping $\sigma: [m]\rightarrow [m]$, $\sigma$ is identity on $[m]\backslash (S_{A, \pi} \cup S_{A, \pi'}\backslash S_B)$, and it forms an one-to-one mapping on $S_{A, \pi} \cup S_{A, \pi'}\backslash S_B$, such that for any $i \in S_{A, \pi}$, $\pi(i) = \pi'(\sigma(i))$. One can prove such mapping always exists when the condition of the Lemma holds. and it is an one-to-one mapping on $[m]$.
	%$S^{B, \pi, \pi'} \subseteq [m]$ 
	%denote the index set 
	%\[
	%S^{B, \pi, \pi'} = \{i:     i \in [m] \text{ and } S\cap (A_{\pi} \cup A_{\pi'(i)}  \neq \emptyset  \}
	%\]
	%We know that for any $i \in S^{B, \pi, \pi'}$
	%Let $\sigma:[m]\rightarrow [m]$ be a bijection and $\sigma\circ \pi  = \pi'$. For any $I \subseteq[m]$, let 
	
	We can extend $\sigma$ and define
	$\sigma(I):=\{\sigma(i): i \in I\}$ for any $I \subseteq [m]$. Since $\sigma$ is an one-to-one mapping on $[m]$, we have that $\sigma$ forms a bijection on $2^{[m]}$ and $|\sigma(I)| = |I|$.
	In order to show $\calG_{\pi}(S) = \calG_{\pi'}(S)$, it suffices to prove 
	\begin{align}
	\sum_{I\subseteq [m]} \beta^{|I|}(1- \beta)^{m - |I|}G(x^{S, I, \pi}) = \sum_{I \subseteq [m]}\beta^{|I|}(1-\beta)^{m - |I|}G(x^{S, I, \pi'}). \label{eq:lower1-ids1}
	\end{align}
	
	Let $\bar{x}^{S, I, \pi}_i := \frac{1}{w}\sum_{j\in [w]}\bar{z}^{S}_{i, j} = \frac{1}{w}\sum_{j\in [w]}\frac{|S\cap B_{i, j}|}{(1 - \alpha)k} =\frac{1}{w}\frac{|S\cap B_{i}|}{(1 - \alpha)k} $ for $i \notin I$ and $\bar{x}^{S, I, \pi}_i := \frac{1}{w}\sum_{j\in [w]}\bar{y}^{S}_{\pi(i), j} = \frac{1}{w}\sum_{j\in [w]}\frac{|S\cap A_{\pi(i), j}|}{\alpha k} = \frac{1}{w}\frac{|S\cap A_{\pi(i)}|}{\alpha k}$. We define $\bar{x}^{S, I, \pi'}_i$ accordingly. Then we have
	
	\begin{align}
	G(x^{S, I, \pi}) =&~ 1 - \prod_{i \in [m]}(1 - g(x^{S, I, \pi}_i))  = 1 - \prod_{i \in [m]}(1 - \bar{x}^{S, I, \pi}_i)^w\notag\\
	=&~ 1- \prod_{i \in [m]\backslash I} \left(1 - \frac{1}{w}\frac{|S\cap B_{i}|}{(1 - \alpha)k}\right)^w\prod_{i\in I} \left(1 - \frac{1}{w}\frac{|S\cap A_{\pi(i)}|}{\alpha k}\right)^w\notag\\
	=&~ 1- \prod_{i \in [m]\backslash I} \left(1 - \frac{1}{w}\frac{|S\cap B_{i}|}{(1 - \alpha)k}\right)^w\prod_{i\in I} \left(1 - \frac{1}{w}\frac{|S\cap A_{\pi'(\sigma(i))}|}{\alpha k}\right)^w\notag\\
	= &~ 1- \prod_{i \in [m]\backslash \sigma(I)} \left(1 - \frac{1}{w}\frac{|S\cap B_{i}|}{(1 - \alpha)k}\right)^w\prod_{i\in \sigma(I)} \left(1 - \frac{1}{w}\frac{|S\cap A_{\pi'(i)}|}{\alpha k}\right)^w\notag\\
	= &~ 1 - \prod_{i \in [m]}(1 - \bar{x}^{S, \sigma(I), \pi'}_i)^w = 1 - \prod_{i \in [m]}(1 - g(x^{S, \sigma(I), \pi'}_i))G(x^{S, \sigma(I), \pi}).\label{eq:lower1-ids2}
	\end{align}
	The first three steps holds due to the definition $G$.
	The fourth step follows from for any $i \in S_{A, \pi}$, $\pi(i) = \pi'(\sigma(i))$, and for $i \notin S_{A, \pi}$, $|S\cap A_{\pi(i)}| = 0 = |S\cap A_{\pi'(\sigma(i))}|$, and therefore $|S\cap A_{\pi(i)}| = |S\cap A_{\pi'(\sigma(i))}|$ holds for all $i \in [m]$.
	The fifth step follows from $\sigma$ is identity on $S_{B}$, and hence,
	\[
	\prod_{i \in [m]\backslash I} \left(1 - \frac{1}{w}\frac{|S\cap B_{i}|}{\alpha k}\right)^w = 	\prod_{i \in [m]\backslash \sigma(I)} \left(1 - \frac{1}{w}\frac{|S\cap B_{i}|}{(1 - \alpha)k}\right)^w.
	\]
	The last three steps follow from the definition of $G$.
	
	Note that Eq.~\eqref{eq:lower1-ids2} implies Eq.~\eqref{eq:lower1-ids1} as $\sigma$ creates an one-to-on mappting from $2^{[m]}$ to $2^{[m]}$ and $|I| = |\sigma(I)|$ for any $I \subseteq [m]$.
\end{proof}

Next, we prove Lemma \ref{hardlemma}.
%On the other side, if the algorithm does not discover the identity of $i_t$, i.e., both queries and the output solution do not include $i_t$, then we claim it can not achieve $(0.584+\eps)$-approximation. 

\begin{proof}[Proof of Lemma~\ref{hardlemma}]
Let $S$ be a set that is balanced, and we assume it consists of $\lambda k$ element from $B_{i^{*}}$ and $(1 - \lambda k)$ from the $A_1, \cdots, A_{m}$. 
%We bound the optimum value. 

For any $I \subseteq [m]$, suppose $i^{*} \in I$. For any $i \notin I$, we know that $\bar{x}_{i}^{S, I,\pi} := \frac{1}{w}\sum_{j \in[w]}\frac{|B_{i, j}\cap S|}{(1 -\alpha)k} = 0$ and for any $i \in I$, $\bar{x}_{i}^{S, I, \pi} = \bar{y}^{S}_{\pi(i)} := \sum_{j\in [w]}\frac{|S \cap A_{\pi(i), j}|}{\alpha k}$, hence, we have 
% then we know that $A_i$ is not contained, and there is no other $B_j$, hence, we have
\begin{align}
\hat{F}(x^{S, I, \pi}) = G(x^{S, I, \pi}) = &~ 1 - \prod_{i \in I}\left(1 - \frac{\bar{x}_{i}^{S, I, \pi}}{\alpha w k} \right)^{w} \notag\\
\leq &~ 1 - \prod_{i \in I}\exp\left(-w(1+\eps)\frac{\bar{x}_{i}^{S, I, \pi}}{\alpha w k} \right)\notag \\
= &~ 1 - \exp\left(-(1+\eps)\frac{\sum_{i \in I}\bar{x}_{i}^{S, I, \pi}}{\alpha k} \right)\notag\\
\leq &~ 1 - \exp\left(-\frac{\sum_{i \in I}\bar{x}_{i}^{S, I, \pi}}{\alpha k} \right)+ \eps.\label{eq:lower1-1}
\end{align}
The first step holds since $S$ is balanced, the second step follows from the definition. The third step follows from $1 - x \geq e^{-(1+\eps)x}$ when $x \leq \eps/10$ and $\frac{x_{i}^{S, I, \pi}}{\alpha w k} \leq \frac{(1 -\lambda)k}{\alpha w k} \leq \frac{1}{\alpha w} \leq \eps/10$. The last step follows from $\exp(-c(1 + \eps)) \geq \exp(-c) - \eps$ for any $\eps > 0$ and $c>0$.

For any $i \in [m]$, let $X^{S, I, \pi}_{i}$ be a random variable such that $X^{S, I, \pi}_{i}= \bar{x}^{S, I, \pi}_{i}$ with probability $\beta$ and $X^{S, I, \pi}_{i}= 0$ with probability $1 - \beta$. Let $K(x): \R^{m}\rightarrow \R$, and for any $x \in \R^{m}$, $K(x) = 1 - \exp(\sum_{i=1}^{m}x_{i} / \alpha k)$. It is easy to verify that $K$ is concave and one has
\begin{align}
\sum_{I \subseteq [m], i^{*} \in I}\beta^{|I|}(1 - \beta )^{m - |I|} \hat{F}(x^{S, I, \pi}) \leq &~ \sum_{I \subseteq [m], i^{*}\in I}\beta^{|I|}(1 - \beta )^{m - |I|}  \left(1 - \exp\left(-\frac{\sum_{i \in I}\bar{x}_{i}^{S, I, \pi}}{\alpha k} \right)+ \eps\right)\notag\\
= &~ \sum_{I \subseteq [m], i^{*}\in I}\beta^{|I|}(1 - \beta )^{m - |I|} \left(1 - \exp\left(-\frac{\sum_{i \in I}\bar{x}_{i}^{S, I, \pi}}{\alpha k} \right)\right)+\beta\eps\notag\\
= &~ \beta \E_{X^{S, I, \pi}_{1}, \cdots, X^{S, I, \pi}_{m}} \left[K(X^{S, I, \pi}_{i}, \cdots, X^{S, I, \pi}_{i})\right] + \beta\eps\notag\\
\leq &~ \beta K(\bar{x}^{S, I, \pi}_{1}, \cdots, \bar{x}^{S, I, \pi}_{m})+ \beta\eps\notag\\
= &~ \beta \left( 1 - \exp\left(\frac{-(1- \lambda)\beta}{\alpha} \right)  \right) + \beta\eps. \label{eq:lower1-2}
\end{align}
The first step follows from Eq.~\eqref{eq:lower1-1}, the third step follows from $\pi(i^{*}) \notin S$, and therefore, $\bar{x}_{i^{*}}^{S, I, \pi} = 0$ when $i^{*} \in I$. The fourth step follows the Jenson's inequality and the concavity of $K$. The last step follows from $\sum_{i \in[m]}\bar{x}^{S, I, \pi}_{i}= (1 - \lambda)k$.

Similarly, when $i^{*} \notin I$, we have $\bar{x}_{i}^{S, I, \pi} = \bar{y}^{S}_{\pi(i)} := \sum_{j\in [w]}\frac{|S \cap A_{\pi(i), j}|}{\alpha k}$ for $i \in I$, $\bar{x}_{i^{*}}^{S, I, \pi}= \bar{z}_{i^{*}}^{S} = \frac{\lambda k }{(1 -\alpha)kw} = \frac{\lambda}{(1 -\alpha)w}$ and $\bar{x}_{i}^{S, I, \pi}  = \frac{1}{w}\sum_{j \in w}\frac{|B_{i, j}\cap S|}{(1 -\alpha)k} = 0$ for $i \in [m] \backslash (I \cup \{i^{*}\})$. Therefore,
\begin{align}
\hat{F}(x^{S, I, \pi}) = G(x^{S, I, \pi})  = &~ 1 - \prod_{i \in I}\left(1 - \frac{\bar{x}^{S, I, \pi}_{i}}{\alpha w k} \right)^{w}\left(1 -\bar{x}^{S, I, \pi}_{i^{*}}\right)^{w} \notag\\
=&~ 1 - \prod_{i \in I}\left(1 - \frac{\bar{x}^{S, I, \pi}_{i}}{\alpha w k} \right)^{w}\left(1 - \frac{\lambda}{(1- \alpha)w}\right)^{w} \notag\\
\leq &~ 1 - \prod_{i \in I}\exp\left(-w(1+\eps)\frac{\bar{x}^{S, I, \pi}_{i}}{\alpha w k} \right)\exp\left(-w(1+\eps)\frac{\lambda}{(1 - \alpha)w} \right)\notag\\
= &~ 1 - \exp\left(-(1+\eps)\frac{\sum_{i \in I}\bar{x}^{S, I, \pi}_{i}}{\alpha  k} \right)\exp\left(-(1+\eps)\frac{\lambda}{1 - \alpha} \right)\notag\\
\leq  &~ 1 - \exp\left(-\frac{\sum_{i \in I}\bar{x}^{S, I, \pi}_{i}}{\alpha  k} \right)\exp\left(-\frac{\lambda}{1 - \alpha} \right) + 2\eps. \label{eq:lower1-3}
\end{align}
The fourth step follows from $1 - x \geq e^{-(1+\eps)x}$ when $x \leq \eps/10$, $\frac{\bar{x}_{S, A, i}}{\alpha w k} \leq \frac{(1 -\lambda)k}{\alpha w k} \leq \frac{1}{\alpha w} \leq \eps/10$ and $\frac{\lambda}{(1-\alpha)w} \leq \frac{1}{(1-\alpha)w} \leq O(\eps)$. The sixth step follows from $\exp(-c(1 + \eps)) \geq \exp(-c) - \eps$ for any $\eps > 0$ and $c>0$.

Hence, we have
\begin{align}
&~\sum_{I \subseteq [m], i_t \notin I}\beta^{|I|}(1 - \beta )^{m - |I|} \hat{F}(x^{S, I, \pi}) \notag\\
\leq&~ \sum_{I \subseteq [m], i\notin I}\beta^{|I|}(1 - \beta )^{m - |I|}\left(1 - \exp\left(-\frac{\sum_{i \in I}x^{S, I, \pi}_i}{\alpha  k} \right)\exp\left(-\frac{\lambda}{1 - \alpha} \right) + 2\eps\right)\notag\\
=&~ \sum_{I \subseteq [m], i\notin I}\beta^{|I|}(1 - \beta )^{m - |I|}\left(1 - \exp\left(-\frac{\sum_{i \in I}x^{S, I, \pi}_i}{\alpha  k} \right)\exp\left(-\frac{\lambda}{1 - \alpha} \right) \right) + 2(1 - \beta)\eps\notag\\
\leq &~ (1 - \beta) \left(1 - \exp\left(\frac{(1- \lambda)\beta}{\alpha}\right)\exp\left( - \frac{\lambda}{1-\alpha}\right)\right) + 2(1 - \beta)\eps.\label{eq:lower1-4}
\end{align}
The first step holds due to Eq.~\eqref{eq:lower1-3}, the third step holds due to concavity and Jensen's inequality.

Summing these up, we have
\begin{align*}
\mathcal{F}_{c, \pi}(S) =&~\min\left\{ \sum_{I \subseteq [m]}\beta^{|I|}(1 - \beta )^{m - |I|} \hat{F}(x^{S, I, \pi}) + \frac{\eps}{k}|S|, 1 \right\}\\
= &~\sum_{I \subseteq [m]}\beta^{|I|}(1 - \beta )^{m - |I|} \hat{F}(x^{S, I, \pi}) + \eps\\
= &~ \sum_{I \subseteq [m], i^{*} \in I}\beta^{|I|}(1 - \beta )^{m - |I|} F(x^{S, I, \pi}) + \sum_{I \subseteq [m], i^{*} \notin I}\beta^{|I|}(1 - \beta )^{m - |I|} \hat{F}(x^{S, I, \pi})\\
\leq &~ \beta \left( 1 - \exp\left(\frac{-(1- \lambda)\beta}{\alpha} \right)  \right) + \beta\eps+ (1 - \beta) \left(1 - \exp\left(\frac{(1- \lambda)\beta}{\alpha}\right)\exp\left( - \frac{\lambda}{1-\alpha}\right)\right) + 2(1 - \beta)\eps\\
\leq &~ \beta \left( 1 - \exp\left(\frac{-(1- \lambda)\beta}{\alpha} \right)  \right) + (1 - \beta) \left(1 - \exp\left(\frac{(1- \lambda)\beta}{\alpha}\right)\exp\left( - \frac{\lambda}{1-\alpha}\right)\right) + 3\eps\\
:= &~ F(\alpha, \beta,\lambda) + 3\eps.
\end{align*}
The fourth step holds due to Eq.~\eqref{eq:lower1-2}~\eqref{eq:lower1-4}.

It remains to solve
\[
\min_{\alpha, \beta \in (0,1)}\max_{\lambda \in (0,1)} F(\alpha, \beta, \lambda).
\]
Recalls that $\alpha$ and $\beta$ is up to our choice, where $\lambda$ is chosen by the algorithm to maximize its approximation ratio, by taking $\alpha = 0.56$, $\beta =0.42$, one can verify that 
\[
Q(\alpha, \beta) := \max_{\lambda \in (0,1)} F(\alpha, \lambda, \beta) < 0.5839.
\]

\end{proof}

\section{Missing proof from Section~\ref{sec:lower2}}
\label{sec:lower2-app}

\subsection{Missing proof from Section~\ref{sec:lower2-construction}}

We first provide the proof of Lemma~\ref{usefulproperties}, which provides basic properties of $\bR$.
\begin{proof}[Proof of Lemma~\ref{usefulproperties}]
	For the first property, for any node $u = (u_1, \ldots, u_{\ell})$ at depth $\ell$, $u \in \bR$ iff (1) none of the nodes on the leaf-to-root path get sampled, i.e.,  $u_1 \notin \bR, \ldots, (u_1, \ldots, u_{\ell-1})\notin \bR$, (2) $u$ is sampled. The probability can be caluclated as 
	\[
	\prod_{i =1}^{\ell-1}(1 - p_i) p_{\ell} = \prod_{i =1}^{\ell-1}\frac{1 - \sum_{j =1}^{i}w_j}{1 - \sum_{j =1}^{i-1}w_j} \cdot \frac{w_\ell}{1 - \sum_{i=0}^{\ell-1}w_{i}} = w_\ell .
	\] 
	The second property follows directly from the $\Sample$ procedure and the fact that $p_{L} =1$.
\end{proof}

We first provide the proof of Lemma~\ref{lem:basicbasic}, which provides basic properties of $\calF$.
\begin{proof}[Proof of Lemma~\ref{lem:basicbasic}]
	For the first claim, for any $S\subseteq V$ and any $R$ in the support of $\mathcal{D}$, it is clear that $g_{R}$ is (continuous) monotone and submodular. By the definition of $x_{u}^{S}$, this implies monotonicity and submodularity on the discrete function. As addition and min operation keep monotonicity and submodularity, we have that $\mathcal{F}$ is monotone and submodular.
	For the second claim, when $|S| \geq k/\eps$, one has $\mathcal{F}(S) \geq\min\{\eps|S|/k, 1\} = 1$. 
\end{proof}

%The next lemma is critical, it shows the optimal value is at most $\frac{1}{2} + H(\eps)$ when it does not include any element from the optimal path. 

Recall we define the objective function $\mathcal{F}: 2^{V}\rightarrow \R^{+}$ as 
\begin{align*}
\mathcal{F}(S) =\min\left\{  \calG(x^{S})+ \frac{\eps}{k}|S|, \hspace{0.04cm}1\right\} =  \min\left\{  \E_{\bR\sim \calD} \big[g_{\bR}(x^{S} )\big]+ \frac{\eps}{k}|S|, \hspace{0.04cm}1\right\},
\end{align*}
and
\begin{align*}
g_{R}(x) :=1 - \prod_{u\in R}\left(1 - x_u\right)\quad\text{and}
\end{align*}

The following Lemma can be seen an extension of Lemma 5.3 in \cite{feldman2020one}.
\begin{lemma}[Restatement of Lemma~\ref{lem:lower2-2}]
	\label{lem:lower2-app}
	For any leaf $u\in U_L$  and $S\subseteq W_u$ with $|S|\le k$ and 
	$S\cap \cal
	=\emptyset$, 
	we have
	\[
	\mathcal{F}(S) \leq 0.5+O\left(\eps\log^2(1/\eps)\right).
	\]
\end{lemma}

For any $\ell\in [L]$, define $A_{\geq i} = \sum_{\ell=i}^{L}a_i$, where $a_i$ is defined in Definition~\ref{def:weight-sequence}. We know that $a_1 = 1$. Here are some simple facts of $\{a_i\}_{i \in [L]}$ and $\{A_{\geq i}\}_{i \in [L]}$. 

\begin{lemma}[Lemma 5.2 in \cite{feldman2020one}]
	\label{lem:tech1-app}
	The weights $w_1, \ldots, w_{L} > 0$, and for every $j = 1,\cdots, L$,
	\begin{align}
	a_j \prod_{i=1}^{j-1}\left(1 - \frac{a_i}{A_{\geq i}}\right) =&~ 1, \label{eq:lower2-2-app}\\
	\sum_{i=j}^{L}(2 - \frac{a_i}{A_{\geq i}}) = &~ \frac{A_{\geq j}}{a_j},\label{eq:lower2-3-app} \\
	2j - H_j \leq \frac{A_{\geq L -j + 1}}{a_{L-j + 1}} \leq &~ 2j -1.\label{eq:lower2-4-app}
	\end{align} 
	Here, we use $H_j$ to denote the harmonic number of $\sum_{i=1}^{j}\frac{1}{i}$.
\end{lemma}

\begin{proof}[Proof of Lemma~\ref{lem:lower2-app}]
	 %Let $W' = W\backslash (A_{i_1} \cup \cdots \cup A_{i_1, \cdots, i_{L}})$. For any $\ell \in [L]$ and $T_1\subseteq [m_1], \cdots, T_{\ell} \subseteq [m_\ell]$, define
	 %\begin{align*}
	 %\Omega_{T_1, \cdots, T_{\ell}} =&~ \Omega(A_{i_1, \cdots, i_{\ell}}) \bigcap \left(\bigcap_{j_1\in T_1} \Omega(A_{j_1})  \bigcap_{j_1' \in [m_1]\backslash T_1} \bar{\Omega}(A_{j_1'})\right) \\
	 %&~\bigcap \cdots \bigcap \left(\bigcap_{j_\ell \in [m_{\ell}]} \Omega(A_{i_1,\cdots i_{\ell-1}, j_{\ell}})  \bigcap_{j_\ell' \in [m_{\ell}]\backslash T_{\ell}} \bar{\Omega}(A_{i_1,\cdots i_{\ell-1}, j_{\ell}'}) \right)
	 %\end{align*}
	 %All we need to know about this overwhelming definition is that
	 For any $\ell \in \{1, \ldots, L-1\}$, and $T_1 \subseteq [m_1]\backslash \{u_1\}, \ldots, T_{\ell} \in [m_{\ell}] \backslash \{u_{\ell}\}$, the $\Sample$ procedure guarantees
	 \begin{align}
	 \label{eq:pr1}
	 \Pr_{\bR\sim \calD}\left[ \bR\cap (W_{u}\backslash \calA_u) = \bigcup_{i \in [\ell]}\bigcup_{j \in T_i} A_{u_1, \ldots, u_{\ell-1}, j}  \right] = w_{\ell} \prod_{i =1}^{\ell}\left (\frac{a_i}{A_{\geq i}} \right)^{|T_i|} \left(1 - \frac{a_i}{A_{\geq i}}\right)^{m_i - |T_i|}.
	 \end{align}
	 %\begin{align}
	 %\label{eq:pr1}
	 %\mu(\Omega_{T_1, \cdots, T_{\ell}} ) = &~ w_{\ell} = w_{\ell} \prod_{i =1}^{\ell}\left (\frac{a_i}{A_{\geq i}} \right)^{|T_i|} \left(1 - \frac{a_i}{A_{\geq i}}\right)^{m_i - |T_i|}
	 %\end{align}
	 Let $R$ be in the support of $\calD$ and $R\cap (W_{u}\backslash \calA_u) = \bigcup_{i \in [\ell]}\bigcup_{j \in T_i} A_{u_1, \ldots, u_{\ell-1}, j}$, then for any fixed subset $S\subseteq W_{u}\backslash \calA_{u}$, the function value is determined only be the restriction $S\cap \bigcup_{i \in [\ell]}\bigcup_{j \in T_i} A_{u_1, \ldots, u_{\ell-1}, j}$, as
	 \begin{align}
	 \label{eq:pr2}
	 g_R(x^{S})= 1 - \prod_{i=1}^{\ell}\prod_{j \in T_i}\left(1 - \frac{|A_{u_1, \cdots, u_{\ell-1}, j}\cap S|}{\eps k}\right) .
	 \end{align}
	%The first property holds due to the total independence, and the second property holds by the definition.
	 
	 Therefore, for any $S\subseteq W_u$, we have that
	\begin{align*}
	&~\calG(x^{S})\\ 
	 =&~ \E_{\bR\sim \calD} \big[g_{\bR}(x^{S} )\big]\\
	= &~ \sum_{\ell =1}^{L-1} \sum_{T_1\subseteq [m_1]\backslash \{u_1\}} \cdots\sum_{T_\ell \subseteq [m_\ell] \backslash \{u_{\ell} \}}  \Pr\left[ R\cap (W_{u}\backslash \calA_u) = \bigcup_{i \in [\ell]}\bigcup_{j \in T_i} A_{u_1, \ldots, u_{\ell-1}, j}  \right]\cdot g_R(x^{S})\\
	 &~ +  \sum_{T_1\subseteq [m_1]\backslash \{u_1\}} \cdots\sum_{T_{L-1} \subseteq [m_{L-1}] \backslash \{u_{L-1} \}}    \Pr\left[ R\cap (W_{u}\backslash \calA_u) = \bigcup_{i \in [\ell]}\bigcup_{j \in T_i} A_{u_1, \ldots, u_{\ell-1}, j} \bigcup A_{u_1, \ldots, u_{L}} \right]\cdot g_R(x^{S})\\
	 \leq &~ \sum_{\ell =1}^{L-1} \sum_{T_1\subseteq [m_1]\backslash \{u_1\}} \cdots\sum_{T_\ell \subseteq [m_\ell] \backslash \{u_{\ell} \}}   \Pr\left[ R\cap (W_{u}\backslash \calA_u) = \bigcup_{i \in [\ell]}\bigcup_{j \in T_i} A_{u_1, \ldots, u_{\ell-1}, j}  \right]\cdot g_R(x^{S}) + w_{L}\\
	= &~ \sum_{\ell =1}^{L-1} \sum_{T_1\subseteq [m_1]\backslash \{u_1\}} \cdots\sum_{T_\ell \subseteq [m_\ell] \backslash \{u_{\ell} \}}   w_{\ell} \prod_{i=1}^{\ell}\left (\frac{a_i}{A_{\geq i}} \right)^{|T_i|} \left(1 - \frac{a_i}{A_{\geq i}}\right)^{m_i - |T_i|}  \left(1 - \prod_{i=1}^{\ell}\prod_{j \in T_i}\left(1 - \frac{|A_{u_1, \cdots, u_{\ell-1}, j}\cap S|}{\eps k}\right)\right)\\
	&~ + w_{L}  \\
	 = &~ \sum_{\ell=1}^{L-1}w_\ell\left(1 - \prod_{i= 1}^{\ell}\prod_{j \in [m_i]} \left(1 - \frac{a_i}{A_{\geq i}}\frac{|A_{u_1, \cdots, u_{\ell-1},j} \cap S|}{\eps k}  \right)\right) + w_{L} := \tilde{f}(S).
	\end{align*}
	The third step holds since the second term is bounded by $w_{L}$. The fourth step follows from Eq.~\eqref{eq:pr1}~\eqref{eq:pr2} and the fifth step follows from calculations.

	For any $S\subseteq W_{u}\backslash \calA_u$, $\ell \in \{1, \ldots, L-1\}$, $j \in [m_\ell]$, let $y_{\ell, j}^{S} = \frac{|A_{u_1, \cdots, u_{\ell-1},j} \cap S|}{\eps k} $. Define the feasible region $\tilde{D}\subseteq \R^{m_1 + \cdots m_{L-1}}$ as
	\[
	\tilde{D} = \left\{  (y_{1, 1}, \ldots, y_{1, m_1}, \ldots, y_{L-1, 1}, \ldots, y_{L-1,m_{L}})\in \R^{m_1+\cdots m_{L-1}}_{+}: \sum_{\ell\in L}\sum_{j \in [m_\ell]}y_{\ell, j} \leq L, y_{\ell,j} \leq1       \right\}.
	\] 
	
	Different from Lemma 5.4 of \cite{feldman2020one}, a direct continuous relaxation of $\tilde{f}$ is not concave. Instead, we consider a concave upper bound $\hat{f}$ on it
	\begin{align}
	\tilde{f}(S) = &~ w_L + \sum_{\ell=1}^{L-1}w_{\ell} \left(1 - \prod_{i = 1}^{\ell}\prod_{j \in [m_i]} \left(1 - \frac{a_i}{A_{\geq i}}y_{i, j}^{S} \right)\right) \notag\\
	\leq &~ w_{L} + \sum_{\ell=1}^{L-1}w_\ell \left(1 - \prod_{i = 1}^{\ell}\prod_{j \in [m_i]} \left(1 - \frac{a_i}{A_{\geq i}}\right)^{y_{i, j}^{S}} \right) \notag\\
	= &~w_{L} + \sum_{\ell=1}^{L-1}w_\ell \left(1 - \prod_{i = 1}^{j} \left(1 - \frac{a_i}{A_{\geq i}}\right)^{\sum_{j\in [m_i]}y_{i, j}^{S}} \right) := \hat{f}(S)\label{eq:lower2-1-app}
	\end{align}
	The second step holds since $(1 - as) \geq (1 - a)^{s}$ holds for any $a, s\in [0,1]$ (see Claim~\ref{clm:tech1} for verification).
	
	One can extend $\hat{f}$ to a continuous concave function $\hat{F}: \R^{L-1}_{+} \rightarrow \R_{+}$. Concretely, let the new feasible region be $\hat{D} \subseteq \R^{L-1}$, $\hat{D} = \{(z_1, \ldots, z_{L-1}) \in \R_{+}^{L-1}, \sum_{\ell =1}^{L-1}z_{\ell} \leq L\}$. For intuition, one should think $z_\ell = \sum_{j \in [m_\ell]}y_{\ell, j}$. For any $z = (z_1, \ldots, z_{L-1})$, we define
	\[
	\hat{F}(z) = w_{L} + \sum_{\ell=1}^{L-1}w_\ell \left(1 - \prod_{i = 1}^{j} \left(1 - \frac{a_i}{A_{\geq i}}\right)^{z_{\ell}} \right) 
	\]
	One can verify that $\hat{F}$ is monotone and concave, since each term in the summation is concave and monotone. 
	
	For any $S\subseteq W_{u}\backslash \calA_u$, $\ell \in \{1, \ldots, L-1\}$, let $z^{S}_{\ell} = \sum_{j\in [m_l]}y_{\ell, j}^{S} = \sum_{j\in [m_l]}\frac{|A_{u_1, \cdots, u_{\ell-1},j} \cap S|}{\eps k} $.
	We know that $z^{S} \in \hat{D}$ for any $S\in W_u\backslash \calA_u$.
	Take an arbitrary tuple $(u_1', \ldots, u_{L-1}') \in [m_1]\times \cdots [m_{L-1}]$ such that $u_{\ell} \neq u_{\ell}'$ holds for all $\ell \in \{1, \ldots, L-1\}$, and let $S = A_{u_1'}\cup \cdots \cup A_{u_1, \ldots u_{L-2},u_{L-1}'}\cup A_{u_1, \ldots, u_{L}}$. It is easy to verify that $S\subseteq W_u \backslash \calA_u$, $|S| = k$. Furthermore, we have
	\begin{align}
	\max_{z \in \hat{D}} \hat{F}(z) \leq &~ \hat{F}(\mathbf{1}) + \max_{z\in \hat{D}}\langle \nabla \hat{F}(\mathbf{1}), z- \mathbf{1} \rangle = \tilde{f}(S) + \max_{z\in \hat{D}}\langle \nabla \hat{F}(\mathbf{1}), z - \mathbf{1} \rangle \notag\\
	= &~ w_{L} + \sum_{\ell=1}^{L-1}w_\ell \left(1 - \prod_{i = 1}^{j} \left(1 - \frac{a_i}{A_{\geq i}}\right)\right) +  \max_{z\in \hat{D}}\langle \nabla \hat{F}(\mathbf{1}), z - \mathbf{1} \rangle\notag\\
	\leq &~ \frac{L}{2L- H_L} + \max_{z\in \hat{D}}\langle \nabla \hat{F}(\mathbf{1}), z - \mathbf{1} \rangle\notag\\
	= &~ \frac{1}{2} + O(\eps\log(1/\eps)) + + \max_{z\in \hat{D}}\langle \nabla \hat{F}(\mathbf{1}), z - \mathbf{1} \rangle\label{eq:lower2-5-app}
	\end{align}
	The first step follows from the concavity, the second step follows from the definition, the third step holds due to Eq~\eqref{eq:lower2-1-app}, the fourth step holds due to Eq.~\eqref{eq:lower2-2-app} and the fact that $A_{\geq 1} \geq 2L - H_{L}$ (taking $j = L$ in Eq.~\eqref{eq:lower2-4-app}). The last step holds due to $L = 1/\eps$.
	
	 It remains to bound the second term as
	\begin{align}
	\max_{z\in \hat{D}}\langle \nabla \hat{F}(\mathbf{1}), z - \mathbf{1} \rangle =&~ L \max_{\ell \in [L-1]}\frac{\partial \hat{F}}{\partial z_{\ell}} - \sum_{\ell=1}^{L-1}\frac{\partial \hat{F}}{\partial z_{\ell}}(\mathbf{1})\notag\\
	\leq &~ \frac{1}{A_{\geq 1}}\left(\frac{L}{2} - \sum_{\ell=1}^{L-1}\left(\frac{1}{2} - \frac{H_{L+1-\ell}}{L+1 - \ell} \right)\right)\notag\\
	= &~\frac{1}{A_{\geq 1}}\sum_{\ell=1}^{L}\frac{H_{\ell}}{\ell}\notag\\
	 \leq &~ \frac{1}{A_{\geq 1}}\log^2( L) \lesssim \eps \log^2(1/\eps)\label{eq:lower2-6-app}
	\end{align}
	The second step follows from Lemma~\ref{lem:tech2-app}, the last step holds due to $L = 1/\eps$, $A_{\geq 1} \geq 2L-1 = 2/\eps -1$ (see Lemma~\ref{lem:tech1-app}). Finally, we conclude the proof by noticing for any $S \subseteq W'$, $|S| = k$,
	\begin{align*}
	\mathcal{F}(S) =&~ \min\left\{  \mathcal{G}(x)   + \frac{\eps}{k}|S|, 1\right\} \leq \E_{\bR\sim \calD}[g_{R}(x^{S})]+\eps = \tilde{f}(S) +\eps \leq\max_{z\in \hat{D}}\hat{F}(z) + \eps\\
	\leq &~ \frac{1}{2} + O\left(\eps \log^2(1/\eps)\right)
	\end{align*}
	The second step holds due to $|S| \leq k$, the fourth step holds due to Eq.~\eqref{eq:lower2-1-app}, the fifth step holds due to Eq.~\eqref{eq:lower2-5-app}  and Eq.~\eqref{eq:lower2-6-app}.
\end{proof}

\begin{lemma}[Claim 5.4 in \cite{feldman2020one}]
	\label{lem:tech2-app}
	For any $\ell = 1,\ldots, L$, one has 
	\[
	\frac{1}{A_{\geq 1}} \left(\frac{1}{2} - \frac{H_{L+1 - \ell}}{L+1-\ell} \right) \leq \frac{\partial \hat{F}}{\partial z_{\ell}}(\mathbf{1}) \leq \frac{1}{A_{\geq 1}}\cdot \frac{1}{2} 
	\]
\end{lemma}

\begin{claim}
	\label{clm:tech1}
	For any $a, s\in [0, 1]$, $(1 -as) \geq (1 - a)^{s}$.
\end{claim}
\begin{proof}
	
	We fix $a \in [0,1]$ and define $h(s ) = 1 -as  - (1 - a)^{s}$. Notice that $h(s) = 0$ when $s = 0$ and $s= 1$. Hence, it suffices to prove the $h'(s)$ is positive at the $s=0$ and it is monotone decreasing. This can be true, as 
	\[
	h'(s) = 1 - a - \log(1-a)(1-a)^{s}
	\]
	when $s=0$, $h'(0) = 1 - a - \log(1 - a) \geq 0$ and $h'(s)$ is monotone decreasing with $s$.
	\end{proof}

\subsection{Missing proof from Section~\ref{sec:lower2-lower}}

We next provide the proof for Claim~\ref{blablaclaim}, which gives a lower bound on the $d^{\ell + 1}$-stage game for finding hidden bijections.
\begin{proof}[Proof of Claim~\ref{blablaclaim}]
	The hidden bijections are drawn independently and uniformly.
	Let $v = (v_1, \ldots, v_{\ell+1})\in [d]^{\ell+1}$ and let $X_v$ denote the number of query for the algorithm to recover $\pi_{v_1, \ldots, v_{\ell}}(v_{\ell+1})$. It is easy to see that $\Pr[X_v | \{X_{v'}\}_{v'\in [d]^{\ell+1}\backslash\{v\}} ] \leq \frac{m_{\ell+1}}{16L} \leq \frac{1}{4L}$.
	Let $Y_v$ be an indicator such that $Y_v = \mathsf{1}\{X_v \leq \frac{m_{\ell+1}}{16L} \}$. Then by Azuma-Hoeffding bound, one has
	\begin{align*}
	\Pr\left[\sum_{v\in [d]^{\ell+1}} Y_v  \geq \frac{1}{2L} d^{\ell+1}\right] \leq \exp(-d^{\ell+1}/12L) \leq 1/n^4.
	\end{align*}
	Hence, we conclude that with probability at least $1 - 1/n^4$, the number of stage that requires less than $\frac{m_{\ell+1}}{16L} $ queries is no more than $\frac{d^{\ell+1}}{2L}$, hence the algorithm needs to pay at least $\frac{m_{\ell+1}}{16L} \cdot \frac{d^{\ell+1}}{2L} = \Omega(m_{\ell+1}d^{\ell+1})$ queries.
	
\end{proof}

\section{Missing proof from Section~\ref{sec:insert-cardinality}}
\label{sec:insert-cardinality-app}

We provide a generic reduction and show that one can assume knowing the value of $\OPT$ in dynamic submodular maximization. This brings an overhead of $O(\eps^{-1}\log(k/\eps))$ on amortized running time.
\begin{lemma}
	\label{lem:opt}
	Consider dynamic submodular maximization in an insertion-only stream, suppose there is an algorithm that maintains an $\alpha$-approximate solution given an estimate of $\OPT$ within a multiplicate error of $(1+\eps)$, and makes $O(\mathcal{T})$ amortized queries per update. Then, there is an algorithm that achieves $(\alpha-\eps)$-approximation (without knowing $\OPT$) with amortized query complexity $O(\mathcal{T}\log(k/\eps)/\eps)$.
\end{lemma}
\begin{proof}
	Let $\mathsf{v}_t$ be the maximum value of a singleton at time $t$, i.e. $\mathsf{v}_t = \max_{e\in V_t}f(e)$. In an insertion-only stream, $\{\mathsf{v}_t\}_{t\geq 1}$ is monotone and we assume $\mathsf{v}_1 = 1$ for simplicity.
	The algorithm maintains multiple copies of the baseline algorithm, where the $i$-th copy guesses $(1+\eps)^{i}$ on $\OPT$. At time $t$, the algorithm instantiate threads $i_t, \ldots, i_t + \eps^{-1}\log k +1$, where $ i_t = \lfloor \log_{1+\eps} \mathsf{v}_t \rfloor$, and it returns the maximum solution set among these threads. 
	For amortized query complexity, the algorithm maintains $O(\eps^{-1}\log k)$ threads at each step and for each element in the stream, it is covered by at most $\eps^{-1}\log (k/\eps)$ threads. To see the later, consider the $t$-th element in the stream, it is taken account by threads $i_t, \ldots, i_t + \eps^{-1}\log (k/\eps)$. For thread $i > i_t + \eps^{-1}\log (k/\eps)$, it doesn't need to consider element before time $t$, as their marginal contribution can be at most $k \cdot v_t  \leq k (1+\eps)^{i_t+1}\leq \eps (1+\eps)^{i}$ in the solution set. Hence, we lose at most $\eps$ factor in the approximation ratio and the amortized query complexity is at most $O(\mathcal{T}\log(k/\eps)/\eps)$.
	For approximation guarantee. We know $\mathsf{v}_t \leq \OPT_t \leq k\mathsf{v}_t$, and therefore, one of the thread will guess $\OPT_t$ within $(1+\eps)$ error. Since we only filter element smaller than $\eps\OPT_t/k$, the algorithm maintains an $(\alpha-\eps)$-approximation solution. 
\end{proof}

\section{Missing proof from Section~\ref{sec:insert-matroid}}
\label{sec:insert-matroid-app}

We first prove the following fact.
\begin{fact}[Restatement of Fact~\ref{fact111}]
	\label{fact:submodular1}
	Suppose $f$ is a monotone submodular funciton, the for any subset $A, B, C$, one has
	\[
	f_{A}(C) \leq f_{A\cup B}(C) + f_{A}(B)
	\]
\end{fact}

\begin{proof}
	One has
	\begin{align*}
	f_{A}(C) - f_{A}(B) =&~ f(A\cup C) - f(A) - f(A\cup B) + f(A) = f(A\cup C) - f(A\cup B) \\
	\leq&~ f(A\cup B\cup C) - f(A\cup B) = f_{A\cup B}(C)
	\end{align*}
	Hence, we conclude
	\[
	f_{A}(C) \leq  f_{A\cup B}(C) + f_{A}(B).
	\]
\end{proof}

We next prove Claim~\ref{claim111}.
\begin{proof}[Proof of Claim~\ref{claim111}]
	For any $e \in E_\ell$, by definition, we have $T_1\cup \cdots  \cup T_{\ell}\cup S_{\ell+1} \cup e \notin M$ and $e \in O_{\ell + 2} \cup \cdots \cup O_{L}$. Due to Lemma~\ref{lem:folklore}, $T_1\cup \cdots  \cup T_{\ell+1} \cup O_{\ell + 2} \cup \cdots \cup O_{L}\in M$, and therefore, by the augmentation property, we know that
	\begin{align*}
	|E_\ell| \leq &~ |O_{\ell+2}\cup \cdots \cup O_{L}| - (|T_1\cup \cdots \cup T_{\ell+1} \cup O_{\ell+2}\cup \cdots \cup O_{L}| - |T_1\cup \cdots \cup T_{\ell}\cup S_{\ell+1}|) \\
	=&~ |S_{\ell+1} \backslash T_{\ell+1}|.\qedhere
	\end{align*}
\end{proof}

We next prove Lemma~\ref{lem:matroid}.
\begin{proof}[Proof of Lemma~\ref{lem:matroid}]
	For the approximation guarantee, we prove one of the branches will yield a good solution. By Lemma~\ref{lem:matroid-approx}, we know there is a branch $a\in \mathcal{A}$ that returns a solution set $S_a$ and satisifes
	\[
	f(S_a) \geq (1-O(\eps))f_{S_a}(O).
	\]
	for any feasible set $O \subseteq V_t$. Taking $O = O_t$, we get
	\[
	f(S_a) \geq (1-O(\eps))f_{S_a}(O_t) \geq (1-O(\eps))(f(O_t) - f(S_a))  \quad \Rightarrow \quad f(S_a) \geq (1/2 - O(\eps))\OPT.
	\]
	
	For the amortized query complexity, by Eq.~\eqref{eq:branch_size}, the total number of the branch is bounded by $k^{\tilde{O}(1/\eps)}$, and by Lemma~\ref{lem:matroid-approx}, the amortized query per branch is $O(L) = O(\eps^{-1}\log(k/\eps))$. Hence,the total amortized number of query is at most $k^{\tilde{O}(1/\eps)}$.\qedhere
	%For each branch $a \in \mathcal{A}$, the \textsc{insert} procedure requies $2$ queries per iteration. For the \textsc{revoke} procedure, each time it is called, it scans all existing elements.
	%Since it is called only after a new element is added to $S_a$, and $|S_a| \leq k$, we know the {\sc revoke} procedure is called at most $k$ times.
	%Hence, the amortized number of query is $O(k)$ for each branch and the total amortized number of query is at msot $k^{\tilde{O}(\eps^{-1})}$.\qedhere
\end{proof}

\end{document}